\documentclass[12pt,a4paper]{article}

\usepackage[latin2]{inputenc}
\usepackage{amssymb}
\usepackage{amsmath}
\usepackage[mathscr]{eucal}
\usepackage{amsthm}
\usepackage{array}
\usepackage{color}

\swapnumbers
\theoremstyle{definition} 
\newtheorem{defin}{Definition}[section]
\newtheorem{thm}[defin]{Theorem}
\newtheorem{rem}[defin]{Remark}
\newtheorem{ex}[defin]{Example}
\newtheorem{example}[defin]{Example}
\newtheorem{cor}[defin]{Corollary}

\newtheorem{lemma}[defin]{Lemma}
\newtheorem{prop}[defin]{Proposition}

\def\vfi{\varphi}
\def\hil{{\mathcal H}}
\def\kil{{\mathcal K}}
\def\A{{\mathcal A}}
\def\B{{\mathcal B}}
\def\C{{\mathcal C}}

\def\X{{\mathcal X}}
\def\Y{{\mathcal Y}}
\def\S{{\mathcal S}}
\def\M{\mathcal{M}}
\def\P{\mathcal{P}}

\def\half{\frac{1}{2}}
\def\iff{\Longleftrightarrow}
\def\imp{\Longrightarrow}
\def\ep{\varepsilon}
\def\N{\mathbb{N}}
\def\bN{\mathbb{N}}
\def\iC{\mathbb{C}}
\def\R{\mathbb{R}}

\def\bz{\left(}
\def\jz{\right)}

\def\inv{^{-1}}
\def\kii{\emph}
\def\kiii{}

\def\egy{\mathbf 1}

\def\E{\mathcal{E}}

\def\map{\Phi}
\def\mapp{\Psi}

\def\bR{\mathbb{R}}
\def\bC{\mathbb{C}}
\def\of{\omega(f)}
\def\trans{^{\mathrm{T}}}

\newcommand{\ki}{\emph}

\newcommand{\s}{\mbox{ }}
\newcommand{\ds}{\mbox{ }\mbox{ }}
\newcommand{\norm}[1]{\left\| #1\right\|}
\newcommand{\hnorm}[1]{\left\| #1\right\|_{\mathrm{HS}}}
\newcommand{\snorm}[1]{\norm{#1}_{S}}
\newcommand{\inner}[2]{\langle #1 , #2\rangle}
\newcommand{\hinner}[2]{\langle #1 , #2\rangle_{\mathrm{HS}}}

\newcommand{\diad}[2]{|#1\rangle\langle #2|}
\newcommand{\pr}[1]{\diad{#1}{#1}}

\newcommand{\sr}[2]{S( #1\| #2)}
\newcommand{\rsr}[3]{S_{#3}( #1\| #2)}
\newcommand{\hdist}[3]{H_{#3}( #1\| #2)}
\newcommand{\chdist}[2]{C( #1\| #2)}

\newcommand{\derright}[1]{\partial^{+} #1}
\newcommand{\modop}[2]{L_{#1}R_{#2\inv}}
\newcommand{\modopp}[2]{\Delta\bz #1/#2\jz}
\newcommand{\fix}[1]{\ker\bz\id-#1\jz}

\newcommand{\tdist}[3]{T_{#3}\bz #1\,||\,#2\jz}
\newcommand{\dist}[2]{D\bz #1\,\|\,#2\jz}
\newcommand{\ns}[3]{\bz #1:#2\jz_{#3}}
\newcommand{\psif}[3]{\psi\bz #3|#1\|#2\jz}
\newcommand{\psift}[3]{\tilde\psi\bz #3|#1\|#2\jz}

\DeclareMathOperator{\id}{id}
\DeclareMathOperator{\Tr}{Tr}
\DeclareMathOperator{\supp}{supp}
\DeclareMathOperator{\spect}{spec}

\DeclareMathOperator{\sgn}{sgn}
\DeclareMathOperator{\ran}{ran}

\DeclareMathOperator{\co}{\overline{co}}

\def\<{\langle}
\def\>{\rangle}

\begin{document}

\centerline{\LARGE {\bf Quantum $f$-divergences and error correction}}
\bigskip

\bigskip
\bigskip
\centerline{\Large
Fumio Hiai$^{1,}$\footnote{E-mail: hiai@math.is.tohoku.ac.jp},
Mil\'an Mosonyi$^{2,3,}$\footnote{E-mail: milan.mosonyi@gmail.com},
D\'enes Petz$^{3,}$\footnote{E-mail: petz@math.bme.hu} and
C\'edric B\'eny$^{2,}$\footnote{E-mail: cedric.beny@gmail.com}
}
	
\medskip
\begin{center}
$^{1}$\,Graduate School of Information Sciences, Tohoku University \\
Aoba-ku, Sendai 980-8579, Japan
\end{center}
\begin{center}
$^2$\,Centre for Quantum Technologies,
National University of Singapore \\
3 Science Drive 2, 117543 Singapore
\end{center}
\begin{center}
$^3$\,Department of Analysis, Budapest University of Technology and Economics \\
Egry J\'ozsef u.~1., Budapest, 1111 Hungary
\end{center}
\bigskip

\begin{abstract}
Quantum $f$-divergences are a quantum generalization of the classical notion of $f$-divergences, 
and are a special case of Petz' quasi-entropies. Many well-known distinguishability measures 
of quantum states are given by, or derived from, $f$-divergences; special examples include the 
quantum relative entropy, the R\'enyi relative entropies, and the Chernoff and Hoeffding 
measures. Here we 
show that the quantum $f$-divergences are monotonic under substochastic maps whenever 
the defining function is operator convex. This extends and unifies all previously known 
monotonicity results for this class of distinguishability measures.
We also analyze the case where the monotonicity inequality holds with equality, and 
extend Petz' reversibility theorem for a large class of $f$-divergences and other 
distinguishability measures. We apply our findings 
to the problem of quantum error correction, and show that if a stochastic map 
preserves the pairwise distinguishability on a set of states, as measured by a suitable 
$f$-divergence, then its action can be reversed on that set by another stochastic map that 
can be constructed from the original one in a canonical way. 
We also provide an integral 
representation for operator convex functions on the positive half-line, which is the main 
ingredient in extending previously known results on the monotonicity inequality and the case 
of equality. We also consider some special cases where the convexity of $f$ is 
sufficient for the monotonicity, and obtain the inverse H\"older inequality for operators as
an application. The presentation is completely self-contained and requires only standard 
knowledge of matrix analysis.
\end{abstract}

\bigskip

\bigskip

\section{Introduction}

In the stochastic modeling of systems, the probabilities of the different outcomes of possible measurements performed on the system are given by a \ki{state}, which is a probability distribution in the case of classical systems and a density operator on the Hilbert space of the system in the quantum case. In applications, it is important to have a 
measure of how different two states are from each other and, as it turns out, such measures arise naturally in statistical problems like state discrimination. Probably the most important statistically motivated distance measure is the \ki{relative entropy}, given as
\begin{equation*}
\sr{\rho}{\sigma}:=\begin{cases}
\Tr\rho(\log\rho-\log\sigma),& \supp\rho\le\supp\sigma,\\
+\infty,& \text{otherwise},
\end{cases}
\end{equation*}
for two density operators $\rho,\sigma$ on a finite-dimensional Hilbert space. Its operational interpretation is given as the optimal exponential decay rate of an error probability in the state discrimination problem of Stein's lemma \cite{CT,HP,ON,Petzbook}, and 
it is the mother quantity for many other relevant notions in information theory, like the entropy, the conditional entropy, the mutual information and the channel capacity \cite{CT,Petzbook}. 

Undisputably the most relevant mathematical property of the relative entropy is its \ki{monotonicity under stochastic maps}, i.e., 
\begin{equation}\label{relentr mon}
\sr{\map(\rho)}{\map(\sigma)}\le \sr{\rho}{\sigma}
\end{equation}
for any two states $\rho,\sigma$ and quantum stochastic map $\map$ 
\cite{Petzbook}. 
Heuristically, \eqref{relentr mon} means that the distinguishability of two states cannot 
increase under further randomization. The monotonicity inequality yields immediately that
if the action of $\map$ can be reversed on the set $\{\rho,\sigma\}$, i.e., there 
exists another stochastic map $\mapp$ such that $\mapp(\map(\rho))=\rho$ and 
$\mapp(\map(\sigma))=\sigma$, then $\map$ preserves the relative entropy of $\rho$ and 
$\sigma$, i.e., inequality \eqref{relentr mon} holds with equality. A highly non-trivial 
observation, made by Petz in \cite{Petz2,Petz3}, is that the converse is also true:
If $\map$ preserves the relative entropy of $\rho$ and $\sigma$ then it is 
reversible on $\{\rho,\sigma\}$ and, moreover, the reverse map can be given in terms of $\map$ and $\sigma$ in a 
canonical way. This fact has found applications in the theory of quantum error correction 
\cite{JP,JP2,Ogawa}, the characterization of quantum Markov chains \cite{HJPW} and the description 
of states with zero quantum discord \cite{Datta,Hayashibook}, among many others.

Relative entropy has various generalizations, most notably R\'enyi's $\alpha$-relative 
entropies \cite{Renyi} that share similar monotonicity and convexity properties with the 
relative entropy and are also related to error exponents in binary state discrimination 
problems \cite{Csiszar2,MH}. 
A general approach to quantum relative entropies was developed by Petz in 1985 \cite{PD26}, who introduced the concept of quasi-entropies
(see also \cite{Petz} and Chapter 7 in \cite{OP}).
Let $\A:=\B(\iC^n)$ denote the algebra of 
linear operators on the finite-dimensional Hilbert space $\iC^n$
(which is essentially the algebra of $n\times n$ matrices with complex entries, and hence we also use the term matrix algebra).
For a positive $A\in\A$ and a strictly positive $B\in \A$, a general $K \in \A$ and 
a real-valued continuous function $f$ on $[0,+\infty)$, the \ki{quasi-entropy} is defined as
\begin{equation*}
S^K_f (A\|B):= \hinner{K B^{1/2}}{f(\modopp{A}{B})
(KB^{1/2})}
=\Tr  B^{1/2} K^*f(\modopp{A}{B})(KB^{1/2}),
\end{equation*}
where $\hinner{X}{Y}:=\Tr X^*Y,\,X,Y\in\A$, is the \ki{Hilbert-Schmidt inner product},
and
$\modopp{A}{B}:\A \to \A$ is the so-called \ki{relative modular operator} acting on $\A$ as 
$\modopp{A}{B}X:= AXB^{-1},\,X\in\A$.
The relative entropy can be obtained as a special case, corresponding to the function 
$f(x):=x\log x$ 
and $K:=I$, and R\'enyi's $\alpha$-relative entropies are related to the quasi-entropies
corresponding to $f(x):=x^\alpha$.

The two most important properties of the quasi-entropy are its monotonicity and joint convexity. Let $\map:\,\A_1 \to \A_2$ be a linear map between two matrix algebras $\A_1$ and $\A_2$, and
let $\map^*:\,\A_2 \to \A_1$ denote its dual with respect to the Hilbert-Schmidt inner products. A trace-preserving map $\map:\,\A_1\to\A_2$
is called a stochastic map if $\map^*$ satisfies the \ki{Schwarz inequality} $\Phi^*(Y^*)\Phi^*(Y)\le\Phi^*(Y^*Y),\,Y\in\A_2$.
The following monotonicity property of the quasi-entropies was shown in
\cite{PD26,Petz}:
Assume that $f$ is an operator monotone decreasing function on $[0,+\infty)$
with $f(0)\le 0$ and $\map:\,\A_1 \to \A_2$ is a stochastic map. Then
\begin{equation}\label{E:quasimonB}
S^K_f (\map(A)\|\map(B)) \le S^{\Phi^*(K)}_f (A\|B)
\end{equation}
holds for any $K \in \A_2$ and invertible positive operators $A,B\in\A_1$.
If $f$ is an operator convex function on $[0,+\infty)$, then $S^K_f (A,B)$
is jointly convex in the variables $A$ and $B$ \cite{OP,PD26,Petz}, i.e.,
\begin{equation*}
S^K_f \big(\sum\nolimits_ip_iA_i\big\|\sum\nolimits_ip_iB_i\big)\le
\sum\nolimits_ip_i S^K_f (A_i\|B_i)
\end{equation*}
for any finite set of positive invertible operators $A_i,B_i\in\A$ and probability weights $\{p_i\}$.

Quasi-entropy is a quantum generalization of the $f$-divergence of classical probability 
distributions, introduced independently by Csisz\'ar \cite{Csi}
and Ali and Silvey \cite{AS},
which is a widely used concept in classical information theory and statistics \cite{LV2,LV}.
This motivates the terminology ``quantum $f$-divergence'',
which we will use in this paper for the quasi-entropies with $K=I$. 
Actually, our notion of $f$-divergence is also a slight generalization of the quasi-entropy 
in the sense that we extend it to cases where the second operator is not invertible.
This extension is the same as in the classical setting, and was already considered in the 
quantum setting, e.g., in \cite{TCR}. We give the precise definition of the quantum 
$f$-divergences in Section 
\ref{sec:basic properties}, where we also give some of their basic properties, and prove 
that they are continuous in their second variable; the latter seems to be a new result. In Section 
\ref{sec:positive maps} we collect various technical statements on positive maps, which are necessary 
for the succeeding sections. In particular, we introduce a generalized notion of Schwarz 
maps, and investigate the properties of this class of positive maps. 

The monotonicity $\rsr{\map(A)}{\map(B)}{f}\le\rsr{A}{B}{f}$ of the $f$-divergences was proved  in \cite{Petz} for the case where $f$ is operator monotone decreasing and $\map$ is 
a stochastic map, and where $f$ is operator convex and $\map$ is the restriction 
onto a subalgebra; in both cases $B$ was assumed to be invertible. This was extended in
\cite{LR} to the case where $f$ is operator convex, $\map$ is stochastic and both $A$ 
and $B$ are invertible, using an integral representation of operator convex functions on $(0,+\infty)$, and in 
\cite{TCR} to the case where $f$ is operator convex and $\map$ is a completely positive 
trace-preserving map, without assuming the invertibility of $A$ or $B$, using the 
monotonicity under restriction onto a subalgebra and Lindblad's representation of completely 
positive maps. 
In Section \ref{sec:monotonicity}
we give a common generalization of these results by proving the monotonicity relation 
for the case where $f$ is operator convex, $\map$ is a substochastic map which preserves the trace of $B$, and both $A$ 
and $B$ are arbitrary positive semidefinite operators. This is based on the continuity result proved in Section \ref{sec:basic properties} and an integral representation of operator convex 
functions on $[0,+\infty)$ that we provide in Section \ref{sec:integral representation}.
To the best of our knowledge, this representation is new, and might be interesting in itself.

It has been known \cite{JP,JP2,Petz2} for the relative entropy and some R\'enyi relative 
entropies that the 
monotonicity inequality for two operators and a $2$-positive trace-preserving map holds with equality 
if and only if the action of the map can be reversed on the given operators.
We extend this result to a large class of $f$-divergences
in Section \ref{sec:equality}, where we show that if a stochastic map $\map$
preserves the $f$-divergence of two operators $A$ and $B$ corresponding to
an operator convex function which is not a polynomial then 
it preserves a certain set of ``primitive'' $f$-divergences, corresponding to the functions
$\vfi_t(x):=-x/(x+t)$ for a set $T$ of $t$'s. Moreover, if this set has large enough 
cardinality (depending on $A,B$ and $\map$) and $\Phi$ is $2$-positive then there exists 
another stochastic map
$\mapp$ reversing the action of $\map$ on $\{A,B\}$, i.e., such that $\mapp(\map(A))=A$
and $\mapp(\map(B))=B$.
In Section \ref{sec:Chernoff}, we formulate 
equivalent conditions for reversibility in terms of the preservation of measures relevant 
to state discrimination, namely the Chernoff distance and the Hoeffding distances, and we 
also show that these measures cannot be represented as $f$-divergences.
In Section \ref{sec:error correction} we apply the above results on reversibility to 
the problem of quantum error 
correction, and give equivalent conditions for the reversibility of a quantum operation on a 
set of states in terms of the preservation of pairwise 
$f$-divergences, Chernoff and Hoeffding distances, and many-copy trace-norm distances.
Related to the latter, we also analyze the connection with the recent results of 
\cite{BNPV}, where reversibility was obtained from the preservation of single-copy 
trace-norm distances under some extra technical conditions, and show that the approach of 
\cite{BNPV} is unlikely to be recovered from our analysis of the preservation of $f$-divergences,
as the quantum trace-norm distances cannot be represented as $f$-divergences. This is in contrast with the classical case, and is another manifestation of the significantly more complicated structure of quantum states and their distinguishability measures, as compared to their classical counterparts.

In our analysis of the monotonicity inequality 
$\rsr{\map(A)}{\map(B)}{f}\le\rsr{A}{B}{f}$  
and the case of the equality, it is essential that $f$ is operator convex; it is an open 
question though whether this is actually necessary. In Appendix \ref{sec:classical} we 
consider some situations where convexity of $f$ is sufficient; this includes the case of 
commuting 
operators, which is essentially a reformulation of the classical case, and the monotonicity 
under the pinching operation defined by the reference operator $B$, which was first proved 
in \cite{Hayashibook} for the R\'enyi relative entropies. Although both of these cases are 
very special and their proofs are considerably simpler than the general case, they are 
important for applications. As an illustration, we derive from these results the 
exponential version of the operator 
H\"older inequality and the inverse H\"older inequality, and analyse the case when they hold 
with equality.

\section{Quantum $f$-divergences: definition and basic properties}\label{sec:basic properties}
\setcounter{equation}{0}

Let $\A$ be a finite-dimensional $C^*$-algebra. {Unless otherwise stated,} we will always assume that $\A$ is a $C^*$-subalgebra of $\B(\hil)$ for some finite-dimensional Hilbert space
$\hil$, i.e., 
$\A$ is a subalgebra of $\B(\hil)$ that is closed under taking the adjoint of operators.
For simplicity, we also assume that the unit of $\A$ coincides with identity operator $I$ on $\hil$; if this is not the case, we can simply consider a smaller Hilbert space.
The Hilbert-Schmidt inner product on $\A$ is defined as
\begin{equation*}
\hinner{A}{B}:=\Tr A^*B,\ds\ds\ds A,B\in\A,
\end{equation*}
with induced norm $\hnorm{A}:=\sqrt{\Tr A^*A},\,A\in\A$. 

We will follow the convention that powers of a positive semidefinite operator are only taken on its support; in particular, if $0\le X\in\A$ then $X\inv$ denotes the generalized inverse of $X$ and $X^0$ is the projection onto the support of $X$. For a real $t\in\R$, 
$X^{it}$ is a unitary on $\supp X$ but not on the whole Hilbert space unless $X^0=I$.
We denote by $\log^*$ the extension of $\log$ to the domain $[0,+\infty)$, defined to be $0$ at $0$. With these conventions, we have $\frac{d}{dz}X^{z}\big\vert_{z=0}=\log^* X$.
We also set
\begin{equation*}
0\cdot \pm\infty:=0,\ds\ds\ds \log 0:=-\infty,\ds\ds\text{and}\ds\ds \log+\infty:=+\infty.
\end{equation*}

For a linear operator $A\in\A$, let $L_A,\,R_A\in\B(\A)$ denote the left and the right multiplications by $A$, respectively, defined as
\begin{equation*}
L_A:\,X\mapsto AX,\ds\ds\ds
R_A:\,X\mapsto XA,\ds\ds\ds X\in\A.
\end{equation*}
Left and right multiplications commute with each other, i.e., $L_AR_B=R_BL_A,\,A,B\in\A$.
If $A,B$ are positive elements in $\A$ with spectral decompositions
$A=\sum_{a\in\spect(A)} aP_a$ and $B=\sum_{b\in\spect(B)} bQ_b$ 
(where $\spect(X)$ denotes the spectrum of $X\in\A$) then the spectral decomposition of 
$L_AR_{B\inv}$ is given by $L_AR_{B\inv}=\sum_{a\in\spect(A)}\sum_{b\in\spect(B)} ab\inv L_{P_a}R_{Q_b}$, and for any function $f$ on $\{ab\inv\,:\,a\in\spect(A),\,b\in\spect(B)\}$,
we have
\begin{equation}\label{L_A R_B functional calculus}
f(L_AR_{B\inv})=\sum_{a\in\spect(A)}\sum_{b\in\spect(B)} f(ab\inv)L_{P_a}R_{Q_b}.
\end{equation}
(Note that we have $0\inv=0$ in the above formulas due to our convention.)

\begin{defin}\label{def:f-div}
Let $A$ and $B$ be positive semidefinite operators on $\hil$
and let $f:\,[0,+\infty)\to$ $\bR$ be a real-valued function on 
$[0,+\infty)$ such that $f$ is continuous on $(0,+\infty)$ and the limit 
\begin{equation*}
\of:=\lim_{x\to+\infty}\frac{f(x)}{x}
\end{equation*}
exists in $[-\infty,+\infty]$. The \ki{$f$-divergence} of $A$ with respect to $B$ is defined as
\begin{equation*}
\rsr{A}{B}{f}:=
\hinner{B^{1/2}}{f\bz \modop{A}{B}\jz B^{1/2}}
\end{equation*}
when $\supp A\le\supp B$. In the general case, we define
\begin{equation}\label{f-div def2}
\rsr{A}{B}{f}:=\lim_{\ep\searrow 0}\rsr{A}{B+\ep I}{f}.
\end{equation}
\end{defin}

\begin{prop}\label{prop:consistency}
The limit in \eqref{f-div def2} exists, and 
\begin{equation*}
\lim_{\ep\searrow 0}\rsr{A}{B+\ep I}{f}=
\hinner{B^{1/2}}{f\bz \modop{A}{B}\jz B^{1/2}}
+\of\Tr A(I-B^0).
\end{equation*}
In particular, Definition \ref{def:f-div} is consistent in the sense that if $\supp A\le\supp B$ then 
\begin{equation*}
\lim_{\ep\searrow 0}\rsr{A}{B+\ep I}{f}=\hinner{B^{1/2}}{f\bz \modop{A}{B}\jz B^{1/2}}.
\end{equation*}
\end{prop}
\begin{proof}
By \eqref{L_A R_B functional calculus}, we have 
$\rsr{A}{B+\ep I}{f}=\sum_{a\in\spect(A)}\sum_{b\in\spect(B)} (b+\ep)f(a/(b+\ep))\Tr P_a Q_b$, and
the assertion follows by a straightforward computation using that for any $a,b\ge 0$,
\begin{equation}\label{limit of mean}
\lim_{0<\tilde b\to b}\tilde bf(a/\tilde b)=
\begin{cases}
bf(a/b),& b>0,\\
a\of,& b=0.
\end{cases}\qedhere
\end{equation}
\end{proof}

\begin{cor}\label{cor:f-div explicit expression}
For $A,B$ and $f$ as in Definition \ref{def:f-div},
\begin{align}
\rsr{A}{B}{f}&=
\hinner{B^{1/2}}{f\bz \modop{A}{B}\jz B^{1/2}}
+\of\Tr A(I-B^0)\label{f-div explicit expression}\\
&=
f(0)\Tr B+
\hinner{B^{1/2}}{(f-f(0))\bz \modop{A}{B}\jz B^{1/2}}
+\of\Tr A(I-B^0)\label{f-div explicit expression2}\\
&=
\sum_{a\in\spect(A)}\bigg(\sum_{b\in\spect(B)\setminus\{0\}} bf(a/b)\Tr P_a Q_b +a\of\Tr P_aQ_0\bigg),\label{f-div explicit expression3}
\end{align}
and 
$\rsr{A}{B}{f}=\hinner{B^{1/2}}{f\bz \modop{A}{B}\jz B^{1/2}}$ if and only if 
$\supp A\le\supp B$ or $\lim_{x\to+\infty}\frac{f(x)}{x}=0$.
\end{cor}

\begin{rem}
Note that $\modop{A}{B}=\modopp{A}{B}$, given in the Introduction, and hence the  $f$-divergence is a special case of the quasi-entropy (with $K=I$) when $\supp A\le\supp B$ or 
$\lim_{x\to+\infty}f(x)/x=0$
\end{rem}

\begin{cor}\label{cor:scaling}
Let $A,A_1,A_2,B,B_1,B_2$ and $f$ be as in Definition \ref{def:f-div}. We have the following:
\begin{enumerate}
\item
For every $\lambda\in[0,+\infty)$, 
\begin{equation*}
\rsr{\lambda A}{\lambda B}{f}=\lambda\rsr{A}{B}{f}.
\end{equation*}
\item
If $A_1^0\vee B_1^0\perp A_2^0\vee B_2^0$ then 
\begin{equation*}
\rsr{A_1+A_2}{B_1+B_2}{f}=\rsr{A_1}{B_1}{f}+\rsr{A_2}{B_2}{f}.
\end{equation*}
\item\label{item:isometry invariance}
If $V:\,\hil\to\kil$ is a linear or anti-linear isometry then 
\begin{equation*}
\rsr{VAV^*}{VBV^*}{f}=\rsr{A}{B}{f}.
\end{equation*}
\item
If $x$ is a unit vector in some Hilbert space $\kil$ then 
\begin{equation*}
\rsr{A\otimes \pr{x}}{B\otimes \pr{x}}{f}=\rsr{A}{B}{f}.
\end{equation*}
\end{enumerate}
\end{cor}
\begin{proof}
Immediate from \eqref{f-div explicit expression3}.
\end{proof}

\begin{rem}\label{rem:transposition}
Note that if $V$ is an anti-linear isometry then there exists a linear isometry $\tilde V$ and a basis $\B$ such that $VAV^*=\tilde V A\trans\tilde V^*,\,A\in\A_+$, where the transposition is in the basis $\B$. Hence, \ref{item:isometry invariance} of Corollary \ref{cor:scaling} is equivalent to the $f$-divergences being invariant under conjugation by an isometry and transposition in an arbitrary basis.
\end{rem}

\begin{ex}\label{ex: f-div Renyi}
Let $f_\alpha(x):=x^{\alpha}$ for $\alpha>0,\,x\ge 0$. For $\alpha=0$, we define $f_0(x):=1,\,x>0,\,f_0(0):=0$.
A straightforward computation yields that 
\begin{equation}\label{Renyi f-div}
\rsr{A}{B}{f_\alpha}=\Tr A^\alpha B^{1-\alpha}+\bz\lim_{x\to+\infty}x^{\alpha-1}\jz\Tr A(I-B^0)
\end{equation}
for any $A,B\in\A_+$, and hence, if $0\le\alpha<1$ then
\begin{equation*}
\rsr{A}{B}{f_\alpha}=\Tr A^\alpha B^{1-\alpha},
\end{equation*}
whereas for $\alpha>1$ we have 
\begin{equation*}
\rsr{A}{B}{f_\alpha}=\begin{cases}
\Tr A^\alpha B^{1-\alpha},& \supp A\le\supp B,\\
+\infty,& \text{otherwise}.
\end{cases}
\end{equation*}
The \ki{R\'enyi relative entropy} of $A$ and $B$ with parameter $\alpha\in [0,+\infty)\setminus\{1\}$ is defined as
\begin{equation*}
\rsr{A}{B}{\alpha}:=\frac{1}{\alpha-1}\log\rsr{A}{B}{f_\alpha}=
\begin{cases}
\frac{1}{\alpha-1}\log\Tr A^\alpha B^{1-\alpha},& \supp A\le\supp B\text{ or }\alpha<1,\\
+\infty,& \text{otherwise}.
\end{cases}
\end{equation*}

The choice $f(x):=x\log x$ yields the \ki{relative entropy} of $A$ and $B$,
\begin{equation*}
\rsr{A}{B}{f}=\begin{cases}
\Tr A\bz\log^* A-\log^* B\jz,&\supp A\le\supp B,\\
+\infty,& \text{otherwise},
\end{cases}
\end{equation*}
where the second case follows from $\lim_{x\to+\infty}\frac{x\log x}{x}=+\infty$.
\end{ex}

The following shows that the representing function for an $f$-divergence is unique:
\begin{prop}
Assume that a function $D:\,\A_+\times\A_+\to\R$ can be represented as an $f$-divergence. 
Then the representing function $f$ is uniquely determined by the restriction of $D$ onto the 
trivial subalgebra as
\begin{equation}\label{expression for f}
f(x)=D(xI\|I)/\dim\hil,\ds\ds\ds x\in[0,+\infty).
\end{equation}
In particular, for every $D:\,\A_+\times\A_+\to\R$ there is at most one function $f$ such that
$D=S_f$ holds.
\end{prop}
\begin{proof}
Formula \eqref{expression for f} is obvious from \eqref{f-div explicit expression3}, and the rest follows immediately.
\end{proof}

In most of the applications, $f$-divergences are used to compare probability 
distributions 
in the classical, and density operators in the quantum case, and one might wonder whether 
there is more freedom in representing a measure as an $f$-divergence if we are only 
interested in density operators instead of general positive semidefinite operators. The 
following simple argument shows that if a measure can be represented as an $f$-divergence on 
quantum states then its values are uniquely determined by its values on classical 
probability distributions.

Given density operators $\rho$ and $\sigma$ with spectral decomposition
$\rho=\sum_{a\in\spect(\rho)}aP_a$ and $\sigma=\sum_{b\in\spect(\sigma)}bQ_b$, 
we can define classical probability density functions $\ns{\rho}{\sigma}{1}$ and $\ns{\rho}{\sigma}{2}$ 
on $\spect(\rho)\times\spect(\sigma)$ as
\begin{equation*}
\ns{\rho}{\sigma}{1}(a,b):=a\Tr P_aQ_b,\ds\ds\ds
\ns{\rho}{\sigma}{2}(a,b):=b\Tr P_aQ_b.
\end{equation*}
This kind of mapping from pairs of quantum states to pairs of classical states was introduced in \cite{NSz}, and is one of the main ingredients in the proofs of the quantum Chernoff and Hoeffding bound theorems.
\begin{lemma}\label{lemma:NSz}
For any two density operators $\rho,\sigma$ and any function $f$ as in Definition \ref{def:f-div}, 
\begin{equation*}
\rsr{\rho}{\sigma}{f}=\rsr{\ns{\rho}{\sigma}{1}}{\ns{\rho}{\sigma}{2}}{f}.
\end{equation*}
\end{lemma}
\begin{proof}
It is immediate from \eqref{f-div explicit expression3}.
\end{proof}
\begin{cor}\label{cor:uniqueness}
Let $f$ and $g$ be functions as in Definition \ref{def:f-div}. If $S_f$ and $S_g$ coincide on classical probability 
distributions then they coincide on quantum states as well. 
\end{cor}
\begin{proof}
Obvious from Lemma \ref{lemma:NSz}.
\end{proof}

\begin{ex}
For two density operators $\rho,\sigma$, their \ki{quantum fidelity} is given by 
$F(\rho,\sigma):=\Tr\sqrt{\rho^{1/2}\sigma\rho^{1/2}}$ \cite{Uhlmann}.
For classical probability distributions, the fidelity coincides with 
$S_{f_{1/2}}$, where $f_{1/2}(x)=x^{1/2}$. 
If the fidelity could be represented as an $f$-divergence for quantum states then the 
representing function should be $f_{1/2}$, due to Corollary \ref{cor:uniqueness}.
However, the corresponding quantum $f$-divergence is 
$\rsr{\rho}{\sigma}{f_{1/2}}=\Tr\rho^{1/2}\sigma^{1/2}$, which is not equal to 
$F(\rho,\sigma)$ in general. This shows that the 
fidelity of quantum states cannot be represented as an $f$-divergence.
\end{ex}

In Sections \ref{sec:Chernoff} and \ref{sec:error correction} we give similar non-represantability results for measures related
to state discrimination on the state spaces of individual algebras.
\medskip

Our last proposition in this section says that when $\omega(f)$ is finite, the $f$-divergence
is continuous in the second variable.

\begin{prop}\label{prop:continuity}
Assume that $\omega(f)$ is finite. Let $A,B,B_k\in \A$ with $A,B,B_k\ge0$ for all $k\in\bN$,
and assume that $\lim_{k\to\infty}B_k=B$. Then
\begin{equation*}
\lim_{k\to\infty}\rsr{A}{B_k}{f}=\rsr{A}{B}{f}.
\end{equation*}
\end{prop}
\begin{proof}
First, by the assumption on $\omega(f)$ and Corollary \ref{cor:f-div explicit expression}, note
that $S(A\|B_k)$ is finite for any $k$. Then by the definition \eqref{f-div def2}, we can choose 
a sequence $\ep_k>0,\,k\in\N$, such that $\lim_{k\to\infty}\ep_k=0$, and for all $k\in\N$,
\begin{equation*}
\rsr{A}{B_k+\ep_kI}{f}-{1\over k}<\rsr{A}{B_k}{f}<\rsr{A}{B_k+\ep_kI}{f}+{1\over k}.
\end{equation*}
Let $\tilde B_k:=B_k+\ep_k I$, which is strictly positive for any $k\in\N$. Obviously,
$\lim_{k\to\infty}\tilde B_k=B$, and the assertion will follow if we can show that 
\begin{equation}\label{lim-S_f}
\lim_{k\to\infty}\rsr{A}{\tilde B_k}{f}=\rsr{A}{B}{f}.
\end{equation}

Let $A=\sum_{a\in\spect(A)}aP_a$, $B=\sum_{b\in\spect(B)}bQ_b$ and $\tilde B_k=\sum_{c\in\spect(\tilde B_k)}cQ_c^{(k)}$ be the spectral
decompositions of the respective operators. Then
\begin{equation*}
S_f(A\|\tilde B_k)=\sum_{a\in\spect(A)}\sum_{c\in\spect(\tilde B_k)}f(a/c)c\Tr P_aQ_c^{(k)}.
\end{equation*}
From the continuity of the eigenvalues and the spectral projections when $\tilde B_k\to B$, we see
that, for every $\delta>0$ with $\delta<{1\over2}\min\{|b-b'|:b,b'\in\spect(B),b\ne b'\}$,
if $k$ is sufficiently large, then we have
\begin{equation*}
\spect(\tilde B_k)\subset\bigcup_{b\in\spect(B)}(b-\delta,b+\delta)
\ds\ds\ds\ds \mbox{(disjoint union)}
\end{equation*}
and moreover,
\begin{equation*}
\hat Q_b^{(k)}:=\sum_{c\in\spect(\tilde B_k)\atop c\in(b-\delta,b+\delta)}Q_c^{(k)}\longrightarrow Q_b
\quad\mbox{as $k\to+\infty$, for all $b\in\spect(B)$}.
\end{equation*}

Due to \eqref{limit of mean}, for every $\ep>0$ there exists a $\delta>0$ as above such that,
for $a\in\spect(A)$, $b\in\spect(B)$ and $c\in\spect(\tilde B_k)$,
\begin{align*}
&|cf(a/c)-bf(a/b)|<\ep\quad
\mbox{if $b>0$ and $c\in(b-\delta,b+\delta)$}, \\
&|cf(a/c)-a\omega(f)|<\ep\quad\ds
\mbox{if $c\in(0,\delta)$}.
\end{align*}
Hence, if $k$ is sufficiently large, then we have by \eqref{f-div explicit expression3}
\begin{align*}
&|\rsr{A}{\tilde B_k}{f}-\rsr{A}{B}{f}| \\
&\quad\le\sum_{a\in\spect(A)}\sum_{b\in\spect(B)\setminus\{0\}}
\left|\sum_{c\in\spect(\tilde B_k)\atop c\in(b-\delta,b+\delta)}
cf(a/c)\Tr P_aQ_c^{(k)}-bf(a/b)\Tr P_aQ_b\right| \\
&\quad\qquad+\sum_{a\in\spect(A)}\left|\sum_{c\in\spect(\tilde B_k)\atop c\in(0,\delta)}
cf(a/c)\Tr P_aQ_c^{(k)}-a\omega(f)\Tr P_aQ_0\right|\\
&\quad\le\sum_{a\in\spect(A)}\sum_{b\in\spect(B)\setminus\{0\}}
\left\{\sum_{c\in\spect(\tilde B_k)\atop c\in(b-\delta,b+\delta)}
|cf(a/c)-bf(a/b)|\Tr P_aQ_c^{(k)}+
\left|bf(a/b)\Tr P_a\left(\hat Q_b^{(k)}-Q_b\right)\right|\right\} \\
&\quad\qquad+\sum_{a\in\spect(A)}\left\{\sum_{c\in\spect(\tilde B_k)\atop c\in(0,\delta)}
|cf(a/c)-a\omega(f)|\Tr P_aQ_c^{(k)}+\left|a\omega(f)\Tr P_a
\left(\hat Q_0^{(k)}-Q_0\right)\right|\right\} \\
&\quad\le\ep\Tr I
+\sum_{a\in\spect(A)}\sum_{b\in\spect(B)\setminus\{0\}}|bf(a/b)|
\norm{\hat Q_b^{(k)}-Q_b}_1
+\sum_{a\in\spect(A)}|a\omega(f)|
\norm{\hat Q_0^{(k)}-Q_0}_1.
\end{align*}
This implies that
\begin{equation*}
\limsup_{k\to\infty}|S_f(A\|\tilde B_k)-S_f(A\|B)|\le\ep\Tr I
\end{equation*}
for every $\ep>0$, and so \eqref{lim-S_f} follows.
\end{proof}

\begin{rem}
The finiteness assumption on $\omega(f)$ is essential in the above proposition. Indeed, take
$f$ such that $\omega(f)=+\infty$ or $-\infty$. Let $A=B=|x\>\<x|$ be a rank $1$ projection,
and $B_k=|x_k\>\<x_k|$ where $\|x_k-x\|\to0$ and $x_k$ is not proportional to $x$ for any $k$.
Then $S_f(A\|B)=f(1)$ while $S_f(A\|B_k)=+\infty$ or $-\infty$, respectively. Note also that
$S_f(A\|B)$ is not continuous in the first variable even when $\omega(f)$ is finite, unless
$f$ is assumed to be continuous at $0$.
\end{rem}

\section{Preliminaries on positive maps}\label{sec:positive maps}
\setcounter{equation}{0}

Let $\A_i\subset\B(\hil_i)$ be finite-dimensional $C^*$-algebras with unit $I_i$ for $i=1,2$.
For a subset $\B\subset\A_i$, we will denote the set of positive elements in $\B$ by $\B_+$; in particular,
$\A_{i,+}$ denotes the set of positive elements in $\A_i$.
For a linear map $\map:\,\A_1\to\A_2$, we denote its adjoint with respect to the Hilbert-Schmidt inner products by $\map^*$. Note that $\map$ and $\map^*$ uniquely determine each other and, moreover, $\map$
is positive/$n$-positive/completely positive if and only if $\map^*$ is positive/$n$-positive/completely positive, and $\map$ is trace-preserving/trace non-increasing if and only if $\map^*$ is unital/sub-unital. 

For given $B\in\A_{1,+}$ and $\map:\,\A_1\to\A_2$, we define
$\map_B:\,\A_1\to\A_2$ and $\map^*_B:\,\A_2\to\A_1$ as
\begin{align}
\map_B(X)&:= \map(B)^{-1/2}\map(B^{1/2}XB^{1/2})\map(B)^{-1/2},\ds\ds\ds X\in\A_1,\label{def:PhiB}\\
\map^*_B(Y)&:= B^{1/2}\map^*\bz\map(B)^{-1/2}Y\map(B)^{-1/2}\jz B^{1/2},\ds\ds\ds Y\in\A_2.\label{def:PhistarB}
\end{align}
With these notations, we have $(\map_B)^*=\map_B^*$ and $(\map_B^*)^*=\map_B$.

For a normal operator $X\in\A_1$, let $P_{\{1\}}(X)$ denote the spectral projection of $X$ onto its fixed-point set. 
Note that if $B\in\A_{1,+}$ then $B^0$ is a projection in $\A_1$ and hence $B^0\A_1 B^0$ is a $C^*$-algebra with unit $B^0$.

 \begin{lemma}\label{lemma:supports}
 If $\map:\,\A_1\to\A_2$ is a positive map and 
 $A,B$ are positive elements in $\A_1$ such that 
 $A^0=B^0$ then $\map(A)^0=\map(B)^0$. In particular, $\map(B)^0=\map(B^0)^0$ for any positive $B\in\A_1$.
 \end{lemma}
 \begin{proof}
 The assumption $A^0=B^0$ is equivalent to the existence of strictly positive numbers $\alpha,\beta$ such that $\alpha A\le B\le \beta A$, which yields
 $\alpha\map(A)\le \map(B)\le \beta\map(A)$
 and hence
 $\map(A)^0=\map(B)^0$.
 \end{proof}

\begin{lemma}\label{lemma:support inequality}
Let $B\in\A_{1,+}$ and let $\map:\,\A_1\to\A_2$ be a 
positive map such that $\map^*(\map(B)^0)\le I_1$ (in particular, this is the case if $\map$ 
is trace non-increasing). Then
\begin{equation*}
\Tr\map(B)\le\Tr B, 
\end{equation*}
and the following are equivalent:
\begin{enumerate}
\item\label{item:trace preserving}
$\Tr\map(B)=\Tr B$.
\item\label{item:functional calculus}
For any function $f$ on $\spect(B)$ such that $f(0)=0$ if $0\in\spect(B)$, we have
\begin{equation*}
f(B)\map^*(\map(B)^0)=\map^*(\map(B)^0)f(B)=f(B).
\end{equation*} 
\item\label{item:projection dominance}
$B^0\le P_{\{1\}}\bz \map^*(\map(B)^0)\jz$.
\item\label{item:trace preserving on algebras}
$\map$ is trace-preserving on $B^0\A_1B^0$.
(In particular, if $A\in\A_{1,+}$ is such that $A^0\le B^0$ then 
$\Tr\map(A)=\Tr A$.)
\item\label{item:PhiB}
For the map $\map_B^*$ given in \eqref{def:PhistarB}, we have
\begin{equation*}
\map_B^*(\map(B))=B.
\end{equation*}
\end{enumerate}
\end{lemma}
\begin{proof}
By assumption, $\map^*(\map(B)^0)\le I_1$ and hence,
\begin{equation*}
0\le\Tr(I_1-\map^*(\map(B)^0))B=\Tr B-\Tr\map^*(\map(B)^0)B=
 \Tr B-\Tr\map(B)^0\map(B)=\Tr B-\Tr\map(B).
\end{equation*}
If $\Tr \map(B)=\Tr B$ then $(I_1-\map^*(\map(B)^0))B=0$, i.e., 
$B=\map^*(\map(B)^0)B$, so we get $B^n=\map^*(\map(B)^0)B^n,\,n\in\N$, 
which yields \ref{item:functional calculus}.
Hence, the implication \ref{item:trace preserving}$\imp$\ref{item:functional calculus} holds.
If \ref{item:functional calculus} holds then we have 
$B^0=\map^*(\map(B)^0)B^0$ and hence,
for any $x\in\hil$ such that $B^0x=x$, we have $x=B^0x=\map^*(\map(B)^0)B^0x=\map^*(\map(B)^0)x$, or equivalently, $x\in\ran P_{\{1\}}\bz \map^*(\map(B)^0)\jz$. This yields 
\ref{item:projection dominance}, and the converse direction 
\ref{item:projection dominance}$\imp$\ref{item:functional calculus} is obvious. Assume now that 
\ref{item:functional calculus} holds. If $X\in B^0\A_1B^0$, then
$XB^0=B^0X=X$, and
\begin{equation*}
\Tr\map(X)=\Tr\map(X)\map(B)^0=\Tr X\map^*(\map(B)^0)=\Tr XB^0\map^*(\map(B)^0)=\Tr XB^0=\Tr X,
\end{equation*}
showing \ref{item:trace preserving on algebras}.
The implication \ref{item:trace preserving on algebras}$\imp$\ref{item:trace preserving} is obvious.

Assume that \ref{item:functional calculus} holds. Then $\map_B^*(\map(B))=B^{1/2}\map^*\bz\map(B)^0\jz B^{1/2}=B$, showing \ref{item:PhiB}. On the other hand, if \ref{item:PhiB} holds then
$B^{1/2}\map^*\bz\map(B)^0\jz B^{1/2}=B$, and hence 
$0=B^{1/2}(I_1-\map^*\bz\map(B)^0\jz)B^{1/2}$. Since 
$I_1-\map^*\bz\map(B)^0\jz\ge 0$, we obtain
$B^{1/2}(I_1-\map^*\bz\map(B)^0\jz)^{1/2}=0$, which in turn yields
$B=B\map^*\bz\map(B)^0\jz$. From this \ref{item:functional calculus} follows as above.
\end{proof}

\begin{cor}\label{cor:trace preserving for A}
Let $A,B\in\A_{1,+}$, and let $\map:\,\A_1\to\A_2$ be a trace non-increasing positive map.
Then $\map$ is  trace-preserving on $(A+B)^0\A_1(A+B)^0$ if and only if
\begin{equation*}
\Tr\map(A)=\Tr A\ds\ds\text{and}\ds\ds \Tr\map(B)=\Tr B.
\end{equation*}
\end{cor}
\begin{proof}
Obvious from Lemma \ref{lemma:support inequality}.
\end{proof}

\begin{cor}\label{cor:0 Renyi}
Let $A,B\in\A_{1,+}$ and let $\map:\,\A_1\to\A_2$ be a trace non-increasing positive map such that $\Tr\map(A)=\Tr A$. Then 
\begin{equation*}
\Tr\map(B)\map(A)^0\ge\Tr BA^0\ds\ds\text{and}\ds\ds
\Tr\map(B)(I_2-\map(A)^0)\le\Tr B(I_1-A^0).
\end{equation*}
Note that the first inequality means the monotonicity of the R\'enyi $0$-relative entropy
$\rsr{A}{B}{0}\ge\rsr{\map(A)}{\map(B)}{0}$ under the given conditions.
\end{cor}
\begin{proof}
Due to Lemma \ref{lemma:support inequality}, the assumptions yield that
$A^0\le P_{\{1\}}\bz \map^*(\map(A)^0)\jz\le\map^*(\map(A)^0)$, and hence
$0\le\Tr B(\map^*(\map(A)^0)-A^0)=\Tr\map(B)\map(A)^0-\Tr BA^0$. The second inequality follows by taking into account that $\Tr\map(B)\le\Tr B$.
\end{proof}
\medskip

The following lemma yields the monotonicity of the R\'enyi $2$-relative entropies, and is 
needed to prove the monotonicity of general $f$-divergences. The statement and its proof 
can be obtained by following the proofs of Theorem 1.3.3, Theorem 2.3.2 (Kadison's inequality) and Proposition 2.7.3 in \cite{Bhatia2} using the weaker conditions given here. For readers' convenience, we include a self-contained proof here.
\begin{lemma}\label{lemma:2 Renyi monotonicity}
Let $A,B\in\A_{1,+}$ and $\map:\,\A_1\to\A_2$ be a positive map. Then
\begin{equation}\label{Bhatia inequality0}
\map(B^0AB^0)\map(B)\inv\map(B^0AB^0)\le \map(B^0AB\inv AB^0).
\end{equation}
In particular, if $A^0\le B^0$ then
\begin{equation}\label{Bhatia inequality}
\map(A)\map(B)\inv\map(A)\le \map(AB\inv A).
\end{equation}
If, moreover, $\map$ is also trace non-increasing then
\begin{equation}\label{2 Renyi monotonicity}
\rsr{\map(A)}{\map(B)}{f_2}=\Tr \map(A)^2\map(B)\inv\le \Tr A^2B\inv=\rsr{A}{B}{f_2}.
\end{equation}
\end{lemma}
\begin{proof}
Define $\Psi:\,\A_1\to\A_2$ as $\Psi(X):=\map(B^{1/2}XB^{1/2}),\,X\in\A_1$. Let
$X:=B^{-1/2}AB^{-1/2}$ and let $X=\sum_{x\in\sigma(X)}xP_x$ be its spectral decomposition. Then
\begin{equation*}
\hat X:=\begin{bmatrix}
\Psi(X^2) & \Psi(X)\\ \Psi(X) & \Psi(I_1)
\end{bmatrix}
=
\sum_{x\in\sigma(X)}\begin{bmatrix}
x^2 & x\\ x & 1
\end{bmatrix}\otimes \Psi(P_x)\ge 0,
\end{equation*}
and hence we have
\begin{equation*}
0\le\hat Y\hat X\hat Y^*=
\begin{bmatrix}
\Psi(X^2)- \Psi(X)\Psi(I_1)\inv\Psi(X) & \Psi(X)(I_2-\Psi(I)^0)\\ (I_2-\Psi(I_1)^0)\Psi(X) & \Psi(I_1)
\end{bmatrix},
\end{equation*}
where
$$
\hat Y:=
\begin{bmatrix}
I_2 & -\Psi(X)\Psi(I_1)\inv\\ 0 & I_2
\end{bmatrix}.
$$
Hence $\Psi(X^2)\ge \Psi(X)\Psi(I_1)\inv\Psi(X)$, which is exactly \eqref{Bhatia inequality0}.
The inequalities in \eqref{Bhatia inequality} and \eqref{2 Renyi monotonicity} follow immediately.
\end{proof}
\bigskip

We say that a map $\map:\,\A_1\to\A_2$ 
is a \ki{Schwarz map} if 
\begin{equation*}
\snorm{\map}:=\inf\{c\in[0,+\infty)\,:\,\map(X)^*\map(X)\le c\map(X^*X),\,X\in\A\}<+\infty.
\end{equation*}
Obviously, if $\map$ is a Schwarz map then $\map$ is positive, 
and we have $\norm{\map}=\|\map(I_1)\|\le\snorm{\map}$.
(Note that $\norm{\map}=\norm{\map(I_1)}$ is true for any positive map $\map$ \cite[Corollary 2.3.8]{Bhatia2}).
We say that $\map$ is a \ki{Schwarz contraction} if it is a Schwarz map with 
$\snorm{\map}\le 1$. 
A Schwarz contraction $\map$ is also a contraction, due to $\norm{\map}\le\snorm{\map}$. 
Note that a positive map $\map$ is a contraction
if and only if it is subunital, which is equivalent to $\map^*$ being trace non-increasing.
We say that a map $\map$ between two finite-dimensional $C^*$-algebras is a \ki{substochastic map} if its Hilbert-Schmidt adjoint $\map^*$ is a Schwarz contraction, and $\map$ is \ki{stochastic} if it is a trace-preserving substochastic map. 
Note that in the commutative finite-dimensional case substochastic/stochastic maps are exactly the ones that can be represented by substochastic/stochastic matrices.

It is known that if $\map$ is $2$-positive then it is a Schwarz map with $\snorm{\map}=\norm{\map}$. 
In general, however, we might have $\norm{\map}<\snorm{\map}<+\infty$, as the following example shows. In particular, not  every Schwarz map is $2$-positive.
\begin{example}\label{ex:Schwarz maps}
Let $\hil$ be a finite-dimensional Hilbert space, and for every $\ep\in\bR$, let
$\map_\ep:\,\B(\hil)\to\B(\hil)$ be the map
\begin{equation*}
\map_\ep(X):=(1-\ep)X\trans+\ep(\Tr X)I/d,\ds\ds\ds X\in\B(\hil),
\end{equation*}
where $d:=\dim\hil>1$ and $X\trans$ denotes the transpose of $X$ in some fixed basis $\{e_1,\ldots,e_d\}$ of $\hil$.
It was shown in \cite{Tomiyama} that $\map_\ep$ is positive if and only if $0\le \ep\le 1+1/(d-1)$, for $k\ge 2$ it is 
$k$-positive  if and only if $1-1/(d+1)\le\ep\le 1+1/(d-1)$, and it is a Schwarz contraction if and only if
$1-1/\bz 1/2+\sqrt{d+1/4}\jz\le\ep\le 1+1/(d-1)$. This already shows that there are parameter values $\ep$ for which 
$\map_\ep$ is a Schwarz contraction but not $2$-positive.
Moreover, if $\ep\in[0,1)$ then for every $c\in[0,+\infty)$ we have
\begin{align*}
&c\map_\ep(X^*X)-\map_\ep(X^*)\map_\ep(X)\\
&\ds=
c(1-\ep)(X^*X)\trans+c\ep(\Tr X^*X)I/d-
(1-\ep)^2(X^*)\trans X\trans\\
&\ds\ds\s -\ep(1-\ep)(\Tr X)(X^*)\trans/d-\ep(1-\ep)(\Tr X^*)X\trans/d
-\ep^2|\Tr X|^2 I/d^2\\
&\ds\ge
(\Tr X^*X)I/d\left[c\ep-d(1-\ep)^2-2\ep(1-\ep)\sqrt{d}-\ep^2\right],
\end{align*}
where we used that $|\Tr X|^2\le(\Tr I)(\Tr X^*X)$ and $X^*X\le\norm{X}^2 I\le(\Tr X^*X)I$.
This shows that $\map_\ep$ is a Schwarz map for every $\ep\in(0,1)$ and  $\snorm{\map_\ep}\le(1/\ep)(d(1-\ep)^2+2\ep(1-\ep)\sqrt{d}+\ep^2)$. 
Note that for $X:=\diad{e_1}{e_2}$ we have 
\begin{equation*}
0\le\inner{e_1}{\bz\snorm{\map_\ep}\map_\ep(X^*X)-\map_\ep(X^*)\map_\ep(X)\jz e_1}=
\snorm{\map_\ep}\ep/d-(1-\ep)^2,
\end{equation*}
which yields that $\snorm{\map_\ep}\ge d(1-\ep)^2/\ep$. In particular, $\lim_{\ep\searrow 0}\snorm{\map_\ep}=+\infty$.
Since $\map_\ep$ is a positive unital map for every $\ep\in[0,1+1/(d-1)]$, we have $\norm{\map_\ep}=1$ for every $\ep\in[0,1+1/(d-1)]$, while  
$\snorm{\map_\ep}>1$ and hence $\norm{\map_\ep}<\snorm{\map_\ep}$ whenever $(1-\ep)^2/\ep>d$.

Similarly, it was shown in \cite{Tomiyama} that the map
\begin{equation*}
\mapp_\ep(X):=(1-\ep)X+\ep(\Tr X)I/d,\ds\ds\ds X\in\B(\hil),
\end{equation*}
is completely positive if and only if $0\le\ep\le 1+1/(d^2-1)$, for $1\le k\le d-1$ it is $k$-positive if and only if
$0\le \ep\le 1+1/(dk-1)$, and it is a Schwarz contraction if and only if $0\le\ep\le 1+1/d$. A similar computation as above shows that $\mapp_\ep$ is a Schwarz map if and only if $0\le \ep<1+1/(d-1)$, and $\lim_{\ep\nearrow 1+1/(d-1)}\snorm{\mapp_\ep}=+\infty$. 

Finally, the map
\begin{equation*}
\Lambda_\ep(X):=(1-\ep)X\trans+\ep X,\ds\ds\ds X\in\B(\hil),
\end{equation*}
positive if and only if $0\le\ep\le 1$, for each $k\ge2$ it is $k$-positive if and only if $\varepsilon=1$, and
it is a Schwarz contraction if and only if $\varepsilon=1$ \cite{Tomiyama}. Moreover, for $X:=\diad{e_1}{e_2}$ and every $c\in\bR$ we have
$\inner{e_1}{\bz c\Lambda_\ep(X^*X)-\Lambda_\ep(X^*)\Lambda_\ep(X)\jz e_1}=-(1-\ep)^2$, and hence 
$\Lambda_\ep$ is a Schwarz map if and only if $\ep=1$.
\end{example}

\begin{lemma}\label{lemma:unital Schwarz contraction}
Let $\map:\,\A_1\to\A_2$ be a substochastic map, and assume that there exists a $B\in\A_{1,+}\setminus\{0\}$ 
such that $\Tr\map(B)=\Tr B$. Then $\snorm{\map^*}=\norm{\map^*}=1$. 
\end{lemma}
\begin{proof}
Let $\tilde\A_1:=B^0\A_1B^0,\,\tilde\A_2:=\map(B)^0\A_2\map(B)^0$, and define
$\tilde\map:\,\tilde\A_1\to\tilde\A_2$ as $\tilde\map(X):=\map(B^0XB^0)=\map(X),\,X\in\tilde\A_1$. 
Then $\tilde\map^*(Y)=B^0\map^*(Y)B^0,\,Y\in\tilde\A_2$, and Lemma \ref{lemma:support inequality} yields that 
$\tilde\map^*(\map(B)^0)=B^0$, i.e., $\tilde\map^*$ is unital. Hence, $1=\|\tilde\map^*\|\le\norm{\map^*}\le\snorm{\map^*}\le 1$, from which the assertion follows. 
\end{proof}

\begin{lemma}\label{lemma:Schwarz adjoint}
The set of Schwarz maps is closed under composition, taking the adjoint, and positive linear combinations.
Moreover, for $\alpha\ge 0$ and $\map,\map_1,\map_2:\A_1\to\A_2$,
\begin{equation}\label{Schwarz identites}
\snorm{\alpha\map}=\alpha\snorm{\map},\ds\ds\ds
\snorm{\map_1+\map_2}\le \snorm{\map_1}+\snorm{\map_2}.
\end{equation}
\end{lemma}
\begin{proof}
The assertion about the composition is obvious. To prove closedness under the adjoint, 
assume that
$\map:\,\A_1\to\A_2$ is a Schwarz map. Our goal is to prove that $\map^*$ is a 
Schwarz map, too. Let $\iota_k$ be the trivial embedding of $\A_k$ into $\B(\hil_k)$ for 
$k=1,2$. The adjoint $\pi_k:=\iota_k^*$ of $\iota_k$ is the trace-preserving conditional 
expectation (or equivalently, the Hilbert-Schmidt orthogonal projection) from $\B(\hil_k)$ 
onto $\A_k$. 
Since $\iota_k$ is completely positive, so is $\pi_k$, and since $\pi_k$ is unital, it is also a Schwarz contraction.
Let $\tilde\map:=\iota_2\circ\map\circ\pi_1$, the adjoint of which is 
$\tilde\map^*=\iota_1\circ\map^*\circ\pi_2$. Note that $\tilde\map$ is a Schwarz map, too, 
with $\|\tilde\map\|_S=\snorm{\map}$, since for any $X\in\B(\hil_1)$,
\begin{align*}
\tilde\map(X^*)\tilde\map(X)=\iota_2\bz\map(\pi_1(X^*))\map(\pi_1(X))\jz\le
\snorm{\map}\iota_2\map\bz\pi_1(X^*)\pi_1(X)\jz \le
\snorm{\map}\tilde\map(X^*X).
\end{align*}
Hence, for any vector $v\in\hil_1$ and any orthonormal basis $\{e_i\}_{i=1}^{d_1}$ in 
$\hil_1$, we have
\begin{equation*}
\snorm{\map}\tilde\map(\pr{v})\ge\tilde\map(\diad{v}{e_i})\tilde\map(\diad{e_i}{v}),\ds\ds\ds
i=1,\ldots,d_1,
\end{equation*}
where $d_1:=\dim\hil_1$.
Let $Y\in\A_2$ be arbitrary. Multiplying the above inequality with $Y$ from the left and 
$Y^*$ from the right, and taking the trace, we obtain
\begin{equation*}
\snorm{\map}\inner{v}{\tilde\map^*(Y^*Y)v}=
\snorm{\map}\Tr Y\tilde\map(\pr{v})Y^*\ge
\Tr Y\tilde\map(\diad{v}{e_i})\tilde\map(\diad{e_i}{v})Y^*.
\end{equation*}
Note that $\Tr:\,\A_2\to\iC$ is completely positive, and hence it is a Schwarz map with 
$\snorm{\Tr}=\norm{\Tr(I_2)}=d_2:=\dim\hil_2$. Hence, the above inequality can be continued 
as
\begin{align*}
d_2\snorm{\map}\inner{v}{\tilde\map^*(Y^*Y)v}&\ge
\Tr Y\tilde\map(\diad{v}{e_i})\Tr\tilde\map(\diad{e_i}{v})Y^*
=
\inner{v}{\tilde\map^*(Y^*)e_i}\inner{e_i}{\tilde\map^*(Y)v}, 
\end{align*}
and summing over $i$ yields 
\begin{align*}
d_1d_2\snorm{\map}\inner{v}{\tilde\map^*(Y^*Y)v}&\ge
\inner{v}{\tilde\map^*(Y^*)\tilde\map^*(Y)v}.
\end{align*}
Since the above inequality is true for any $v\in\hil_1$, and $\tilde\map^*(Y)=\map^*(Y)$ for any $Y\in\A_2$, the assertion follows.

The assertion on positive linear combinations follows from \eqref{Schwarz identites}, and the first identity in 
\eqref{Schwarz identites} is obvious. To see the second identity, assume first that $\map_1$ and $\map_2$ are Schwarz contractions. Then, for any $\ep\in[0,1]$ and any $X\in\A_1$ we have
\begin{align*}
&\bz(1-\ep)\map_1+\ep\map_2\jz(X^*X)-\bz(1-\ep)\map_1+\ep\map_2\jz(X^*)\bz(1-\ep)\map_1+\ep\map_2\jz(X)\\
&\ds=
(1-\ep)\left[\map_1(X^*X)-\map_1(X^*)\map_1(X)\right]
+\ep\left[\map_2(X^*X)-\map_2(X^*)\map_2(X)\right]\\
&\ds\ds\ds+\ep(1-\ep)\left[\bz\map_1(X)-\map_2(X)\jz^*\bz\map_1(X)-\map_2(X)\jz\right]\ge 0,
\end{align*}
and hence $(1-\ep)\map_1+\ep\map_2$ is a Schwarz contraction for any $\ep\in[0,1]$. Finally, let $\map_1,\map_2:\,\A_1\to\A_2$ be non-zero Schwarz maps. Then $\tilde\map_k:=\map_k/\snorm{\map_k}$ is a Schwarz contraction for $k=1,2$, and choosing $\ep:=\snorm{\map_2}/\bz\snorm{\map_1}+\snorm{\map_2}\jz$, we get 
\begin{equation*}
\snorm{\map_1+\map_2}=\bz\snorm{\map_1}+\snorm{\map_2}\jz\|(1-\ep)\tilde\map_1+\ep\tilde\map_2\|_S\le
\snorm{\map_1}+\snorm{\map_2}.\qedhere
\end{equation*}
\end{proof}

Lemma \ref{lemma:Schwarz} and Corollary \ref{cor:condexp} below are well-known when $\map$ and $\gamma$ are unital $2$-positive maps. Their proofs are essentially the same for Schwarz contractions, which we provide here for the readers' convenience.

\begin{lemma}\label{lemma:Schwarz}
Let $\map:\,\A_1\to\A_2$ be a Schwarz map, and let
\begin{equation*}
\M_{\map}:=\{X\in\A_1\,:\,\map(X)\map(X^*)=\snorm{\map}\map(XX^*)\}.
\end{equation*} Then
\begin{equation}\label{mult dom}
X\in\M_{\map}\ds\ds\ds\text{if and only if}\ds\ds\ds
\map(X)\map(Z)=\snorm{\map}\map(XZ),\ds Z\in\A_1.
\end{equation}
Moreover, the set $\M_{\map}$
is a vector space that is closed under multiplication.
\end{lemma}
\begin{proof}
We may assume that $\snorm{\map}>0$, since otherwise $\map=0$ and the assertions 
become trivial. Define 
$\gamma(X_1,X_2):=\snorm{\map}\map(X_1X_2^*)-\map(X_1)\map(X_2)^*,\,X_1,X_2\in\A_1$.
Let $X\in\M_{\map}$, $Z\in\A_1$ and $t\in\bR$. Then
\begin{align*}
0&\le \gamma(tX+Z,tX+Z)=t^2\gamma(X,X)+t[\gamma(X,Z)+\gamma(Z,X)]+\gamma(Z,Z)\\
&=
t[\gamma(X,Z)+\gamma(Z,X)]+\gamma(Z,Z).
\end{align*}
Since this is true for any $t\in\bR$, we get $\gamma(X,Z)+\gamma(Z,X)=0$, and repeating the same argument with $iZ$ in place of $Z$, we get $\gamma(X,Z)-\gamma(Z,X)=0$. Hence,
$\map(X)\map(Z)=\snorm{\map}\map(XZ)$. The implication in the other direction is obvious.
The assertion about the algebraic structure of $\M_{\map}$ follows immediately from \eqref{mult dom}.
\end{proof}

For a map $\gamma$ 
from a $C^*$-algebra into itself, we denote by $\fix{\gamma}$ the set of fixed points of $\gamma$.
\begin{cor}\label{cor:condexp}
Let $\gamma:\,\A\to\A$ be a Schwarz contraction, and assume that there exists a strictly positive linear functional $\alpha$ on $\A$ such that $\alpha\circ\gamma=\alpha$. Then
$\snorm{\gamma}=\norm{\gamma}=1$,
$\fix{\gamma}$ is a non-zero $C^*$-algebra, $\gamma$ is a $C^*$-algebra morphism on $\fix{\gamma}$, and $\gamma_{\infty}:=\lim_{n\to\infty}\frac{1}{n}\sum_{k=1}^n\gamma^k$ is an $\alpha$-preserving conditional expectation onto $\fix{\gamma}$.
\end{cor}
\begin{proof}
The assumption $\alpha\circ\gamma=\alpha$ is equivalent to $\gamma^*(A)=A$, where 
$\alpha(X)=\Tr AX,\,X\in\A$, and $A$ is strictly positive definite.
Thus $1$ is an eigenvalue of $\gamma^*$ and therefore also of $\gamma$. Hence, the 
fixed-point set of $\gamma$ is non-empty, and it is obviously a linear subspace in $\A$, 
which is also self-adjoint due to the positivity of $\gamma$. If $X\in\fix{\gamma}$ then 
$0\le\alpha\bz \gamma(X^*X)-\gamma(X^*)\gamma(X)\jz=\alpha\bz 
\gamma(X^*X)\jz-\alpha(X^*X)=0$, and hence $\gamma(X^*X)=\gamma(X^*)\gamma(X)=X^*X$, i.e., 
$X^*X\in\fix{\gamma}$. The polarization identity then yields that $\fix{\gamma}$ is closed 
also under multiplication, so it is a $C^*$-subalgebra of $\A$. Let $\tilde I$ be the unit of 
$\fix{\gamma}$; then $1=\|\tilde I\|=\|\gamma(\tilde 
I)\|\le\norm{\gamma}\le\snorm{\gamma}\le 1$, so $\snorm{\gamma}=1$. Repeating the above 
argument with $X^*$ yields that $\fix{\gamma}\subset\M_\gamma\cap\M_\gamma^*$, where $\M_\gamma$ is defined as in Lemma \ref{lemma:Schwarz}. Moreover, by 
Lemma \ref{lemma:Schwarz}, $\gamma$ is a $C^*$-algebra morphism on 
$\M_\gamma\cap\M_\gamma^*$, and hence also on $\fix{\gamma}$.
Note that $\inner{X}{Y}:=\alpha(X^*Y)$ defines an inner product on $\A$ with respect to 
which $\gamma $ is a contraction, and hence $\gamma_\infty$ exists and is the orthogonal 
projection onto $\fix{\gamma}$, due to von Neumann's mean ergodic theorem. By Lemma 
\ref{lemma:Schwarz} we have 
$\gamma(XY)=\gamma(X)\gamma(Y)=X\gamma(Y)$ for any $X\in\fix{\gamma}$ and $Y\in\A$, which 
yields that $\gamma_{\infty}$ is a conditional expectation.
\end{proof}

\begin{lemma}\label{lemma:decomposition}
Let $B_1:=B\in\A_{1,+}$ be non-zero, and let $\map:\,\A_1\to\A_2$ be a trace 
non-increasing $2$-positive map 
such that $\Tr\Phi(B)=\Tr B$. Let $B_2:=\map(B)$. Then there exist decompositions
$\supp B_m=\bigoplus_{k=1}^r\hil_{m,k,L}\otimes \hil_{m,k,R},\,m=1,2$, 
invertible density operators $\omega_{B,k}$ on $\hil_{1,k,R}$ 
and $\tilde \omega_{B,k}$ on $\hil_{2,k,R}$, and
unitaries $U_k\,:\hil_{1,k,L}\to\hil_{2,k,L}$
such that 
\begin{align}
\fix{\map_B^*\circ\map}_+&=\bigoplus_{k=1}^r\B(\hil_{1,k,L})_+\otimes\omega_{B,k},\nonumber\\
\map(A_{1,k,L}\otimes\omega_{B,k})&=U_k A_{1,k,L}U_k^*\otimes \tilde \omega_{B,k},\ds\ds\ds A_{1,k,L}\in\B(\hil_{1,k,L}).\label{unitaries}
\end{align}
\end{lemma}
\begin{proof}
Let $\tilde\A_1:=B^0\A_1B^0,\,\tilde\A_2:=\map(B)^0\A_2\map(B)^0$, and define
$\tilde\map:\,\tilde\A_1\to\tilde\A_2$ as $\tilde\map(X):=\map(B^0XB^0)=\map(X),\,X\in\tilde\A_1$. Then
$\tilde\map^*(Y)=B^0\map^*(Y)B^0$, $Y\in\tilde\A_2$, and a straightforward computation verifies that
$\tilde\map_B(X):=\tilde\map(B)^{-1/2}\tilde\map(B^{1/2}XB^{1/2})\tilde\map(B)^{-1/2}=\map_B(X),\,X\in\tilde\A_1$, and $\tilde\map_B^*(Y):=B^{1/2}\tilde\map^*(\tilde\map(B)^{-1/2}Y\tilde\map(B)^{-1/2})B^{1/2}=
\map_B^*(Y),\,Y\in\tilde\A_2$.
Let $\gamma_1:=\tilde\map^*\circ\tilde\map_B$ and
$\gamma_2:=\tilde\map_B\circ\tilde\map^*$. 
Obviously, $\gamma_1$ and $\gamma_2$ are again $2$-positive and, since
\begin{align*}
\gamma_1(B^0)&=\tilde\map^*(\map(B)^0)=B^0\map^*(\map(B)^0)B^0=B^0,\\
\gamma_2(\map(B)^0)&=
\map(B)^{-1/2}\map(B^{1/2}\map^*(\map(B)^0)B^{1/2})\map(B)^{-1/2}=\map(B)^0
\end{align*} 
due to Lemma \ref{lemma:support inequality},
they are also unital. Hence, $\snorm{\gamma_i}=\norm{\gamma_i}=1,\,i=1,2$.
Note that if $A_1:=A\in\fix{\map_B^*\circ\map}_+$ then $A^0\le B^0$ and hence $A\in\tilde\A_1$, and 
\begin{equation*}
\gamma_1^*(A+B)=\map_B^*(\map(A+B))={A+B},\ds\ds
\gamma_2^*(\map(A+B))=\map(\map_B^*(\map(A+B)))=\map(A+B).
\end{equation*} 
Let $A_2:=\map(A_1)$.
By the above,
$\gamma_m$ leaves the faithful state $\alpha_m$ with density $(A_m+B_m)/\Tr (A_m+B_m)$ invariant, and hence, by Corollary \ref{cor:condexp},
$\fix{\gamma_m}$ is a $C^*$-algebra of the form
$\fix{\gamma_m}=\bigoplus_{k=1}^r\B(\hil_{m,k,L})\otimes I_{m,k,R}$, where 
$\bigoplus_{k=1}^r\hil_{m,k,L}\otimes \hil_{m,k,R}$ is a decomposition of $\supp B_m$. Moreover,
$\lim_{n\to\infty}\frac{1}{n}\sum_{k=1}^n\gamma_m^k$ gives an $\alpha_m$-preserving conditional expectation onto 
$\fix{\gamma_m}$, for $m=1,2$. Hence, by Takesaki's theorem \cite{Takesaki},
$(A_m+B_m)^{it}\fix{\gamma_m}(A_m+B_m)^{-it}=\fix{\gamma_m}$. Now the argument of Section 3 in \cite{MP} yields the 
existence of invertible density operators $\omega_{A,B,k}$ on $\hil_{1,k,R}$ and positive definite operators 
$X_{1,k,L,A,B}$ on  $\hil_{1,k,L}$ such that 
$A+B=\bigoplus_{k=1}^r X_{1,k,L,A,B}\otimes \omega_{A,B,k}$. By Theorem 9.11 in \cite{OP}, we have 
$(A+B)^{it}B^{-it}\in\fix{\gamma_1}$ for every $t\in\R$, which yields that $\omega_{A,B,k}$ is independent of $A$, and hence that every $A\in\fix{\map_B^*\circ\map}_+$ can be written in the form $A=\bigoplus_{k=1}^r A_{1,k,L}\otimes \omega_{B,k}$ with $\omega_{B,k}:=\omega_{A,B,k}$ and some positive semidefinite operators $A_{1,k,L}$ on $\hil_{1,k,L}$. 
This shows that $\fix{\map_B^*\circ\map}_+\subset\bigoplus_{k=1}^r\B(\hil_{1,k,L})_+\otimes\omega_{B,k}$.
For the proof of \eqref{unitaries}, we refer to Theorem 4.2.1 in \cite{MPhD}. Finally, the decomposition
$B=\oplus_{k=1}^r B_{1,k,L}\otimes \omega_{B,k}$ together with \eqref{unitaries} shows that 
$\fix{\map_B^*\circ\map}_+\supset\bigoplus_{k=1}^r\B(\hil_{1,k,L})_+\otimes\omega_{B,k}$.
\end{proof}

\section{Monotonicity}\label{sec:monotonicity}
\setcounter{equation}{0}
Now we turn to the proof of the monotonicity of the $f$-divergences under substochastic maps.
Let $\A_i\subset\B(\hil_i)$ be finite-dimensional $C^*$-algebras for $i=1,2$. 
Recall that we call a map $\map:\,\A_1\to\A_2$ substochastic if $\map^*$ satisfies the Schwarz inequality
\begin{equation*}
\map^*(Y^*)\map^*(Y)\le\map^*(Y^*Y),\ds\ds\ds Y\in\A_2,
\end{equation*}
and $\map$ is called stochastic if it is a trace-preserving substochastic map.

For a  
$B\in\A_{1,+}$ and a substochastic map $\map:\,\A_1\to\A_2$, we define the map
$V:\,\A_2\to\A_1$ as
\begin{equation}\label{def:V}
V(X):=\map^*(X\map(B)^{-1/2})B^{1/2},\ds\ds\ds X\in\A_2.
\end{equation}
Note that $V=R_{B^{1/2}}\circ\map^*\circ R_{\map(B)^{-1/2}}$ and hence
$V^*=R_{\map(B)^{-1/2}}\circ\map\circ R_{B^{1/2}}$,
which yields 
\begin{equation}\label{V star}
V^*(B^{1/2})=\map(B)^{1/2}.
\end{equation}
\begin{lemma}\label{lemma:V}
We have the following equivalence:
\begin{equation*}
V(\map(B)^{1/2})=B^{1/2}\ds\ds\ds\text{if and only if}\ds\ds\ds
\Tr\map(B)=\Tr B.
\end{equation*}
\end{lemma}
\begin{proof}
By definition,
\begin{equation*}
V(\map(B)^{1/2})=\map^*(\map(B)^{1/2}\map(B)^{-1/2})B^{1/2}=
\map^*(\map(B)^0)B^{1/2}.
\end{equation*}
Hence, if $\Tr\map(B)=\Tr B$ then $V(\map(B)^{1/2})=B^{1/2}$ due to Lemma \ref{lemma:support inequality}. On the other hand,
$B^{1/2}=V(\map(B)^{1/2})=\map^*(\map(B)^0)B^{1/2}$ yields $\map^*(\map(B)^0)B^n=B^n,\,n\in\N$,  and hence also \ref{item:functional calculus} of Lemma \ref{lemma:support inequality}, which in turn yields $\Tr\map(B)=\Tr B$.
\end{proof}

\begin{lemma}\label{lemma:contraction}
The map $V$ is a contraction and
\begin{equation}\label{contraction inequality}
V^*\bz \modop{A}{B}\jz V\le L_{\map(A)}R_{\map(B)\inv}.
\end{equation}
Moreover, when $\map^*$ is a C$^*$-algebra morphism, $V$ is an isometry if $\map(B)$ is invertible, and \eqref{contraction inequality} holds with equality if $B$ is invertible.
\end{lemma}
\begin{proof}
Let $X\in\A_2$. Then,
\begin{align}
 \hnorm{VX}^2&=
\Tr (VX)^*(VX)
 =
 \Tr B^{1/2}\map^*(\map(B)^{-1/2}X^*) \map^*(X\map(B)^{-1/2})B^{1/2}\nonumber\\
 &\le
\snorm{\map^*} \Tr B^{1/2} \map^*(\map(B)^{-1/2}XX^*\map(B)^{-1/2})B^{1/2}\label{ineq:monotonicity of f-div1}\\
&=
\snorm{\map^*} \Tr\map(B)\map(B)^{-1/2}XX^*\map(B)^{-1/2}
=
\snorm{\map^*}\Tr\map(B)^0 XX^*\nonumber\\
& \le
 \snorm{\map^*}\Tr XX^*=\snorm{\map^*}\hnorm{X}^2\le\hnorm{X}^2.\label{ineq:monotonicity of f-div2}
\end{align}
If $\map^*$ is a C$^*$-algebra morphism then $\snorm{\map^*}=1$ and the inequality in \eqref{ineq:monotonicity of f-div1} holds with equality, and if $\map(B)$ is invertible then and the inequality in
\eqref{ineq:monotonicity of f-div2} holds with equality.
Similarly,
 \begin{align}
 \hinner{X}{V^*\bz \modop{A}{B}\jz V X}&=
  \Tr (VX)^*A(VX)B\inv\nonumber\\
  &=
  \Tr B^{1/2} \map^*(\map(B)^{-1/2}X^*) A \map^*(X\map(B)^{-1/2}) B^{1/2}B\inv\nonumber\\
  &=
  \Tr A \map^*(X\map(B)^{-1/2}) B^0 \map^*(\map(B)^{-1/2}X^*)\nonumber\\
  &\le
  \Tr A  \map^*(X\map(B)^{-1/2}) \map^*(\map(B)^{-1/2}X^*) \label{ineq:monotonicity of f-div3}\\
  &\le
  \snorm{\map^*}\Tr A  \map^*(X\map(B)^{-1/2}\map(B)^{-1/2}X^*)  \label{ineq:monotonicity of f-div4}\\
  &=
  \snorm{\map^*}\Tr \map(A) X\map(B)^{-1} X^* 
  =
  \snorm{\map^*}\hinner{X}{L_{\map(A)}R_{\map(B)\inv}X}\nonumber\\
&\le
\hinner{X}{L_{\map(A)}R_{\map(B)\inv}X}.\label{ineq:monotonicity of f-div5}
 \end{align}
If $\map^*$ is a C$^*$-algebra morphism then $\snorm{\map^*}=1$ and the inequalities in
\eqref{ineq:monotonicity of f-div4} and \eqref{ineq:monotonicity of f-div5} hold with equality,
and if $B$ is invertible then \eqref{ineq:monotonicity of f-div3} holds with equality.
\end{proof}

Recall that a real-valued function $f$ on $[0,+\infty)$
is \ki{operator convex} if $f(tA+(1-t)B)\le tf(A)+(1-t)f(B)$, $t\in[0,1]$,
for any  positive semi-definite operators $A,B$ on any finite-dimensional Hilbert space (or equivalently, on some infinite-dimensional Hilbert space).
For a continuous real-valued function $f$ on $[0,+\infty)$, the following are equivalent
(see \cite[Theorem 2.1]{HanPed}): (i) $f$ is operator convex on $[0,+\infty)$ and $f(0)\le 0$; (ii) $f(V^*AV)\le V^*f(A)V$ for any contraction $V$ and any positive semi-definite operator $A$.
The function $f$ is \ki{operator monotone decreasing} if $f(A)\ge f(B)$ whenever $A$ and $B$ are such that $0\le A\le B$. If $f$ is operator monotone decreasing on $[0,+\infty)$ then it is also operator convex (see the proof of \cite[Theorem 2.5]{HanPed} or  \cite[Theorem V.2.5]{Bhatia}).
A function $f$ is \ki{operator concave} (resp., \ki{operator monotone increasing}) if
$-f$ is operator convex (resp., operator monotone decreasing).
An operator convex function on $[0,+\infty)$ is automatically continuous on $(0,+\infty)$, but might be 
discontinuous at $0$. For instance, a straightforward computation shows that the characteristic 
function $\egy_{\{0\}}$ of the set $\{0\}$ is operator convex on $[0,+\infty)$.
It is easy to verify that the functions
\begin{equation}\label{def:vfi_t}
\vfi_t(x):=-\frac{x}{x+t}=-1+\frac{t}{x+t}
\end{equation}
are operator monotone decreasing and hence operator convex on $[0,+\infty)$ for every $t\in(0,+\infty)$.

\begin{thm}\label{thm:monotonicity for F}
Let $A,B\in\A_{1,+}$, let $\map:\,\A_1\to\A_2$ be a substochastic map such that
$\Tr\map(B)=\Tr B$, and let $f$ be an operator convex function on $[0,+\infty)$. Assume that
\begin{equation}\label{monotonicity conditions}
\Tr\map(A)=\Tr A\ds\ds\ds\text{or}\ds\ds\ds 0\le\of.
\end{equation}
Then,
\begin{equation}\label{f-div monotonicity inequality}
\rsr{\map(A)}{\map(B)}{f}\le \rsr{A}{B}{f}.
\end{equation}
\end{thm}
\begin{proof}
First we prove the theorem when $f$ is continuous at $0$.
Due to Theorem \ref{thm:integral representation}, we have the representation
\begin{equation*}
f(x)=f(0)+ax+bx^2
+\int_{(0,\infty)}\bz\frac{x}{1+t}+\vfi_t(x)\jz\,d\mu(t),\ds x\in[0,+\infty),
\end{equation*}
where $b\ge 0$ and $\vfi_t(x)$ is given in \eqref{def:vfi_t}. Define
\begin{equation*}
\Delta:=\modop{A}{B}\ds\ds\ds\text{and}\ds\ds\ds \tilde\Delta:= L_{\map(A)}R_{\map(B)\inv}.
\end{equation*}
Then
\begin{align}
\rsr{A}{B}{f}=&
f(0)\Tr B+
a\Tr AB^0+b\Tr A^2B\inv\nonumber\\
&+
\int_{(0,+\infty)}\bz\frac{\Tr AB^0}{1+t}+\rsr{A}{B}{\vfi_t}\jz\,d\mu(t)+\of\Tr A(I-B^0).
\label{fdiv decomposition}
\end{align}

Note that $\Tr B=\Tr\map(B)$ by assumption and, since $b\ge 0$, we have $b\Tr A^2B\inv\ge b\Tr \map(A)^2\map(B)\inv$ due to Lemma \ref{lemma:2 Renyi monotonicity}.
Since $\vfi_t$ is operator convex, operator 
monotonic decreasing and $\vfi_t(0)=0$, we have
\begin{equation}\label{F-3.14}
V^*\vfi_t(\Delta)V\ge \vfi_t(V^*\Delta V)\ge \vfi_t(\tilde\Delta)
\end{equation}
for the contraction $V$ defined in \eqref{def:V},
due to \eqref{contraction inequality} and \cite[Theorem 2.1]{HanPed} as mentioned above.
Hence, by Lemma \ref{lemma:V},
\begin{align}
\rsr{A}{B}{\vfi_t}=\hinner{B^{1/2}}{\vfi_t(\Delta) B^{1/2}}&=
\hinner{V\map(B)^{1/2}}{\vfi_t(\Delta) V\map(B)^{1/2}}\nonumber\\
&\ge
\hinner{\map(B)^{1/2}}{\vfi_t(\tilde\Delta) \map(B)^{1/2}}=\rsr{\map(A)}{\map(B)}{\vfi_t}.\label{F-3.15}
\end{align}
Therefore, in order to prove the monotonicity inequality \eqref{f-div monotonicity inequality}, it suffices to prove the monotonicity of the remaining terms in 
\eqref{fdiv decomposition}.

Assume first that $\supp A\le\supp B$, and hence also $\Tr\map(A)=\Tr A$ (see Lemma \ref{lemma:support inequality}). Then
$\Tr AB^0=\Tr A=\Tr\map(A)=\Tr\map(A)\map(B)^0$, which also yields 
$\Tr A(I_1-B^0)=\Tr\map(A)(I_2-\map(B)^0)$. Hence, all the terms in \eqref{fdiv decomposition} are monotonic non-increasing under $\map$, and therefore 
we have the inequality \eqref{f-div monotonicity inequality}.

If $\of=+\infty$, then either $\supp A\nleq\supp B$, in which case
\begin{equation*}
\rsr{A}{B}{f}=+\infty\ge \rsr{\map(A)}{\map(B)}{f},
\end{equation*}
or we have $\supp A\le\supp B$,
and hence \eqref{f-div monotonicity inequality} follows by the previous argument.

Next, assume that $\Tr\map(A)=\Tr A$, and define
$B_\ep:=B+\ep A,\,\ep>0$. Then $\Tr\map(B_\ep)=\Tr\map(B)+\ep\Tr\map(A)=\Tr B+\ep\Tr A=\Tr B_\ep$, and
$\supp A\le \supp B_\ep$. Hence, by the previous argument,
\begin{align}\label{mon proof2}
\rsr{\map(A)}{\map(B_\ep)}{f}\le\rsr{A}{B_\ep}{f}.
\end{align}
By the previous paragraph, it is sufficient to consider the case where $\of$ is finite, and therefore
Proposition \ref{prop:continuity} can be used to obtain \eqref{f-div monotonicity inequality} 
by taking the limit $\ep\searrow 0$ in \eqref{mon proof2}.

Finally, assume that $0\le\of<+\infty$. By Proposition \ref{prop:finite of}, this yields the representation
\begin{equation*}
f(x)=f(0)+\of x+\int_{(0,\infty)}\vfi_t(x)\,d\mu(t),
\end{equation*}
and hence
\begin{align*}
\rsr{A}{B}{f}&=f(0)\Tr B+\of\Tr AB^0+
\int_{(0,+\infty)}\rsr{A}{B}{\vfi_t}\,d\mu(t)+\of\Tr A(I-B^0)\\
&=
f(0)\Tr B+\of\Tr A+\int_{(0,+\infty)}\rsr{A}{B}{\vfi_t}\,d\mu(t).
\end{align*}
Since $\Tr\map(A)\le\Tr A$, inequality \eqref{f-div monotonicity inequality} follows.

So far, we have proved the theorem for the case where $f$ is continuous at $0$. Consider the functions $\tilde f_\alpha(x):=-x^\alpha,\,x\ge 0,\,0<\alpha<1$. Then $\tilde f_\alpha$ is operator convex, continuous at $0$ and $\omega(\tilde f_\alpha)=0$ for all $\alpha\in(0,1)$. Hence, by the above, we have 
\begin{equation}\label{Renyi monotonicity 01}
-\Tr\map(A)^{\alpha}\map(B)^{1-\alpha}=\rsr{\map(A)}{\map(B)}{\tilde f_\alpha}\le\rsr{A}{B}{\tilde f_\alpha}=-\Tr A^{\alpha}B^{1-\alpha},\ds\ds \alpha\in(0,1).
\end{equation}
Taking the limit $\alpha\searrow 0$, we obtain
\begin{equation}\label{Renyi monotonicity 0}
\Tr\map(A)^0\map(B)\ge\Tr A^0B,
\end{equation}
which in turn yields
\begin{equation}\label{singular monotonicity}
\rsr{\map(A)}{\map(B)}{\egy_{\{0\}}}=\Tr\map(B)-\Tr\map(A)^0\map(B)\le\Tr B-\Tr A^0B=\rsr{A}{B}{\egy_{\{0\}}}.
\end{equation}

Assume now that $f$ is an operator convex function on $[0,+\infty)$, that is not necessarily continuous at $0$. 
Convexity of $f$ yields that $f(0^+):=\lim_{x\searrow 0}f(x)$ is finite,  and $\alpha:=f(0)-f(0^+)\ge 0$. Note that
$\tilde f:=f-\alpha\egy_{\{0\}}$ is operator convex and continuous at $0$, $\omega(\tilde f)=\omega(f)$, and
$\rsr{A}{B}{f}=\rsr{A}{B}{\tilde f}+\alpha\rsr{A}{B}{\egy_{\{0\}}}$ for any $A,B\in\A_{1,+}$.
Applying the previous argument to $\tilde f$ and using \eqref{singular monotonicity}, we see that 
\begin{align*}
\rsr{\map(A)}{\map(B)}{f}&=\rsr{\map(A)}{\map(B)}{\tilde f}+\alpha\rsr{\map(A)}{\map(B)}{\egy_{\{0\}}}\\
&\le
\rsr{A}{B}{\tilde f}+\alpha\rsr{A}{B}{\egy_{\{0\}}}=
\rsr{A}{B}{f}
\end{align*}
if any of the conditions in \eqref{monotonicity conditions} holds,
completing the proof of the theorem.
\end{proof}

\begin{rem}
Note that $\supp A\le\supp B$ is also sufficient for \eqref{f-div monotonicity inequality} to hold, 
due to Lemma \ref{lemma:support inequality}.
\end{rem}

\begin{ex}\label{ex:Renyi1}
Let $A,B\in\A_{1,+}$ and $\map:\,\A_1\to\A_2$ be a substochastic map such that $\Tr\map(B)=\Tr B$.
Let $\sgn x:=x/|x|,\,x\ne 0$, and define
$\tilde f_\alpha:=\sgn(\alpha-1)f_\alpha,\,0<\alpha\ne 1$, where $f_\alpha$ is given in
Example \ref{ex: f-div Renyi}.
Since $\tilde f_\alpha$ is operator convex, and
$\omega(\tilde f_\alpha)\ge 0$ for all $\alpha\in[0,2]\setminus\{1\}$, Theorem \ref{thm:monotonicity for F} yields that 
\begin{align}
\sgn(\alpha-1)\Tr\map(A)^\alpha\map(B)^{1-\alpha}&=
\rsr{\map(A)}{\map(B)}{\tilde f_\alpha}\nonumber\\
&\le
\rsr{A}{B}{\tilde f_\alpha}=\sgn(\alpha-1)\Tr A^\alpha B^{1-\alpha}
\label{Renyi monotonicity2}
\end{align}
when $\alpha\in(1,2]$ and $\supp A\le\supp B$. (Note that 
$\rsr{\map(A)}{\map(B)}{\tilde f_\alpha}\le\rsr{A}{B}{\tilde f_\alpha}=+\infty$ is trivial when $\alpha\in(1,2]$ and 
$\supp A\nleq\supp B$.)
The same inequality has been shown in the proof of Theorem \ref{thm:monotonicity for F} for $\alpha\in[0,1)$; see 
\eqref{Renyi monotonicity 01} and \eqref{Renyi monotonicity 0}.
This yields the monotonicity of the R\'enyi relative entropies,
\begin{align}
\rsr{\map(A)}{\map(B)}{\alpha}
&=\frac{1}{\alpha-1}\log\rsr{\map(A)}{\map(B)}{f_\alpha}
\le
\frac{1}{\alpha-1}\log\rsr{A}{B}{f_\alpha}=
\rsr{A}{B}{\alpha}\label{Renyi monotonicity}
\end{align}
for $\alpha\in[0,2]\setminus\{1\}$.

Since $\of\ge 0$ for $f(x):=x\log x$, Theorem \ref{thm:monotonicity for F} also yields the monotonicity of the relative entropy,
\begin{equation*}
\sr{\map(A)}{\map(B)}\le \sr{A}{B}.
\end{equation*}
\end{ex}
\medskip

\begin{rem}
In the proof of Theorem \ref{thm:monotonicity for F} it was essential that $f$ is operator 
convex, but it is not known if it is actually necessary. See Appendix \ref{sec:classical} for some special cases where convexity of $f$ is sufficient.
\end{rem}

Theorem \ref{thm:monotonicity for F} yields the joint convexity of the $f$-divergences:
\begin{cor}\label{cor:joint convexity}
Let $A_i,B_i\in\A_+$ and $p_i\ge0$ for $i=1,\ldots,r$, and let $f$ be an operator convex function on $[0,+\infty)$. Then
\begin{equation*}
S_f\bz\sum\nolimits_i p_iA_i\big\|\sum\nolimits_i p_iB_i\jz\le\sum\nolimits_i p_i\rsr{A_i}{B_i}{f}.
\end{equation*}
\end{cor}
\begin{proof}
Let $\delta_1,\ldots,\delta_r$ be a set of orthogonal rank-one projections on $\iC^r$, and define
$A:=\sum_{i=1}^r p_iA_i\otimes\delta_i,\,B:=\sum_{i=1}^r p_iB_i\otimes\delta_i$. The map 
$\map:\,\A\otimes\B(\iC^r)\to\A$, given by $\map(X\otimes Y):=X\Tr Y,\,X\in\A,\,Y\in\B(\iC^r)$, is completely positive and trace-preserving and hence, by Theorem \ref{thm:monotonicity for F},
\begin{equation}\label{joint convexity}
S_f\bz\sum\nolimits_i p_iA_i\big\|\sum\nolimits_i p_iB_i\jz
=\rsr{\map(A)}{\map(B)}{f}\le\rsr{A}{B}{f}=\sum\nolimits_i p_i\rsr{A_i}{B_i}{f},
\end{equation}
where the last identity is due to Corollary \ref{cor:scaling}.
\end{proof}
\medskip

\begin{rem}\label{rem:decomposability}
For an operator convex function $f$ on $[0,+\infty)$ let $\M_f(\A_1,\A_2)$ denote the set
of positive linear maps $\map:\,\A_1\to\A_2$ such that the monotonicity
$\rsr{\map(A)}{\map(B)}{f}\le\rsr{A}{B}{f}$ holds for all $A,B\in\A_1$. 
The joint convexity of the $f$-divergences shows that $\M_f(\A_1,\A_2)$ is convex.
Indeed, if $\map_1,\map_2\in\M_f(\A_1,\A_2)$ then
Corollary \ref{cor:joint convexity} yields
\begin{align*}
&\rsr{(1-\lambda)\map_1(A)+\lambda\map_2(A)}{(1-\lambda)\map_1(B)+\lambda\map_2(B)}{f}\\
&\ds\le
(1-\lambda)\rsr{\map_1(A)}{\map_1(B)}{f}+\lambda\rsr{\map_2(A)}{\map_2(B)}{f}\\
&\ds\le
(1-\lambda)\rsr{A}{B}{f}+\lambda\rsr{A}{B}{f}=\rsr{A}{B}{f}
\end{align*}
for any $\lambda\in[0,1]$ and $A,B\in\A_1$.
Note also that if $\map_1\in\M_f(\A_1,\A_2)$ and $\map_2\in\M_f(\A_2,\A_3)$ then
$\map_2\circ\map_1\in\M_f(\A_1,\A_3)$.

We say that a linear map $\map:\,\A_1\to\A_2$ 
is a \ki{co-Schwarz map} if there is a
$c\in[0,\infty)$ such that
\begin{equation*}
\map(X^*)\map(X)\le c\map(XX^*),\ds\ds\ds X\in\A_1,
\end{equation*}
and it is a \ki{co-Schwarz contraction} if the above inequality holds with $c=1$. It is
easy to see that a linear map $\map:\,\A_1\to\A_2$ is a co-Schwarz map (resp., a
co-Schwarz contraction) if and only if there is a Schwarz map (resp., a Schwarz contraction)
$\tilde\map:\,\A_1\trans\to\A_2$ such that $\map=\tilde\map\circ T$, where
$T(X):=X\trans$ denotes the transpose of $X\in\A_1$ with respect to a fixed orthonormal
basis of $\hil_1$, and $\A_1\trans:=\{X\trans\,:\,X\in\A_1\}\subset B(\hil_1)$.
Furthermore, we say that $\map$ is \ki{co-substochastic} (resp., \ki{co-stochastic}) if
$\map^*$ is a a co-Schwarz contraction (resp., a unital co-Schwarz contraction).
Theorem \ref{thm:monotonicity for F} holds also when $\map:\,\A_1\to\A_2$ is a
co-substochastic map. This follows immediately from Theorem \ref{thm:monotonicity for F}
and the fact that transpositions leave every $f$-divergences invariant (see
\ref{item:isometry invariance} of Corollary \ref{cor:scaling}). Alternatively, this can be
proved by replacing the operator $V$ defined in \eqref{def:V} with the conjugate-linear map 
\begin{equation}\label{(4.22)}
\hat V(X):=\map^*(\map(B)^{-1/2}X^*)B^{1/2},\ds\ds\ds X\in\A_2,
\end{equation}
and following the proofs of Lemma \ref{lemma:contraction} and Theorem \ref{thm:monotonicity
for F} with $\hat V$ in place of $V$. 

Recall that a positive map is called decomposable if it can be written as the sum of a completely positive map and a completely positive map composed with a transposition. By the above, a similar notion of decomposability is sufficient for the monotonicity of the $f$-divergences. Namely, if a trace-preserving positive map $\map:\,\A_1\to\A_2$ is decomposable in the sense that it can be written as a convex combination of a stochastic and a co-stochastic map then 
$\map\in\M_f(\A_1,\A_2)$ for any operator convex function $f$ on $[0,+\infty)$.
Example \ref{ex:Schwarz maps} provides simple examples of trace-preserving positive maps that are decomposable in this sense but which are neither stochastic nor co-stochastic.
\end{rem}

\section{Equality in the monotonicity}\label{sec:equality}
\setcounter{equation}{0}

In this section we analyze the situation where the monotonicity inequality 
\begin{equation*}
\rsr{\map(A)}{\map(B)}{f}\le\rsr{A}{B}{f}
\end{equation*} 
holds with equality, based on the integral representation of operator convex functions that we give in Section \ref{sec:integral representation}. 

By Theorem \ref{thm:integral representation}, every operator convex function $f$ on $[0,+\infty)$ admits a decomposition
\begin{equation}\label{integral decomposition3}
f(x)=\alpha\egy_{\{0\}}(x)+f(0^+)+ax+bf_2(x)
+\int_{(0,\infty)}\bz\frac{x}{1+t}+\vfi_t(x)\jz\,d\mu_f(t),\ds\ds\ds x\in[0,+\infty),
\end{equation}
where $\alpha,b\ge 0$, $f(0^+):=\lim_{x\searrow 0}f(x)$, $\egy_{\{0\}}$ is the characteristic function of the singleton $\{0\}$,
$f_2(x):=x^2$,  $\vfi_t(x)$ is given in \eqref{def:vfi_t}, and $\mu_f$ is a 
positive measure on $(0,+\infty)$.

Recall that $\spect(X)$ denotes the spectrum of an operator $X$. 
We will use the notation $|H|$ to denote the cardinality of a set $H$.
Given $B\in\A_{1,+}$ and a positive map $\map:\,\A_1\to\A_2$, let
$\map_B:\,\A_1\to\A_2$ and $\map_B^*:\,\A_2\to\A_1$ 
be the maps defined in \eqref{def:PhiB} and \eqref{def:PhistarB}.

\begin{thm}\label{thm:equality}
Let $A,B\in\A_{1,+}$ be such that $\supp A\le\supp B$,
let $\map:\A_1\to\A_2$ be a substochastic map such that $\Tr\map(B)=\Tr B$, and define
\begin{equation*}
\Delta:=\modop{A}{B}\ds\ds\ds\text{and}\ds\ds\ds \tilde\Delta:= L_{\map(A)}R_{\map(B)\inv}.
\end{equation*}
Then, for the following conditions 
\ref{item:existence of inverse}--\ref{item:special inverse}, we have
\begin{equation*}
\mbox{\ref{item:existence of inverse}$\imp$\ref{item:existence of inverse2}$\imp$
\ref{item:equality1}$\imp$\ref{item:equality2}$\iff$\ref{item:primitives}$\iff$
\ref{item:cocycle}$\iff$\ref{item:cocycle3}$\iff$\ref{item:cocycle2}$\iff$
\ref{item:logarithms}$\imp$\ref{item:special inverse},}
\end{equation*}
and if $\map$ is $2$-positive then
\ref{item:special inverse}$\imp$\ref{item:existence of inverse} holds as well.
\begin{enumerate}
\item\label{item:existence of inverse}
There exists a stochastic map $\mapp:\A_2\to\A_1$ such that
\begin{equation}\label{reversion on A and B}
\mapp(\map(A))=A,\ds\ds\ds\mapp(\map(B))=B.
\end{equation}
\item\label{item:existence of inverse2}
There exists a substochastic map $\mapp:\A_2\to\A_1$ such that 
\eqref{reversion on A and B} holds.
\item\label{item:equality1}
For every operator convex function $f$ on $[0,+\infty)$,
\begin{equation}\label{fdiv equality}
\rsr{\map(A)}{\map(B)}{f}=\rsr{A}{B}{f}.
\end{equation}
\item\label{item:equality2}
The equality in \eqref{fdiv equality} holds for some operator convex function $f$ 
on $[0,+\infty)$ such that 
\begin{equation}\label{support condition}
|\supp\mu_f|\ge |\spect(\Delta)\cup\spect(\tilde\Delta)|.
\end{equation}
\item\label{item:primitives}
There exists a $T\subset(0,+\infty)$ such that $|T|\ge |\spect(\Delta)\cup\spect(\tilde\Delta)|$ and
\begin{equation*}
\rsr{\map(A)}{\map(B)}{\vfi_t}=\rsr{A}{B}{\vfi_t},\ds\ds\ds t\in T.
\end{equation*}
\item\label{item:cocycle}
$B^0\map^*\bz \map(B)^{-z}\map(A)^{z}\jz=B^{-z}A^{z}$ for all $z\in\iC$.
\item\label{item:cocycle3}
$B^0\map^*\bz \map(B)^{-\alpha}\map(A)^{\alpha}\jz=B^{-\alpha}A^{\alpha}$ for some $\alpha\in(0,2)\setminus\{1\}$.
\item\label{item:cocycle2}
$B^0\map^*\bz \map(B)^{-it}\map(A)^{it}\jz=B^{-it}A^{it}$ for all $t\in\R$.
\item\label{item:logarithms}
$B^0\map^*\bz \log^*\map(A)-(\log^*\map(B))\map(A)^0\jz=\log^* A-(\log^* B)A^0$.
\item\label{item:special inverse}
$\map^*_B(\map(A))=A$.
\end{enumerate}
Moreover, \ref{item:existence of inverse2}$\imp$\ref{item:equality1} holds 
without assuming that
$\supp A\le\supp B$. If $\map$ is $n$-positive/
completely positive then $\mapp$ in \ref{item:existence of inverse} can also be 
assumed to be $n$-positive/completely positive.
\end{thm}
\begin{proof}
The implication 
\ref{item:existence of inverse}$\imp$\ref{item:existence of inverse2} is obvious. Assume 
that \ref{item:existence of inverse2} holds, and let $\tilde A:=\map(A),\,\tilde 
B:=\map(B)$. Then $\Tr A=\Tr\Psi(\tilde A)\le\Tr\tilde A=\Tr\map(A)\le\Tr A$ and similarly 
for $B$ and $\tilde B$, which yields $\Tr\Psi(\tilde A)=\Tr\tilde 
A,\,\Tr\Psi(\tilde B)=\Tr\tilde B$ and $\Tr\map(A)=\Tr A,\,\Tr\map(B)=\Tr B$ (note that this 
latter is automatic here, and not necessary to assume from the beginning).
Applying Theorem \ref{thm:monotonicity for F} twice, we get that
$\rsr{A}{B}{f}=\rsr{\Psi(\tilde A)}{\Psi(\tilde B)}{f}\le \rsr{\tilde A}{\tilde B}{f}=
\rsr{\map(A)}{\map(B)}{f}$ $\le\rsr{A}{B}{f}$ for any operator convex function $f$ on $[0,+\infty)$, proving 
\ref{item:equality1}.
The implication
\ref{item:equality1}$\imp$\ref{item:equality2} is again obvious.

Note that if $A=0$ then $\rsr{A}{B}{f}=f(0)\Tr B$ for any function $f$, and 
\ref{item:existence of inverse}--\ref{item:special inverse} hold true automatically. Hence, 
for the rest we will assume that $A\ne 0$ and hence also $B\ne 0$.

Assume that \ref{item:equality2} holds, i.e., $\rsr{\map(A)}{\map(B)}{f}=\rsr{A}{B}{f}$ for an operator convex function $f$ on $[0,+\infty)$ satisfying \eqref{support condition}. By \eqref{integral decomposition3},
we have
\begin{align*}
\rsr{A}{B}{f}&=
\alpha\rsr{A}{B}{\egy_{\{0\}}}+
f(0^+)\Tr B+a\Tr A+b\rsr{A}{B}{f_2}\\
&\ds+
\int_{(0,+\infty)}\bz\frac{\Tr A}{1+t}+\rsr{A}{B}{\vfi_t}\jz\,d\mu(t)
\end{align*}
(cf.~\eqref{fdiv decomposition}).
Note that $\Tr\map(B)=\Tr B$ by assumption and $\Tr\map(A)=\Tr A$ follows due to 
Lemma \ref{lemma:support inequality}.
Thus,
\begin{align*}
0&=\rsr{A}{B}{f}-\rsr{\map(A)}{\map(B)}{f}\\
&=
\alpha\bz\rsr{A}{B}{\egy_{\{0\}}}-\rsr{\map(A)}{\map(B)}{\egy_{\{0\}}}\jz+
b\bz\rsr{A}{B}{f_2}-\rsr{\map(A)}{\map(B)}{f_2}\jz\\
&\ds+
\int_{(0,+\infty)}\bz
\rsr{A}{B}{\vfi_t}-\rsr{\map(A)}{\map(B)}{\vfi_t}\jz\,d\mu_f(t).
\end{align*}
By Theorem \ref{thm:monotonicity for F}, the $f$-divergences corresponding to 
$\egy_{\{0\}},\,f_2$ and $\vfi_t$ are monotonic non-increasing under $\map$, and hence the
above equality yields that
\begin{align*}
\rsr{\map(A)}{\map(B)}{\vfi_t}=\rsr{A}{B}{\vfi_t}
\end{align*}
for all $t\in\supp\mu_f$.
This gives \ref{item:primitives} with $T:=\supp \mu_f$.

Assume now that \ref{item:primitives} holds. This means that
for every $t\in T$,
\begin{align*}
0=\rsr{A}{B}{\vfi_t}-\rsr{\map(A)}{\map(B)}{\vfi_t}
=
\hinner{\map(B)^{1/2}}{( V^*\vfi_t(\Delta) V-\vfi_t(\tilde\Delta))\map(B)^{1/2}},
\end{align*}
where we used that $V\map(B)^{1/2}=B^{1/2}$ due to Lemma \ref{lemma:V}
(note that $\omega(\vfi_t)=0,\,t>0$).
By \eqref{F-3.14} this is equivalent to
\begin{equation*}
V^*\vfi_t(\Delta) V\map(B)^{1/2}=\vfi_t(\tilde\Delta)\map(B)^{1/2},\ds\ds\ds t\in T,
\end{equation*}
or equivalently,
\begin{equation*}
V^*\left[-I_1+t(\Delta+tI_1)\inv\right] B^{1/2}=\left[-I_2+t(\tilde\Delta+tI_2)\inv\right]\map(B)^{1/2},\ds\ds\ds t\in T.
\end{equation*}
By \eqref{V star} we get
\begin{equation*}
V^*(\Delta+tI_1)\inv B^{1/2}=(\tilde\Delta+tI_2)\inv\map(B)^{1/2},\ds\ds\ds t\in T.
\end{equation*}
Using Lemma \ref{lemma:approximation} below and the assumption that 
$|T|\ge |\spect(\Delta)\cup\spect(\tilde\Delta)|$, we obtain
\begin{equation}\label{(5.5)}
V^*h(\Delta) B^{1/2}=h(\tilde\Delta)\map(B)^{1/2}
\end{equation}
for any function $h$ on $\spect(\Delta)\cup\spect(\tilde\Delta)$. In particular,
\begin{equation}\label{monotonicity equality1}
V^*(\Delta+tI_1)^{-\gamma} B^{1/2}=(\tilde\Delta+tI_2)^{-\gamma}\map(B)^{1/2},\ds\ds\ds \gamma,t>0.
\end{equation}
Using \eqref{monotonicity equality1} with $\gamma=1$ and $\gamma=2$, we obtain
\begin{align*}
\hnorm{V^*(\Delta+tI_1)\inv B^{1/2}}^2&=
\hinner{(\tilde\Delta+tI_2)\inv\map(B)^{1/2}}{(\tilde\Delta+tI_2)\inv\map(B)^{1/2}}\\
&=
\hinner{(\tilde\Delta+tI_2)^{-2}\map(B)^{1/2}}{\map(B)^{1/2}}\\
&=
\hinner{V^*(\Delta+tI_1)^{-2} B^{1/2}}{\map(B)^{1/2}}\\
&=
\hinner{(\Delta+tI_1)^{-2} B^{1/2}}{B^{1/2}}\\
&=
\hnorm{(\Delta+tI_1)\inv B^{1/2}}^2.
\end{align*}
Therefore, we have $\hnorm{V^*x}^2=\hnorm{x}^2$ for $x:=(\Delta+tI_1)^{-1} B^{1/2}$, and since $V$ is a contraction, we get
$0\le\hnorm{VV^*x-x}^2=\hnorm{VV^*x}^2-2\hnorm{V^*x}^2+\hnorm{x}^2=\hnorm{VV^*x}^2-\hnorm{x}^2\le 0$, by which
$VV^*(\Delta+tI_1)\inv B^{1/2}=(\Delta+tI_1)\inv B^{1/2}$. Substituting \eqref{monotonicity equality1} with $\gamma=1$, we finally obtain
\begin{equation}\label{monotonicity equality2}
V(\tilde\Delta+tI_2)\inv\map(B)^{1/2}=(\Delta+tI_1)\inv B^{1/2},\ds\ds\ds t>0,
\end{equation}
and using again Lemma \ref{lemma:approximation}, we get 
\begin{equation*}
Vh(\tilde\Delta)\map(B)^{1/2}=h(\Delta) B^{1/2}
\end{equation*}
for any function $h$ on $\spect(\Delta)\cup\spect(\tilde\Delta)$.
By the definition \eqref{def:V} of $V$, this means that
\begin{equation*}
\map^*\bz \bz h(\tilde\Delta)\map(B)^{1/2}\jz\map(B)^{-1/2}\jz B^{1/2}=h(\Delta) B^{1/2}.
\end{equation*}
In particular, the choice $h(x):=x^{z}, x>0,\,h(0):=0$, yields 
\begin{equation}\label{cocycle2}
\map^*\bz \map(A)^z\map(B)^{-z}\jz B^{1/2} =A^{z}B^{1/2-z},\ds\ds\ds z\in\bC.
\end{equation}
Multiplying from the right with $B^{-1/2}$ and taking the adjoint, we obtain 
\ref{item:cocycle}.

The implication \ref{item:cocycle}$\imp$\ref{item:cocycle3} is obvious. Assume now that 
\ref{item:cocycle3} holds, i.e., 
$B^{-\alpha}A^{\alpha}=B^0\map^*\bz \map(B)^{-\alpha}\map(A)^{\alpha}\jz$ for some $\alpha\in(0,2)\setminus\{1\}$. Multiplying by $B$ and taking the trace, we obtain
\begin{align*}
\rsr{A}{B}{f_\alpha}&=\Tr A^{\alpha}B^{1-\alpha}=
\Tr B\map^*\bz \map(B)^{-\alpha}\map(A)^{\alpha}\jz
=
\Tr \map(B)\map(B)^{-\alpha}\map(A)^{\alpha}\\
&=
\rsr{\map(A)}{\map(B)}{f_\alpha},
\end{align*}
where $f_\alpha(x):=x^\alpha,\,x\ge 0$. Since the support of the representing measure $\mu_{f_\alpha}$ is $(0,+\infty)$ (see Example \ref{ex:representations}), we see that \ref{item:cocycle3} implies
\ref{item:equality2}.
The equivalence of \ref{item:cocycle} and \ref{item:cocycle2} is obvious from the fact that the functions $z\mapsto B^0\map^*\bz \map(B)^{-z}\map(A)^{z}\jz$ and $z\mapsto B^{-z}A^{z}$ are both analytic on the whole complex plane.
Differentiating \ref{item:cocycle2} at $t=0$, we obtain \ref{item:logarithms}. A straightforward computation shows that \ref{item:logarithms} yields \ref{item:equality2} for $f(x):=x\log x$, that is, the equality for the standard relative entropy (note that the support of the representing measure for $x\log x$ is $(0,+\infty)$ by Example \ref{ex:representations}).
Hence, we have proved that \ref{item:existence of inverse}$\imp$\ref{item:existence of inverse2}$\imp$\ref{item:equality1}$\imp$\ref{item:equality2}$\iff$\ref{item:primitives}$\iff$\ref{item:cocycle}
$\iff$\ref{item:cocycle3}$\iff$\ref{item:cocycle2}$\iff$\ref{item:logarithms}.

Assume now that \ref{item:cocycle} holds.
In particular, the choice $z=0$ yields
\begin{equation}\label{support equality}
B^0\map^*\bz\map(A)^0\jz=A^0
\end{equation}
(recall that $A^0\le B^0$).
Since $\map$ is substochastic, we have $\map^*(Y^*Y)\ge\map^*(Y^*)\map^*(Y)\ge\map^*(Y^*)B^0\map^*(Y)$, and multiplying from both sides by $B^0$, we obtain that $\Psi(Y):=B^0\map^*(Y)B^0$, $Y\in\A_2$, is a Schwarz contraction. 
For $u_t:=\map(B)^{-it}\map(A)^{it}$ and $w_t:=B^{-it}A^{it}$,
we have
\begin{equation*}
u_tu_t^*=\map(B)^{-it}\map(A)^0\map(B)^{it},\ds\ds\ds
w_tw_t^*=B^{-it}A^0B^{it},\ds\ds\ds t\in\R.
\end{equation*}
Note that \ref{item:cocycle} says that $B^0\map^*(u_t)=w_t$, and hence
$\Psi(u_t)=w_t B^0=w_t$. Thus,
\begin{align*}
0&\le
\Tr B^{1/2}\bz\Psi(u_tu_t^*)-\Psi(u_t)\Psi(u_t^*) \jz B^{1/2}
=
\Tr B\map^*(u_tu_t^*)-\Tr Bw_tw_t^*\\
&=
\Tr \map(B)\map(B)^{-it}\map(A)^0\map(B)^{it}-\Tr BB^{-it}A^0B^{it}
=
\Tr \map(B)\map(A)^0-\Tr BA^0\\
&=
\Tr B\map^*(\map(A)^0)-\Tr BA^0
=
\Tr BA^0-\Tr BA^0
=0,
\end{align*}
where we used \eqref{support equality}.
Hence, $B^{1/2}\Psi(u_tu_t^*)B^{1/2}=B^{1/2}\Psi(u_t)\Psi(u_t^*)B^{1/2}$, and multiplying from both sides with $B^{-1/2}$, we obtain
$\Psi(u_tu_t^*)=\Psi(u_t)\Psi(u_t^*)$. 
Since $\Psi(u_t)\ne 0$, and $\Psi$ is a Schwarz contraction, this yields that $\snorm{\mapp}=1$ and $u_t\in\M_{\Psi}$.
Hence, by Lemma \ref{lemma:Schwarz}, $\Psi(u_tY)=\Psi(u_t)\Psi(Y)=w_t\map^*(Y)B^0$ for all $Y\in\A_2$ and $t\in\R$, i.e.,
\begin{equation*}
B^0\map^*\bz \map(B)^{-it}\map(A)^{it}Y\jz B^0=B^{-it}A^{it}\map^*(Y)B^0,\ds\ds\ds t\in\R,\,Y\in\A_2.
\end{equation*}
Note that the maps $z\mapsto B^0\map^*\bz \map(B)^{-z}\map(A)^{z}Y\jz B^0$ and
$z\mapsto B^{-z}A^{z}\map^*(Y)B^0$ are analytic on the whole complex plane and coincide on $i\R$ and thus they are equal for every $z\in\iC$. Choosing $z=1/2$ and $Y:=\map(A)^{1/2}\map(B)^{-1/2}$, we get 
\begin{align*}
B^0\map^*\bz \map(B)^{-1/2}\map(A)^{1/2}\map(A)^{1/2}\map(B)^{-1/2}\jz B^0
&=
B^{-1/2}A^{1/2}\map^*(\map(A)^{1/2}\map(B)^{-1/2})B^0\\
&=
B^{-1/2}A^{1/2}A^{1/2}B^{-1/2},
\end{align*}
where we used the adjoint of \ref{item:cocycle} with $z=1/2$. Multiplying from both sides by $B^{1/2}$, we obtain \ref{item:special inverse}.

Finally, assume that \ref{item:special inverse} holds, and hence
\begin{equation*}
\map_B^*(\map(A))=A,\ds\ds\ds
\map_B^*(\map(B))=B.
\end{equation*}
Note that 
$\map_B^*$ is not necessarily trace-preserving, as $(\map_B^*)^*(I_{1})=
\map_B(I_{1})=\map(B)^0$, which might be strictly smaller than $I_{2}$. However, 
if $\rho$ is a density operator on $\hil_1$ then the map
$X\mapsto\map_B(X)+(\Tr \rho X)(I_2-\map(B)^0)$ is obviously unital and hence its adjoint
$\mapp:\,\A_2\to\A_1,\,\mapp(Y)=\map_B^*(Y)+[\Tr(I_2-\map(B)^0)Y]\rho$ is trace-preserving. Moreover, $\mapp(\map(A))=\map_B^*(\map(A))$ and $\mapp(\map(B))=\map_B^*(\map(B))$, as one can easily verify.
Since $\mapp$ is obtained from $\map^*$ by composing it with completely positive maps and adding a completely positive map, it inherits the positivity of $\map^*$, i.e., 
if $\map$, and hence $\map^*$, is $n$-positive/completely positive then so is $\mapp$. In particular, if $\map$ is $2$-positive then $\mapp^*$ is a unital $2$-positive map and hence it is also a Schwarz contraction, i.e., $\mapp$ is stochastic. Thus
\ref{item:special inverse}$\imp$\ref{item:existence of inverse} holds in this case.
\end{proof}
\begin{lemma}\label{lemma:approximation}
If $f$ is a complex-valued function on finitely many points $\{x_i\}_{i\in I}\subset [0,+\infty)$ then for any pairwise different positive numbers $\{t_i\}_{i\in I}$, there exist complex numbers $\{c_i\}_{i\in I}$ such that $f(x_i)=\sum_{j\in I}c_j\frac{1}{x_i+t_j},\,i\in I$.
\end{lemma}
\begin{proof}
The matrix $C$ with entries $C_{ij}:=\frac{1}{x_i+t_j},\,i,j\in I$, is a Cauchy matrix which is invertible due to the assumptions that $x_i\ne x_j$ and $t_i\ne t_j$ for $i\ne j$. From this the statement follows.
\end{proof}

\begin{cor}
Assume that $\supp A_i\le\supp B_i,\,i=1,\ldots,r$, in the setting of Corollary 
\ref{cor:joint convexity}. Then equality holds in \eqref{joint convexity} if and only if
\begin{equation*}
p_iA_i=p_iB_i^{1/2}\bz\sum\nolimits_j p_jB_j\jz^{-1/2}\bz\sum\nolimits_j p_jA_j\jz\bz\sum\nolimits_j p_jB_j\jz^{-1/2} B_i^{1/2},\ds\ds\ds i=1,\ldots,r.
\end{equation*}
\end{cor}
\begin{proof}
It is immediate from writing out the equality $A=\map_B^*(\map(A))$ given in \ref{item:special inverse} in the setting of Corollary \ref{cor:joint convexity}.
\end{proof}

\begin{rem}
Note that if $\supp A\le\supp B$ and $\Tr\map(B)=\Tr B$ then for a linear function $f(x)=f(0)+ax$, the preservation of the $f$-divergence is automatic,
and has no implication on the reversibility of $\map$ on $\{A,B\}$. Indeed, 
we have 
$\Tr\map(A)=\Tr A$ due to Lemma \ref{lemma:support inequality}, and 
\begin{equation*}
\rsr{\map(A)}{\map(B)}{f}=f(0)\Tr \map(B)+a\Tr \map(A)=f(0)\Tr B+a\Tr A=\rsr{A}{B}{f}.
\end{equation*}

The $f$-divergence corresponding to the quadratic function $f_2(x):=x^2$ is 
$\rsr{A}{B}{f_2}=\Tr A^2B\inv$ (when $\supp A\le\supp B$). Preservation of the 
$f$-divergence by a stochastic map is not automatic in this case; however, it is not 
sufficient for the reversibility of the map, either.
Indeed, it was shown in Example 2.2 of \cite{JPP} that there exists a positive definite 
operator $D_{123}$ on a tripartite Hilbert space $\hil_1\otimes\hil_2\otimes\hil_3$, such 
that 
\begin{equation}\label{2-div counterexample}
D_{123}(\tau_1\otimes D_{23})\inv=(D_{12}\otimes\tau_3)(\tau_1\otimes D_2\otimes\tau_3)\inv,
\end{equation}
but
\begin{equation}\label{2-div counterexample2}
D_{123}^{it}(\tau_1\otimes D_{23})^{-it}\ne (D_{12}\otimes\tau_3)^{it}(\tau_1\otimes D_2\otimes\tau_3)^{-it} \ds\text{for some}\ds t\in\R,
\end{equation}
where $\tau_i:=\frac{1}{\dim\hil_i}I_i$, and
$D_{23}:=\Tr_{\hil_1}D_{123},\s D_{12}:=\Tr_{\hil_3}D_{123}$,
$D_2:=\Tr_{\hil_1\otimes\hil_3}D_{123}$.
Define $\hil:=\hil_1\otimes\hil_2\otimes\hil_3,\,A:=D_{123}$ and $B:=\tau_1\otimes D_{23}$. Let $\A_1:=\B(\hil)$,
$\A_2:=\B(\hil_1\otimes\hil_2)\otimes I_{3}$ and let $\map^*$ be the identical embedding of $\A_2$ into $\A_1$.
Then,
\eqref{2-div counterexample} reads as
\begin{equation*}
AB\inv=\map(A)\map(B)\inv.
\end{equation*}
Multiplying both sides by $A$ and taking the trace, we obtain
\begin{equation}\label{2-div counterexample3}
\Tr A^2B\inv=\Tr A\map(A)\map(B)\inv.
\end{equation}
Note that $\map$ is the orthogonal (with respect to the Hilbert-Schmidt inner product) projection from $\A_1$ onto $\A_2$, i.e., $\map$ is the conditional expectation onto $\A_2$ with respect to $\Tr$, and $\map(A)\map(B)\inv\in\A_2$. Hence, we have
$\Tr A\map(A)\map(B)\inv=\Tr\map(A)^2\map(B)\inv$.
Hence, \eqref{2-div counterexample3} can be rewritten as
\begin{equation*}
\rsr{A}{B}{f_2}=\Tr A^2B\inv=\Tr\map(A)^2\map(B)\inv=\rsr{\map(A)}{\map(B)}{f_2}.
\end{equation*}
However, \eqref{2-div counterexample2} tells that 
\begin{equation*}
A^{it}B^{-it}\ne \map^*\bz\map(A)^{it}\map(B)^{-it}\jz\ds\text{for some}\ds t\in\R,
\end{equation*}
and hence \ref{item:cocycle2} in Theorem \ref{thm:equality} is not satisfied. Since $\map$ is $2$-positive (actually, completely positive), it means that none of 
\ref{item:existence of inverse}--\ref{item:special inverse} of Theorem \ref{thm:equality}
are satisfied.
\end{rem}

\begin{rem}
It was shown in \cite{Csi} that,
in the classical setting, preservation of an $f$-divergence by $\map$ is equivalent to the 
reversibility condition \ref{item:special inverse} of Theorem \ref{thm:equality} whenever 
$f$ is strictly convex.
We reformulate the classical case in our setting in Appendix \ref{sec:classical}, and use 
the condition for equality to give a necessary and sufficient condition for the equality in 
the operator H\"older and inverse H\"older inequalities.
\end{rem}
\begin{rem}\label{rem:support}
The classical case suggests that the support condition \eqref{support condition} might be 
too restrictive in general. On the other hand, \cite{Jencova2} provides an example where 
the $f$-divergence corresponding to a function $f$ with $|\supp\mu_f|=1$ is preserved and yet the 
reversibility property \ref{item:special inverse} of Theorem \ref{thm:equality} fails to 
hold. This shows that the support condition \eqref{support condition} cannot be completely removed in general.
\end{rem}

\begin{rem}
Theorem \ref{thm:equality} holds also if we replace $\map$ and $\mapp$ with
co-(sub)stochastic maps, and change conditions \ref{item:cocycle}--\ref{item:cocycle2} to the following:
\begin{itemize}
\item[\ref{item:cocycle}$'$] $B^0\map^*(\map(A)^z\map(B)^{-z})=B^{-z}A^z$ for all $z\in\bC$.
\item[\ref{item:cocycle3}$'$] $B^0\map^*(\map(A)^\alpha\map(B)^{-\alpha})=B^{-\alpha}A^\alpha$ for some
$\alpha\in(0,2)\setminus\{1\}$.
\item[\ref{item:cocycle2}$'$] $B^0\map^*(\map(A)^{it}\map(B)^{-it})=B^{-it}A^{it}$ for all $t\in\bR$.
\end{itemize}
In the proof of \ref{item:primitives}$\imp$\ref{item:cocycle}$'$,
the previous equality $Vh(\tilde\Delta)\map(B)^{1/2}=h(\Delta)B^{1/2}$ in
\eqref{(5.5)} is replaced with
$$
\hat V h(\tilde\Delta)\map(B)^{1/2}=\bar h(\Delta)B^{1/2}
$$
due to the conjugate-linearity of $\hat V$, where $\hat V$ is given in \eqref{(4.22)}.
In the proof of \ref{item:cocycle}$'$ $\imp$\ref{item:special inverse}, let
$u_t:=\map(A)^{it}\map(B)^{-it}$ and $w_t:=B^{-it}A^{it}$; then
$$
u_t^*u_t=\map(B)^{-it}\map(A)^0\map(B)^{it},\quad
w_tw_t^*=B^{-it}A^0B^{-it},\qquad t\in\bR.
$$
Using that $\map$ is a co-Schwarz contraction, we have $\map(u_t^*u_t)=\map(u_t)\map(u_t^*)$.
From the multplicative domain for a co-Schwarz contraction, we have
$\map(Yu_t)=\map(u_t)\map(Y)=w_t\map^*(Y)B^0$ for all $Y\in\A_2$ and $t\in\bR$. The rest
of the proof is as before with $Y=\map(B)^{-1/2}\map(A)$.
The implication \ref{item:special inverse}$\imp$\ref{item:existence of inverse} holds also if we assume $\map$ to be $2$-copositive.
\end{rem}

\begin{rem}
Note that the assumption that $\map$ is substochastic guarantees that $(\map_B^*)^*=\map_B$ is a Schwarz map, which is also subunital. However, as Example \ref{ex:Schwarz maps}
shows, there exist subunital Schwarz maps that are not Schwarz contractions.
Even more, it was shown in \cite{Jencova2} that if $\map$ is not $2$-positive then there 
exists a positive invertible $B$ such that $\map_B$ is not a Schwarz contraction. 
To circumvent this problem, we assumed that $\map$ is $2$-positive in the proof of 
\ref{item:special inverse}$\imp$\ref{item:existence of inverse} of Theorem 
\ref{thm:equality}. 
Note on the other hand that the monotonicity inequality holds not only for substochastic 
maps but also for \ki{Schwarz decomposable maps}, i.e., for those maps that can be 
decomposed as a convex combination of a 
substochastic and a co-substochastic map; see Remark \ref{rem:decomposability}.
Hence, the implication \ref{item:special inverse}$\imp$\ref{item:equality1} might still hold 
even if $\map_B^*$ is not a substochastic map. It is easy to see that this is the case, for 
instance, if $\map$ is $2$-decomposable, i.e., it is the convex combination of two trace 
non-increasing maps, one being $2$-positive and the other a composition of a $2$-positive 
map with a transposition. It is an open question whether the Schwarz decomposability of $\map$ implies that $\map_B^*$ is Schwarz decomposable for every positive semidefinite $B$.
\end{rem}

\section{Distinguishability measures related to binary state discrimination}
\label{sec:Chernoff}
\setcounter{equation}{0}

Let $\A\subset\B(\hil)$ be a $C^*$-algebra, where $\hil$ is a finite-dimensional Hilbert space, and let $\S(\A)$ be the state space of $\A$, i.e., $\S(\A):=\{A\in\A_+\,:\,\Tr A=1\}$ is the set of density operators in $\A$.

\begin{defin}
For $A,B\in\A_+$, the \ki{Chernoff distance} $\chdist{A}{B}$ of $A$ and $B$ is defined as
\begin{equation}\label{Chernoff def}
\chdist{A}{B}:=
\sup_{0\le\alpha<1}\left\{(1-\alpha)\rsr{A}{B}{\alpha}\right\}
=
-\min_{0\le\alpha\le 1}\psif{A}{B}{\alpha},
\end{equation}
where $\rsr{A}{B}{\alpha}$ is the R\'enyi relative entropy defined in Example \ref{ex: f-div Renyi}, and
\begin{align}
\psif{A}{B}{\alpha}:=\log\Tr A^{\alpha}B^{1-\alpha},\ds\ds\ds\alpha\in\bR.
\end{align}
For every $r\in\bR$, we define the \ki{Hoeffding distance} $\hdist{A}{B}{r}$ of $A$ and $B$ as
\begin{equation}\label{Hoeffding def}
\hdist{A}{B}{r}:=
\sup_{0\le\alpha<1}\rsr{e^r A}{B}{\alpha}=
\sup_{0\le\alpha<1}\left\{-\frac{\alpha r}{1-\alpha}+\rsr{A}{B}{\alpha}\right\}
=\sup_{0\le \alpha<1}\frac{-\alpha r-\psif{A}{B}{\alpha}}{1-\alpha}.
\end{equation}
\end{defin}

\begin{rem}\label{rem:Hoeffding properties}
Note that 
\begin{equation}\label{Hoeffding def2}
\hdist{A}{B}{r}
=\sup_{s\ge 0}\{-sr-\psift{A}{B}{s}\},
\end{equation}
where
\begin{align*}
\psift{A}{B}{s}:=(1+s)\psif{A}{B}{s/(1+s)},\ds s\in[0,+\infty),\ds\ds\ds
\psift{A}{B}{s}:=+\infty,\ds s<0.
\end{align*}
For simplicity, we will use the notation $\psi(\alpha)=\psif{A}{B}{\alpha}$ and
$\tilde\psi(s):=\psift{A}{B}{s}$.
Let $\tilde\psi^*(r):=\sup_{s\in\bR}\{sr-\tilde\psi(s)\}$ be the
\ki{polar function}, or \ki{Legendre-Fenchel transform} of $\tilde\psi$ \cite{ET}.
By \eqref{Hoeffding def2},
$\hdist{\rho}{\sigma}{r}=\tilde\psi^*(-r),\,r\in\bR$.
It is easy to see (by computing its second derivative) that $\psi$ is convex, 
and hence so is $\tilde\psi$. Furthermore,
$\tilde\psi'(s)=\psi(s/(1+s))+\psi'(s/(1+s))/(1+s),\,s\in(0,+\infty)$, and
$\derright{\tilde\psi}(0)=\psi(0)+\psi'(0)$, where $\derright{\tilde\psi}(0)$ is the right derivative of $\tilde\psi$ at $0$. In particular, 
$\lim_{s\to+\infty}\tilde\psi'(s)=\psi(1)$. Hence,
\begin{equation*}
\hdist{A}{B}{r}=\tilde\psi^*(-r)=\begin{cases}
-\tilde\psi(0)=-\psi(0),&-r<\psi(0)+\psi'(0),\\
+\infty,&-r>\psi(1).
\end{cases}
\end{equation*}
It is easy to see that
\begin{align*}
\psi(0)&=-\rsr{A}{B}{0},\ds
\text{and if}\ds A^0\ge B^0\ds\text{then}\ds
\psi'(0)=-\sr{B}{A},\\
\psi(1)&=-\rsr{B}{A}{0},\ds
\text{and if}\ds A^0\le B^0\ds\text{then}\ds
\psi'(1)=\sr{A}{B}.
\end{align*}

Being a polar function, $\tilde\psi^*$ is convex, and hence so is
the function $r\mapsto\hdist{\rho}{\sigma}{r}$.
Moreover, $\tilde\psi$ is lower semicontinuous and thus the bipolar theorem (see, e.g., Proposition 4.1 in \cite{ET})
yields that $\tilde\psi$ is the polar function of its polar $\tilde\psi^*$. Hence,
for every $s\in[0,+\infty)$, we have
\begin{align*}
(1+s)\psi\bz\frac{s}{1+s}\jz=\tilde\psi(s)
&=\sup_{r\in\R}\{sr-\tilde\psi^*(r)\}
=
\sup_{\psi(0)+\psi'(0)\le -r\le \psi(1)}\{-rs-\tilde\psi^*(-r)\}.
\end{align*}
Replacing $s$ with $\alpha/(1-\alpha)$, we finally get that for every $\alpha\in[0,1)$,
\begin{equation}\label{Hoeffding inversion}
-\rsr{A}{B}{\alpha}=\frac{\psi(\alpha)}{1-\alpha}
=\sup_{r\in\bR}\left\{\frac{-r\alpha}{1-\alpha}-\hdist{A}{B}{r}\right\}
=\sup_{-\psi(1)\le r\le -\psi(0)-\psi'(0)}
\left\{\frac{-r\alpha}{1-\alpha}-\hdist{A}{B}{r}\right\}.
\end{equation}
That is, the R\'enyi $\alpha$-relative entropies with parameter $\alpha\in[0,1)$ and the Hoeffding 
distances mutually determine each other.

If $\Tr A\le 1$ then $\psi(1)=\log\Tr AB^0\le 0$, and hence the optimization is over non-negative
values of $r$ in the last formula of \eqref{Hoeffding inversion}.
Thus, $\alpha\mapsto \rsr{A}{B}{\alpha}$ is monotonic increasing on $[0,1)$ and hence
\begin{equation*}
\hdist{A}{B}{0}=\lim_{\alpha\nearrow 1}\rsr{A}{B}{\alpha}=:\rsr{A}{B}{1}.
\end{equation*}
Note that $\tilde\psi^*$ is lower semicontinuous (see, e.g., Proposition 4.1 and Corollary 4.1 in \cite{ET}), and hence
$\tilde\psi^*(0)\le\liminf_{r\searrow 0}\tilde\psi^*(-r)$. On the other hand,
it is obvious from the definition that $r\mapsto\hdist{A}{B}{r}=\tilde\psi^*(-r)$ is monotonic decreasing on $\bR$, and hence we finally obtain
\begin{equation}\label{0 Hoeffding}
\lim_{r\searrow 0}\hdist{A}{B}{r}=\lim_{r\searrow 0}\tilde\psi^*(-r)=\tilde\psi^*(0)=
\hdist{A}{B}{0}=\rsr{A}{B}{1}.
\end{equation}
Finally, it is easy to verify that 
\begin{equation}\label{0 Hoeffding2}
\rsr{A}{B}{1}=\sr{A}{B}\ds\ds\ds\text{if}\ds\ds\ds \Tr A= 1.
\end{equation}
\end{rem}
\bigskip

The importance of the above measures comes from the problem of binary state discrimination, that 
we briefly describe below. Assume that 
we have several identical copies of a quantum system, and we know that either all of them 
are in a state described by a density operator $\rho$, or all of them are in a state 
described by a density operator $\sigma$. We assume that the system's Hilbert space $\hil$ 
is finite-dimensional. Our goal is to give a good guess on the true 
state of the system, based on the outcome of a binary POVM measurement $(T,I-T)$ on a fixed 
number (say $n$) copies, where $T$ is an operator on $\hil^{\otimes n}$ satisfying $0\le 
T\le I$. If the outcome corresponding to $T$ happens then we conclude that the state of the 
system is $\rho$, and an error occurs if the true state is $\sigma$, which has 
probability $\beta_n(T):=\Tr\sigma^{\otimes n} T$. 
Similarly, the outcome corresponding to $I-T$ yields the guess $\sigma$ for the true state, 
and the probability of error in this case is
$\alpha_n(T):=\Tr\rho^{\otimes n}(I-T)$. If, moreover, there are prior probabilities $p$ and 
$1-p$ assigned to $\rho$ and $\sigma$, then the optimal Bayesian error probability is given 
by 
\begin{align*}
P_{n,p}:=\min_{0\le T\le I}\{p\alpha_n(T)+(1-p)\beta_n(T)\}
=(1-\norm{p\rho^{\otimes n}-(1-p)\sigma^{\otimes n}})/2,
\end{align*}
where the minimum is reached at $T=\{p\rho^{\otimes n}-(1-p)\sigma^{\otimes n}>0\}$, the spectral projection corresponding to the positive part of the spectrum of $p\rho^{\otimes n}-(1-p)\sigma^{\otimes n}$.
For every $p\in(0,1)$, let
\begin{equation}
\tdist{\rho^{\otimes n}}{\sigma^{\otimes n}}{p}:=
\begin{cases}
-\log\frac{1}{2p}(1-\norm{p\rho^{\otimes n}-(1-p)\sigma^{\otimes n}}_1)=-\log\frac{1}{p}P_{n,p},& 0<p\le 1/2,\\
-\log\frac{1}{2(1-p)}(1-\norm{p\rho^{\otimes n}-(1-p)\sigma^{\otimes n}}_1)=-\log\frac{1}{1-p}P_{n,p},& 1/2<p<1.
\end{cases}\label{Tp def}
\end{equation}
The theorem for the \ki{quantum Chernoff bound} \cite{Aud,NSz} says that, as the number of copies $n$ tends to infinity, the error probabilities $P_{n,p}$ decay exponentially, and the rate of the decay is given by the Chernoff distance. More formally, 
\begin{equation}\label{Chernoff limit}
-\lim_{n\to\infty}(1/n)\log P_{n,p}=
\lim_{n\to\infty}(1/n)\tdist{\rho^{\otimes n}}{\sigma^{\otimes n}}{p}=\chdist{\rho}{\sigma},\ds\ds\ds p\in(0,1).
\end{equation}
In the asymmetric setting of the \ki{quantum Hoeffding bound}, the error probabilities $\alpha_n$ are required to be exponentially small, and $\beta_n$ is optimized under this constraint, i.e., one is interested in the quantities
\begin{equation*}
\beta_{n,r}:=\min\{\beta_n(T)\,:\,\alpha_n(T)\le e^{-nr},\,T\in\B(\hil^{\otimes n}),\,0\le T\le I\},
\end{equation*}
where $r$ is some fixed positive number.
The theorem for the quantum Hoeffding bound \cite{Hayashi,Nagaoka} says that, for every $r>0$, the error probabilities $\beta_{n,r}$ decay exponentially fast as $n$ goes to infinity, and the decay rate is given by the Hoeffding distance with parameter $r$. Moreover, if $\supp\rho\le\supp\sigma$, then for every $r> 0$ we have a real number $a_r$ such that
\cite{HMO,Nagaoka}
\begin{equation}\label{Hoeffding limit}
-\lim_{n\to\infty}(1/n)\log\beta_{n,r}=
\lim_{n\to\infty}(1/n)\tdist{\rho^{\otimes n}}{\sigma^{\otimes n}}{\frac{e^{-na_r}}{1+e^{-na_r}}}=\hdist{\rho}{\sigma}{r}.
\end{equation}

Note that for density operators $\rho$ and $\sigma$, $\psi(\alpha|\rho\|\sigma)=\log\Tr\rho^\alpha\sigma^{1-\alpha}\le 0$ for every $\alpha\in[0,1]$ due to H\"older's inequality \eqref{Holder1}. Hence, $\chdist{\rho}{\sigma}\ge 0$, and $\chdist{\rho}{\sigma}= 0$ if and only if equality holds in H\"older's inequality, which is equivalent to $\rho=\sigma$. Similarly, $\hdist{\rho}{\sigma}{r}\ge 0$ for every $r\in\bR$, and
$\hdist{\rho}{\sigma}{r}=0$ if and only if 
$\rho=\sigma$, or 
$\supp \rho\ge\supp\sigma$ and 
$r\ge \sr{\sigma}{\rho}$. 

\begin{prop}\label{prop:ch monotonicity}
Let $A,B\in\A_{1,+}$ and let $\map:\,\A_1\to\A_2$ be a substochastic map such that
$\Tr\map(B)=\Tr B$. Then
\begin{equation}\label{Chernoff monotonicity}
\chdist{\map(A)}{\map(B)}\le\chdist{A}{B}\ds\ds\text{and}\ds\ds
\hdist{\map(A)}{\map(B)}{r}\le\hdist{A}{B}{r},\ds r\in\bR.
\end{equation}
If there exists a substochastic map $\mapp:\,\A_2\to\A_1$ such that 
$\mapp(\map(A))=A$ and $\mapp(\map(B))=B$ then the inequalities in 
\eqref{Chernoff monotonicity} hold with equality.
\end{prop}
\begin{proof}
By Example \ref{ex:Renyi1}, $\rsr{\map(A)}{\map(B)}{\alpha}\le\rsr{A}{B}{\alpha}$ 
for every $\alpha\in[0,1)$, and equality holds for every $\alpha\in[0,1)$
if there exists a substochastic map $\mapp:\,\A_2\to\A_1$ such that 
$\mapp(\map(A))=A$ and $\mapp(\map(B))=B$, due to Theorem \ref{thm:equality}. The assertion then follows immediately from 
the definitions \eqref{Chernoff def} and \eqref{Hoeffding def}.
\end{proof}

Our goal now is to give the converse of the above proposition, i.e., to show that 
equality in the inequalities of \eqref{Chernoff monotonicity} yields the existence 
of a substochastic map $\mapp:\,\A_2\to\A_1$ such that 
$\mapp(\map(A))=A$ and $\mapp(\map(B))=B$. This would be immediate from 
Theorem \ref{thm:equality} if the Chernoff and the Hoeffding distances could be represented as $f$-divergences (at least when $\map$ is also assumed to be $2$-positive). However, no such representation is possible, as is shown in the following proposition:

\begin{prop}\label{prop:Chernoff not fdiv}
The Chernoff and the Hoeffding distances cannot be represented as $f$-divergences on the state space of any non-trivial finite-dimensional $C^*$-algebra.
\end{prop}
\begin{proof}
Let $\A\subset\B(\hil)$ where $\dim\hil\ge 2$, and let $e_1,e_2$ be orthonormal vectors in 
$\hil$ such that $\pr{e_j}\in\A,\,j=1,2$. Define $\rho:=\pr{e_1},\,\sigma_p:=p\pr{e_1}+(1-p)\pr{e_2},\,p\in(0,1)$. One can 
easily check that $\chdist{\rho}{\sigma_p}=\hdist{\rho}{\sigma_p}{r}=-\log p$ for every 
$r> 0$, while $\rsr{\rho}{\sigma_p}{f}=pf(1/p)+(1-p)f(0)$ for any function $f$ on 
$[0,+\infty)$. Hence, if any of the above measures can be represented as an $f$-divergence, 
then we have
$pf(1/p)+(1-p)f(0)=-\log p$ for the representing function $f$, and taking the limit 
$p\searrow 0$ yields $\of=+\infty$. In particular, $\rsr{\sigma_p}{\rho}{f}=+\infty$ for 
every $p\in(0,1)$. On the other hand, $\chdist{\sigma_p}{\rho}=-\log p$ and 
$\hdist{\sigma_p}{\rho}{r}=0$ if $r\ge-\log p$. That is, $\chdist{\sigma_p}{\rho}$ is finite 
for every $p\in(0,1)$ and for every $r>0$ there exists a $p\in(0,1)$ such that 
$\hdist{\sigma_p}{\rho}{r}$ is finite.
\end{proof}

Note, however, that for the applications of Theorems \ref{thm:monotonicity for F} and 
\ref{thm:equality}, it is sufficient to have a more general representability. 
Indeed, let $\A$ be a finite-dimensional $C^*$-algebra and $D:\,\S(\A)\times\S(\A)\to\bR$. 
We say that
\ki{$D$ is a monotone function of an $f$-divergence} on the state space of $\A$ if there 
exists an operator convex function $f:\,[0,+\infty)\to\bR$ and a strictly 
monotonic increasing function 
$g:\,\{\rsr{\rho}{\sigma}{f}\,:\,\rho,\sigma\in\S(\A)\}\to\bR\cup\{\pm\infty\}$ such that
\begin{equation*}
\dist{\rho}{\sigma}=g\bz\rsr{\rho}{\sigma}{f}\jz,\ds\ds\ds\rho,\sigma\in\S(\A).
\end{equation*}
Obviously, if $D$ is a monotone function of an $f$-divergence then it is monotonic non-increasing under stochastic maps due to Theorem \ref{thm:monotonicity for F}.
Moreover, if 
$\dist{\map(\rho)}{\map(\sigma)}=\dist{\rho}{\sigma}$ for some stochastic map $\map$ and 
$\rho,\sigma\in\S(\A)$ such that $\supp\rho\le\supp\sigma$, and
the representing function $f$ satisfies $|\supp \mu_f|\ge|\spect(L_\rho R_{\sigma\inv})\cup\spect(L_{\map(\rho)} R_{\map(\sigma)\inv})|$ then $\map_\sigma^*(\map(\rho))=\rho$, due to \ref{item:equality2} of Theorem \ref{thm:equality}.
For instance, the R\'enyi $\alpha$-relative entropy is a monotone function of the 
$\tilde f_\alpha$-divergence with $g(x):=\frac{1}{\alpha-1}\log\sgn(\alpha-1)x$, for every 
$\alpha\in[0,2]\setminus\{1\}$. However, the same argument as in Proposition 
\ref{prop:Chernoff not fdiv} yields that none of the R\'enyi relative entropies with 
parameter $\alpha\in(0,1)$ can be represented as $f$-divergences. 

\begin{prop}
For any $r\in(0,+\infty)$ and any non-trivial $C^*$-algebra $\A$, the Hoeffding distance $H_r$ cannot be represented on the state space of $\A$ as a monotone function of an $f$-divergence
with an operator convex function $f$ on $[0,+\infty)$ such that $|\supp\mu_f|\ge 6$.
\end{prop}
\begin{proof}
Let $\A\subset\B(\hil)$ be a $C^*$-algebra and let $e_1,e_2$ be orthogonal vectors in $\hil$ such that $\pr{e_1},\pr{e_2}\in\A$. Choose $p,q\in(0,1)$ such that $p\ne q$ and $q\log\frac{q}{p}+(1-q)\log\frac{1-q}{1-p}<r$, and define $\rho:=p\pr{e_1}+(1-p)\pr{e_2}$ and 
$\sigma:=q\pr{e_1}+(1-q)\pr{e_2}$. Then $\psi(0|\rho\|\sigma)=0$ and 
$-\psi(0|\rho\|\sigma)-\psi'(0|\rho\|\sigma)=\sr{\sigma}{\rho}=q\log\frac{q}{p}+(1-q)\log\frac{1-q}{1-p}<r$, and 
hence $\hdist{\rho}{\sigma}{r}=-\psi(0|\rho\|\sigma)=0$. Define $\map:\,\A\to\A,\,\map(X):=(\Tr X)I/(\dim\hil)$. Then 
$\map$ is completely positive and trace-preserving, $\map(\rho)=\map(\sigma)$, and hence 
$\hdist{\map(\rho)}{\map(\sigma)}{r}=0=\hdist{\rho}{\sigma}{r}$. 
Note that $|\spect\bz L_{\rho}R_{\sigma\inv}\jz|\le 5$ and $|\spect\bz L_{\map(\rho)}R_{\map(\sigma)\inv}\jz|=1$.
If we had 
$\hdist{\rho}{\sigma}{r}=g\bz\rsr{\rho}{\sigma}{f}\jz$ and
$\hdist{\map(\rho)}{\map(\sigma)}{r}=g\bz\rsr{\map(\rho)}{\map(\sigma)}{f}\jz$
for some strictly monotone $g$ and an operator convex $f$ on $[0,+\infty)$ such that 
$|\supp\mu_f|\ge 6$ then Theorem \ref{thm:equality} would yield
$\map_\sigma^*(\map(\rho))=\rho$. However, $\map(\rho)=\map(\sigma)$ and hence
$\map_\sigma^*(\map(\rho))=\map_\sigma^*(\map(\sigma))=\sigma\ne\rho$.
\end{proof}

The above proposition also shows that the preservation of a Hoeffding distance of a pair $(\rho,\sigma)$ by a stochastic map for a given parameter $r$ might not be sufficient for the reversibility of $\map$ on $\{\rho,\sigma\}$ in the sense of Theorem \ref{thm:equality}; the reason for this in the above proof is that 
the Hoeffding distance might be equal to zero even for non-equal states. 
The Chernoff distance, on the other hand, is always strictly positive for unequal states; yet
the following example shows that the preservation of the Chernoff distance 
is not sufficient for reversibility in general, either.

\begin{ex}\label{ex:Chernoff}
Let $\hil:=\bC^3$ and let $\A$ be the commutative $C^*$-algebra of operators on $\hil$ that are diagonal in some fixed basis $e_1,e_2,e_3$. Let
$\rho:=(2/3)\pr{e_1}+(1/3)\pr{e_2}$, $\sigma:=(1/6)\pr{e_1}+(1/3)\pr{e_2}+(1/2)\pr{e_3}$, and define $\map:\,\A\to\A$ as 
\begin{equation*}
\map(\pr{e_1}):=\map(\pr{e_2}):=\pr{e_1},\ds\ds\ds\map(\pr{e_3}):=\pr{e_3}.
\end{equation*}
Then $\map$ is completely positive and trace-preserving, and we have $\map(\rho)=\pr{e_1}$, $\map(\sigma)=(1/2)\pr{e_1}+(1/2)\pr{e_3}$. 
For every $\alpha\in\bR$, we have
$\Tr\rho^\alpha\sigma^{1-\alpha}=\frac{2+4^\alpha}{6}$ and 
$\Tr\map(\rho)^\alpha\map(\sigma)^{1-\alpha}=2^{\alpha-1}$, and hence
\begin{align*}
\chdist{\map(\rho)}{\map(\sigma)}&=-\log\psif{\map(\rho)}{\map(\sigma)}{0}=\rsr{\map(\rho)}{\map(\sigma)}{0}=\log2
=\rsr{\rho}{\sigma}{0}\\
&=-\log\psif{\rho}{\sigma}{0}=\chdist{\rho}{\sigma}.
\end{align*}
On the other hand, it is easy to see that $\map_\sigma^*(\map(\rho))=(1/3)\pr{e_1}+(2/3)\pr{e_2}\ne\rho$, and therefore \ref{item:special inverse} of Theorem \ref{thm:equality} does not hold, and hence
$\map$ is not reversible on the pair $\{\rho,\sigma\}$.
\end{ex}
\begin{rem}
Note that in the setting of Theorem \ref{thm:equality}, if $\map$ is $2$-positive and 
$\rsr{\map(A)}{\map(B)}{\alpha}=\rsr{A}{B}{\alpha}$ for some $\alpha\in(0,1)$ then 
$\map_B^*(\map(A))=A$, i.e., the preservation of a R\'enyi $\alpha$-relative entropy with 
some $\alpha\in(0,1)$ is sufficient for the reversibility of $\map$ on $\{A,B\}$. The above 
example shows that the same is not true for the $0$-relative entropy.
\end{rem}

\begin{cor}
Let $\A$ be a $C^*$-algebra which contains at least $3$ orthogonal non-zero projections. Then the Chernoff distance cannot be represented on its state space as a monotone function of an $f$-divergence with an 
operator convex $f$ on $[0,+\infty)$ such that $|\supp\mu_f|\ge 6$. 
\end{cor}
\begin{proof}
Immediate from Example \ref{ex:Chernoff}.
\end{proof}

After the above preparation, we are ready to prove the analogue of Theorem \ref{thm:equality} for the preservation of the Chernoff and the Hoeffding distances. The preservation of the Chernoff distance was already treated in the proof of Theorem 6 in \cite{Jencova} in the case where both operators are invertible density operators and the substochastic map is the trace-preserving conditional expectation onto a subalgebra. We use essentially the same proof to treat the general case below.

\begin{thm}\label{thm:equality2}
Let $A,B\in\A_{1,+}$ be such that $\supp A\le\supp B$,
let $\map:\A_1\to\A_2$ be a substochastic map such that $\Tr\map(B)=\Tr B$, 
and assume that 
\ref{item:ch preservation} or \ref{item:Hoeffding pres} below holds:

\begin{enumerate}

\item\label{item:ch preservation}
$\chdist{\map(A)}{\map(B)}\ne\rsr{\map(A)}{\map(B)}{0},\s
\chdist{\map(A)}{\map(B)}\ne\rsr{\map(B)}{\map(A)}{0}$, and
\begin{equation*}
\chdist{\map(A)}{\map(B)}=\chdist{A}{B}.
\end{equation*}

\item\label{item:Hoeffding pres}
For some  $r\in(-\psif{\map(A)}{\map(B)}{1},-\psif{\map(A)}{\map(B)}{0}-\psi'(0|\map(A)\|\map(B))$,
\begin{equation}\label{Hoeffding equality}
\hdist{\map(A)}{\map(B)}{r}=\hdist{A}{B}{r}.
\end{equation}
\end{enumerate}
Then $\map_B^*(\map(A))=A$, and if $\map$ is $2$-positive then there exists a stochastic map $\mapp:\,\A_2\to\A_1$ such that 
$\mapp(\map(A))=A$ and $\mapp(\map(B))=B$.
\end{thm}
\begin{proof}
Assume first that \ref{item:ch preservation} holds. Due to the assumptions 
$\chdist{\map(A)}{\map(B)}\ne\rsr{\map(A)}{\map(B)}{0}=-\psif{\map(A)}{\map(B)}{0}$,
$\chdist{\map(A)}{\map(B)}\ne\rsr{\map(B)}{\map(A)}{0}=-\psif{\map(A)}{\map(B)}{1}$, and the 
definition 
\eqref{Chernoff def} of the Chernoff distance, there exists an $\alpha^*\in(0,1)$ such that 
$\chdist{\map(A)}{\map(B)}=-\psif{\map(A)}{\map(B)}{\alpha^*}$. Using the monotonicity 
relation \eqref{Renyi monotonicity 01}, we get 
\begin{align*}
\chdist{\map(A)}{\map(B)}&=-\log\Tr\map(A)^{\alpha^*}\map(B)^{1-\alpha^*}
\le
-\log\Tr A^{\alpha^*}B^{1-\alpha^*}\le 
\chdist{A}{B}=\chdist{\map(A)}{\map(B)}.
\end{align*}
Hence, $\Tr\map(A)^{\alpha^*}\map(B)^{1-\alpha^*}=\Tr A^{\alpha^*}B^{1-\alpha^*}$, which yields $\map_B^*(\map(A))=A$ due to \ref{item:equality2} of Theorem \ref{thm:equality}.

Assume next that \eqref{Hoeffding equality} holds for some
$r\in(-\psif{\map(A)}{\map(B)}{1},-\psif{\map(A)}{\map(B)}{0}-\psi'(0|\map(A)\|\map(B))$. Then there exists an $s^*\in(0,+\infty)$ such that 
$\hdist{\map(A)}{\map(B)}{r}=-s^*r-\tilde\psi(s^*|\map(A)\|\map(B))$ (see Remark \ref{rem:Hoeffding properties}). Thus, $\hdist{\map(A)}{\map(B)}{r}=
-\alpha^* r/(1-\alpha^*)+\rsr{\map(A)}{\map(B)}{\alpha^*}$, where $\alpha^*:=\frac{s^*}{1+s^*}\in(0,1)$. Using the monotonicity \eqref{Renyi monotonicity}, we obtain
\begin{align*}
\hdist{\map(A)}{\map(B)}{r}&=
-\alpha^* r/(1-\alpha^*)+\rsr{\map(A)}{\map(B)}{\alpha^*}\\
&\le
-\alpha^* r/(1-\alpha^*)+\rsr{A}{B}{\alpha^*}
\le
\hdist{A}{B}{r}=\hdist{\map(A)}{\map(B)}{r}.
\end{align*}
Hence, $\Tr\map(A)^{\alpha^*}\map(B)^{1-\alpha^*}=\Tr A^{\alpha^*}B^{1-\alpha^*}$, which yields $\map_B^*(\map(A))=A$ due to \ref{item:equality2} of Theorem \ref{thm:equality}.

Finally, if $\map$ is $2$-positive then $\map_B^*(\map(A))=A$ yields the existence of $\mapp$ in the last assertion the same way as in the proof of 
\ref{item:special inverse}$\imp$\ref{item:existence of inverse} in Theorem \ref{thm:equality}.
\end{proof}

\begin{cor}\label{cor:equality3}
Assume in the setting of Theorem \ref{thm:equality2} that $\supp A=\supp B$ and $\Tr A=\Tr B$. If $\chdist{\map(A)}{\map(B)}=\chdist{A}{B}$ then $\map_B^*(\map(A))=A$.
\end{cor}
\begin{proof}
Let $\psi(\alpha):=\psif{\map(A)}{\map(B)}{\alpha},\,\alpha\in\bR$.
By the assumptions, we have $\supp\map(A)=\supp\map(B)$ and $\Tr\map(A)=\Tr\map(B)$, and hence
$\psi(0)=\psi(1)$. Since $\psi$ is convex, there are two possibilities: either $\psi$ is constant, or the minimum of $\psi$ on $[0,1]$ is attained at some $\alpha^*\in(0,1)$.
In the latter case we have $\chdist{\map(A)}{\map(B)}\ne\rsr{\map(A)}{\map(B)}{0},\s
\chdist{\map(A)}{\map(B)}\ne\rsr{\map(B)}{\map(A)}{0}$, and hence the assertion follows due to 
Theorem \ref{thm:equality2}.
If $\psi$ is constant then we have 
$\Tr\map(A)^\alpha\map(B)^{1-\alpha}=e^{\psi(\alpha)}=e^{\psi(1)}=\Tr\map(A)=(\Tr\map(A))^\alpha(\Tr\map(B))^{1-\alpha}$ for every $\alpha\in[0,1]$, and the equality case in H\"older's inequality yields that 
$\map(A)$ is constant multiple of $\map(B)$ (see Corollary \ref{cor:Holder}).
Since $\Tr\map(A)=\Tr\map(B)$, this yields that $\map(A)=\map(B)$. Similarly,
\begin{align*}
-\min_{0\le\alpha\le 1}\psif{A}{B}{\alpha}&=\chdist{A}{B}=\chdist{\map(A)}{\map(B)}
=
-\log\Tr\map(A)=-\log\Tr A=-\psif{A}{B}{0},
\end{align*}
and since $\Tr A=\Tr B$, we also have $-\log\Tr A=-\log\Tr B=-\psif{A}{B}{1}$. Hence,
$\alpha\mapsto\psif{A}{B}{\alpha}$ is constant on $[0,1]$, and the same argument as above yields that $A=B$. Therefore, 
$\map_B^*(\map(A))=\map_B^*(\map(B))=B=A$.
\end{proof}

\begin{rem}
Note that the interval 
$(-\psif{\map(A)}{\map(B)}{1},-\psif{\map(A)}{\map(B)}{0}-\psi'(0|\map(A)\|\map(B))$
in \ref{item:Hoeffding pres} of Theorem \ref{thm:equality2} might be empty; this happens if 
and only if $\alpha\mapsto\psif{\map(A)}{\map(B)}{\alpha}$ is constant. A characterization of 
this situation was given in Lemma 3.2 of \cite{HMO}.
\end{rem}

\section{Error correction}\label{sec:error correction}
\setcounter{equation}{0}

Noise in quantum mechanics is usually modeled by completely positive trace non-increasing 
maps. The aim of error correction is, given a noise operation $\map$, to identify a 
subset $\C$ of the state space (called the code) and a quantum operation $\Psi$ such that 
it reverses the action of the noise on the code, i.e., $\Psi(\map(\rho))=\rho,\,\rho\in\C$.
It was first noticed in \cite{Petz2} that the preservation of certain distinguishability 
measures 
of two states by the noise operation is a sufficient condition for correctability of the 
noise on those two states. This result was later extended to general families of states in 
\cite{JP,JP2}. The measures considered in these papers were the R\'enyi relative entropies 
and the standard relative entropy. Recently, the same problem was considered in \cite{BNPV} 
using 
the measures $T_p$ given in \eqref{Tp def}, and similar results were found, although only under some extra technical
conditions. Below we summarize these results and extend them to a wide class of measures, 
based on Theorem \ref{thm:equality}. 

Let $\A_i$ be a $C^*$-algebra on $\hil_i$ for $i=1,2$, and let $\S(\A_i)$ denote the set of density operators in $\A_i$.
For a non-empty set $\C\subset\S(\A_1)$, let $\co\C$ denote the closed convex hull of $\C$,
and let 
$\supp\C$ be the supremum of the supports of all states in $\C$. Note that there exists a state
$\sigma\in\co\C$ such that $\supp\sigma=\supp\C$. We introduce the notation
$d^2:=(\dim\hil_1)^2+(\dim\hil_2)^2$.
Note that if $X\in\A_1$ and $\map:\,\A_1\to\A_2$ is a trace non-increasing positive map then
\begin{align*}
\|\Phi(X)\|_1&=\max\{\Tr\Phi(X)S\,:\,S\in\A_2\ \mbox{self-adjoint},\,-I_2\le S\le I_2\} \\
&=\max\{\Tr X\Phi^*(S)\,:\,S\in\A_2\ \mbox{self-adjoint},\,-I_2\le S\le I_2\} \\
&\le\max\{\Tr XR\,:\,R\in\A_1\ \mbox{self-adjoint},\,-I_1\le R\le I_1\}=\|X\|_1,
\end{align*}
which in particular yields that 
the measures $T_p$ are monotonic non-increasing under substochastic maps.

\begin{thm}\label{thm:error correction}
Let $\map:\,\A_1\to\A_2$ be a trace-preserving $2$-positive map, and let $\C\subset\S(\A_1)$ be a non-empty set of states. The following are equivalent:
\begin{enumerate}
\item\label{item:ec existence of inverse}
There exists a stochastic map $\mapp:\A_2\to\A_1$ such that
for every $\rho\in\co\C$,
\begin{equation}\label{ec inverse relation}
\mapp(\map(\rho))=\rho.
\end{equation}
\item\label{item:ec equality1}
For every operator convex function $f$ on $[0,+\infty)$, and every $\rho,\sigma\in\co\C$,
\begin{equation}\label{preserved measures1}
\rsr{\map(\rho)}{\map(\sigma)}{f}=\rsr{\rho}{\sigma}{f}.
\end{equation}
\item\label{item:ec equality2}
The equality \eqref{preserved measures1} holds for 
every $\rho\in\C$ and for some $\sigma\in\S(\A_1)$ such that $\supp\sigma\ge\supp\C$, and
some operator convex $f$ on $[0,+\infty)$ such that $|\supp\mu_f|\ge d^2$.
\item\label{item:ec primitives}
$\rsr{\map(\rho)}{\map(\sigma)}{\vfi_t}=\rsr{\rho}{\sigma}{\vfi_t}$
for every $\rho\in\C$ and for some $\sigma\in\S(\A_1)$ such that $\supp\sigma\ge\supp\C$, and a set $T$ of $t$'s such that $|T|\ge d^2$.
\item\label{item:Hoeffding preservation}
For every $\rho,\sigma\in\co\C$ and every $r\in\bR$,
\begin{equation}\label{Hoeffding preservation}
\hdist{\map(\rho)}{\map(\sigma)}{r}=\hdist{\rho}{\sigma}{r}.
\end{equation}
\item\label{item:Hoeffding preservation2}
The equality in \eqref{Hoeffding preservation} holds for every $\rho\in\C$ and for some $\sigma\in\S(\A_1)$ such that $\supp\sigma\ge\supp\C$, and for every $r\in (0,\delta)$ for some $\delta>0$.
\item\label{item:ec special inverse}
For every $\rho\in\co\C$ and every $\sigma\in\co\C$ such that $\supp\sigma=\supp\C$, 
\begin{equation}\label{special inverse map}
\map^*_\sigma(\map(\rho))=\rho.
\end{equation}
\item\label{item:ec special inverse2}
The equality \eqref{special inverse map} holds for every $\rho\in\C$ and some $\sigma\in\S(\A_1)$.
\item\label{item:decomposition}
There exist decompositions $\supp\C=\bigoplus_{k=1}^r\hil_{1,k,L}\otimes\hil_{1,k,R}$ and
$\supp\map(\C)=\bigoplus_{k=1}^r\hil_{2,k,L}\otimes\hil_{2,k,R}$,
invertible density operators $\omega_k$ on $\hil_{1,k,R}$ and $\tilde\omega_k$ on $\hil_{2,k,R}$, and unitaries
$U_k:\,\hil_{1,k,L}\to \hil_{2,k,L}$, $k=1,\ldots,r$, such that every $\rho\in\C$ can be written in the form
\begin{equation*}
\rho=\bigoplus_{k=1}^r p_k\rho_{k,L}\otimes\omega_k
\end{equation*}
with some density operators $\rho_{k,L}$ on $\hil_{1,k,L}$ and probability distribution $\{p_k\}_{k=1}^r$, and
\begin{equation*}
\map(A\otimes\omega_k)=U_kAU_k^*\otimes \tilde\omega_k,\ds\ds\ds A\in\B(\hil_{1,k,L}).
\end{equation*}
\end{enumerate}

Moreover, if $\map$ is $n$-positive/completely positive then $\mapp$ in \ref{item:ec existence of inverse} can also be chosen to be $n$-positive/completely positive.
The implications \ref{item:ec existence of inverse}$\imp$\ref{item:ec equality1}$\imp$
\ref{item:ec equality2}$\imp$\ref{item:ec primitives}$\imp$\ref{item:ec special inverse2} hold also if we only assume $\map$ to be substochastic.

Furthermore, criterion \ref{item:preservation of Tp} below is sufficient for 
\ref{item:ec existence of inverse}--\ref{item:ec special inverse2} to hold, and it is also necessary if $\map$ is completely positive.
\begin{enumerate}
\setcounter{enumi}{9}
\item\label{item:preservation of Tp}
For every $\rho\in\C$, every $p\in(0,1)$, every $n\in\N$, and for some $\sigma\in\S(\A_1)$ such that $\supp\sigma\ge\supp\C$,
\begin{equation}\label{preservation of Tp}
\tdist{\map^{\otimes n}(\rho^{\otimes n})}{\map^{\otimes n}(\sigma^{\otimes n})}{p}=
\tdist{\rho^{\otimes n}}{\sigma^{\otimes n}}{p}.
\end{equation}
\end{enumerate}
\end{thm}
\begin{proof}
The implications \ref{item:ec existence of inverse}$\imp$\ref{item:ec equality1}$\imp$
\ref{item:ec equality2}$\imp$\ref{item:ec primitives}$\imp$\ref{item:ec special inverse2} follow immediately from
Theorem \ref{thm:equality} under the condition that $\map$ is substochastic (note that in the implication \ref{item:ec equality2}$\imp$\ref{item:ec primitives}, $T$ 
can be chosen to be $\supp\mu_f$, and hence it is independent of the pair $(\rho,\sigma)$).
If \ref{item:ec special inverse2} holds then
$\rho=\map_\sigma^*(\map(\rho))=
\sigma^{1/2}\map^*\bz\map(\sigma)^{-1/2}\map(\rho)\map(\sigma)^{-1/2}\jz\sigma^{1/2}$ implies that 
$\supp\rho\le\supp\sigma$ for every $\rho\in\co\C$, and hence $\map_\sigma^*$ can be completed to a map $\mapp$ as required in \ref{item:ec existence of inverse} the same way as in the proof of \ref{item:special inverse}$\imp$\ref{item:existence of inverse} in Theorem \ref{thm:equality}.
This proves \ref{item:ec special inverse2}$\imp$\ref{item:ec existence of inverse}.
Assume that \ref{item:ec existence of inverse} holds.
Fixing any $\rho\in\co\C$ and $\sigma\in\co\C$ such that $\supp\sigma=\supp\C$, 
we have $\mapp(\map(\rho))=\rho$ and $\mapp(\map(\sigma))=\sigma$, and 
Theorem \ref{thm:equality} yields \eqref{special inverse map} for this pair $(\rho,\sigma)$, 
proving \ref{item:ec existence of inverse}$\imp$\ref{item:ec special inverse}. 
The implication \ref{item:ec special inverse}$\imp$\ref{item:ec special inverse2} is 
obvious. 

The implication \ref{item:ec existence of inverse}$\imp$\ref{item:Hoeffding preservation} follows by Proposition \ref{prop:ch monotonicity}, 
and the implication \ref{item:Hoeffding preservation}$\imp$\ref{item:Hoeffding 
preservation2} is obvious. Assume now that \ref{item:Hoeffding preservation2} holds.
Then, by \eqref{0 Hoeffding} and \eqref{0 Hoeffding2}, we have $\sr{\map(A)}{\map(B)}=\sr{A}{B}$, i.e., the equality holds for the standard relative entropy, which is the $f$-divergence corresponding to $f(x)=x\log x$. Since the support of the representing measure for $x\log x$ is $(0,+\infty)$, this yields \ref{item:ec equality2}. The implication 
\ref{item:preservation of Tp}$\imp$\ref{item:Hoeffding preservation2} follows from \eqref{Hoeffding limit}.
Assume that $\map$ is completely positive and \ref{item:ec existence of inverse} holds. Then we can assume $\mapp$ to 
be completely positive, and hence $\map^{\otimes n}$ and $\mapp^{\otimes n}$ are positive and trace-preserving for 
every $n\in\N$. Thus, by the monotonicity of the measures $T_p$,
$\tdist{\rho^{\otimes n}}{\sigma^{\otimes n}}{p}
=\tdist{\mapp^{\otimes n}(\map^{\otimes n}(\rho^{\otimes n}))}{\mapp^{\otimes n}(\map^{\otimes n}(\sigma^{\otimes n}))}{p}\le
\tdist{\map^{\otimes n}(\rho^{\otimes n})}{\map^{\otimes n}(\sigma^{\otimes n})}{p}
\le
\tdist{\rho^{\otimes n}}{\sigma^{\otimes n}}{p}$, and hence
\ref{item:preservation of Tp} holds.

Finally, 
\ref{item:ec special inverse}$\imp$\ref{item:decomposition} follows due to Lemma \ref{lemma:decomposition},
and \ref{item:decomposition}$\imp$\ref{item:ec special inverse} is a matter of straightforward computation.
\end{proof}

Briefly, the above theorem tells that if the noise doesn't decrease some suitable measure of the pairwise 
distinguishability on a set of states then its action can be reversed on that set with some other quantum operation; 
moreover, the reversion operation can be constructed by using the noise operation and any state with maximal support. 
There are apparent differences between the conditions given above; indeed, \ref{item:ec equality2} tells that the 
preservation of one single $f$-divergence is sufficient, while \ref{item:ec primitives} requires the preservation of 
sufficiently (but finitely) many $f$-divergences, \ref{item:Hoeffding preservation} requires the preservation of a 
continuum number of measures, and \ref{item:preservation of Tp} requires even more. The equivalence between 
\ref{item:ec equality2} and \ref{item:ec primitives} is easy to understand; as we have seen in the proof of Theorem 
\ref{thm:equality}, as far as monotonicity and equality in the monotonicity are considered, any $f$-divergence with 
an operator convex $f$ which is not a polynomial is equivalent to the collection of $\vfi_t$-divergences with $t\in\supp\mu_f$, and the condition on the 
cardinality of $\supp\mu_f$ is imposed so that any function on the joint spectrum of the relative modular operators 
can be decomposed as a linear combination of $\vfi_t$'s, which in turn is used to construct the inversion map 
$\map_\sigma^*$. 
It is an open question how much the condition on the cardinality of $\supp f$ can be improved; cf.~Remark \ref{rem:support}.

Note that \ref{item:ec equality2} tells in particular that the preservation of the pairwise 
R\'enyi relative entropies 
for one single parameter value $\alpha\in(0,2)$ is sufficient for reversibility. This is in 
contrast with 
\ref{item:Hoeffding preservation2}, where the preservation of continuum many Hoeffding 
distances are required, despite
the symmetry suggested by \eqref{Hoeffding def} and \eqref{Hoeffding inversion}. 
On the other hand, we have the following:

\begin{prop}\label{prop:Hoeffding preservation}
In the setting of Theorem \ref{thm:error correction}, assume that 
there exists a $\C_0\subset\S(\A_1)$ such that $\co\C_0=\co\C$, and a $\sigma\in\S(\A_1)$ such that $\supp\sigma\ge\supp\C$, and the following hold:
\begin{equation*}
0<m:=\inf_{\rho\in\C_0}\{-\psif{\map(\rho)}{\map(\sigma)}{0}
-\psi'(0|\map(\rho)\|\map(\sigma))\}
\end{equation*}
and for some $r\in(0,m)$,
\begin{equation*}
\hdist{\map(\rho)}{\map(\sigma)}{r}=\hdist{\rho}{\sigma}{r},\ds\ds\ds\rho\in\C_0.
\end{equation*}
Then $\map_\sigma^*(\map(\rho))=\rho$ for every $\rho\in\co\C$.
\end{prop}
\begin{proof}
Immediate from Theorem \ref{thm:equality2}.
\end{proof}

Finally, if all the states in $\C$ have the same support then some of the conditions in 
Theorem \ref{thm:error correction} and Proposition \ref{prop:Hoeffding preservation} can be 
simplified, and we can give a simple condition in terms of preservation of the Chernoff 
distance:

\begin{prop}
Let $\map:\,\A_1\to\A_2$ be a trace-preserving $2$-positive map and let $\C\subset\S(\A_1)$ 
be a non-empty set of states such that $\supp\rho=\supp\C$ for every $\rho\in\C$. Assume 
that there exists a $\sigma\in\S(\A_1)$ such that $\supp\sigma=\supp\C$ and one of the 
following holds:
\begin{enumerate}
\item\label{item:Tp preservation}
There exists a $p\in(0,1)$ such that
\begin{equation}\label{preservation of Tp2}
\tdist{\map^{\otimes n}(\rho^{\otimes n})}{\map^{\otimes n}(\sigma^{\otimes n})}{p}=
\tdist{\rho^{\otimes n}}{\sigma^{\otimes n}}{p},\ds\ds\ds \rho\in\C,\ds n\in\bN.
\end{equation}

\item\label{item:Chernoff preservation}
For every $\rho\in\C$,
\begin{equation*}
\chdist{\map(\rho)}{\map(\sigma)}=\chdist{\rho}{\sigma}.
\end{equation*}

\item\label{item:Hoeffding pres3}
There exists a $\C_0$ such that $\co\C_0=\co\C$ and an $r\in(0,\inf_{\rho\in\C_0}\sr{\map(\sigma)}{\map(\rho}))$ such that for every $\rho\in\C_0$,
\begin{equation}\label{Hoeffding pres3}
\hdist{\map(\rho)}{\map(\sigma)}{r}=\hdist{\rho}{\sigma}{r}.
\end{equation}
\end{enumerate}
Then 
\begin{equation}\label{reversion}
\map_\sigma^*(\map(\rho))=\rho,\ds\ds\ds \rho\in\co\C.
\end{equation}
\end{prop}
\begin{proof}
The implication \ref{item:Tp preservation}$\imp$\ref{item:Chernoff preservation} is immediate from \eqref{Chernoff limit}, and 
\ref{item:Chernoff preservation} implies \eqref{reversion} due to Corollary \ref{cor:equality3}. Assume now that \ref{item:Hoeffding pres3} holds. Since $\supp\rho=\supp\sigma,\,\rho\in\C_0$, we have
$\psif{\map(\rho)}{\map(\sigma)}{0}=0$ and 
$-\psi'(0|\map(\rho)\|\map(\sigma))=\sr{\map(\sigma)}{\map(\rho)},\,\rho\in\C_0$. Hence, \eqref{Hoeffding pres3} yields 
\eqref{reversion} due to Proposition \ref{prop:Hoeffding preservation}.
\end{proof}
\medskip

Note that the conditions \eqref{preservation of Tp} and \eqref{preservation of Tp2} are 
very different from the others, as they require the preservation of some measure for
arbitrary tensor powers. These conditions could be 
simplified if the trace-norm distance could be 
represented as an $f$-divergence. Note that this is possible in the classical case; indeed, 
if $p$ and $q$ are 
probability density functions on some finite set $\X$, and $f(x):=|x-1|,\,x\in\R$, then 
\begin{equation*}
\rsr{p}{q}{f}=\sum_{x\in\X} q(x)|p(x)/q(x)-1|=\sum_{x\in\X}|p(x)-q(x)|=\norm{p-q}_1.
\end{equation*}
Note, however, that the above $f$ is not operator convex, and hence the proof given in Theorem \ref{thm:equality} 
wouldn't work for it. Even worse, the trace-norm distance cannot be represented as an $f$-divergence, as we show below by a simple argument.

\begin{cor}\label{cor:trace norm}
If the observable algebra of a quantum system is non-commutative then the trace-norm distance on its state space cannot be represented as an $f$-divergence.
\end{cor}
\begin{proof}
Assume that $\A\subset\B(\hil)$ is non-commutative; then we can find orthonormal vectors 
$e_1,e_2\in\hil$ such that $\diad{e_i}{e_j}\in\A,\,i=1,2$. 
(For simplicity, we neglect possible higher multiplicities; taking them into account would only result in a constant multiplication factor in the formulas below.)
Assume that the trace-norm
distance can be represented as an $f$-divergence. Then, for every $s\in[0,1]$ and 
$t\in(0,1)$, when $\rho:=s|e_1\>\<e_1|+(1-s)|e_2\>\<e_2|$ and $\sigma:=t|e_1\>\<e_1|+(1-t)|e_2\>\<e_2|$,
we have
\begin{equation*}
tf(s/t)+(1-t)f((1-s)/(1-t))=\rsr{\rho}{\sigma}{f}=\norm{\rho-\sigma}_1=2|s-t|.
\end{equation*}
Letting $s=t$ gives $f(1)=0$. Letting $t\searrow0$ gives $s\omega(f)+f(1-s)=2s$ for all
$s\in(0,1]$. This implies that $\omega(f)$ is finite and $\omega(f)+f(0)=2$.
Now let $\rho:=\pr{e_1}$ and $\sigma:=\pr{\psi}$, where $\psi:=(e_1+e_2)/\sqrt{2}$. Then
$\|\rho-\sigma\|_1=\sqrt2$, while by (2.6) one can easily compute
\begin{equation*}
\rsr{\rho}{\sigma}{f}=\half f(1)+\half\omega(f)+\half f(0)
=\half(\omega(f)+f(0))=1.\qedhere
\end{equation*}
\end{proof}

\begin{rem}
A similar argument as above can be used to show that for any $p\in(0,1)$, the measure 
$D_p(\rho\|\sigma):=1-\norm{p\rho-(1-p)\sigma}_1$ cannot be represented as an $f$-divergence
on the state space of any non-commutative finite-dimensional $C^*$-algebra.
\end{rem}

\begin{rem}
In general, a function on pairs of classical probability distributions might have several different extensions to 
quantum states. A function that can be represented as an $f$-divergence has an extension given by the corresponding 
quantum $f$-divergence. It is not clear whether this extension has any operational significance in the case of 
$f(x):=|x-1|$.
\end{rem}

While the impossibility to represent the trace-norm distance as an $f$-divergence shows that the approach followed in Theorem \ref{thm:error correction} cannot be used to simplify the condition in \ref{item:preservation of Tp} of the theorem, other approaches might lead to better results. Indeed, the results of the recent paper \cite{BNPV} can be reformulated in the following way:
\begin{thm}\label{thm:BNPV}
Let $\C\subset\S(\A_1)$ be a convex set of states and let $\map:\,\A_1\to\A_2$ be a completely positive trace-preserving map such that
\begin{equation*}
\tdist{\map(\rho)}{\map(\sigma)}{p}=\tdist{\rho}{\sigma}{p},\ds\ds\ds p\in(0,1).
\end{equation*}
Then the fixed-point set of $\map_P^*\circ\map$ is a $C^*$-subalgebra of $P\A_1P$, where $P$ is the projection onto $\supp\C$, and the trace-preserving conditional expectation $\P$ from $P\A_1P$ onto $\fix{\map_P^*\circ\map}$ is $T_p$-preserving for all $p\in(0,1)$. If, moreover, the restriction of $\P$ onto $\C$ is surjective onto 
the state space of $\fix{\map_P^*\circ\map}$ then (i)--(x) of Theorem \ref{thm:error correction} hold.
\end{thm}

Note that the continuum many conditions requiring the preservation of $T_p$ for all $p\in(0,1)$ in Theorem \ref{thm:BNPV} can be simplified to a single condition, requiring that $\map$ is trace-norm preserving on the real subspace generated by $\C$. 
Note also that the surjectivity condition is sufficient but obviously not necessary. It is, however, an open question whether it can be completely removed.  
In the approach followed in \cite{BNPV}, it is important that one starts with a convex set of states. The same problem was studied in \cite{Jencova} in a different setting, and the following has been shown:
\begin{thm}\label{thm:Jencova}
Let $\rho,\sigma\in\S(\A)$ be invertible density operators and $\map$ be the trace-preserving conditional expectation onto a subalgebra $\A_0$ of $\A$. Assume that 
$\tdist{\map(\rho)}{\map(\sigma)}{p}=\tdist{\rho}{\sigma}{p}$ for every $p\in(0,1)$, and
$\A_0$ is commutative or $\rho$ and $\sigma$ commute. Then $\map_\sigma^*(\map(\rho))=\rho$ and $\map_\rho^*(\map(\sigma))=\sigma$.
\end{thm}
\begin{rem}
In \cite{Jencova} the condition $\tdist{\map(\rho)}{\map(\sigma)}{p}=\tdist{\rho}{\sigma}{p},\,p\in(0,1)$, was called $2$-sufficiency, and
\begin{equation}\label{(2,n)}
T_p(\map^{\otimes n}(\rho^{\otimes n})\|\map^{\otimes n}(\sigma^{\otimes n}))=T_p(\rho^{\otimes n}\|\sigma^{\otimes n}),\ds\ds\ds p\in(0,1),\s n\in\bN, 
\end{equation}
 was called $(2,n)$-sufficiency. It was also shown in Theorem 6 of \cite{Jencova} that in the setting of Theorem \ref{thm:Jencova}, 
\eqref{(2,n)} is sufficient for the conclusion of Theorem \ref{thm:Jencova} to hold.
\end{rem}

\section{An integral representation for operator convex functions}\label{sec:integral representation}
\setcounter{equation}{0}

Operator monotone and operator convex functions play an important role in quantum 
information theory \cite{Petzbook}. 
Several ways are known to decompose them as integrals of some families of functions of simpler forms
\cite{Bhatia,Hiai}.
Here we present a representation that is well-suited for our analysis of $f$-divergences, and seems to be a new result.

\begin{thm}\label{thm:integral representation}
A continuous real-valued function $f$ on $[0,+\infty)$ is operator convex if and only if 
there exist a real number $a$, a non-negative number $b$, and a non-negative measure $\mu$ on $(0,+\infty)$, 
satisfying 
\begin{equation}\label{finite integral1}
\int_{(0,+\infty)}{d\mu(t)\over(1+t)^2}<+\infty,
\end{equation}
such that 
\begin{equation}\label{integral representation1}
f(x)=f(0)+ax+bx^2
+\int_{(0,+\infty)}\biggl({x\over1+t}-{x\over x+t}\biggr)\,d\mu(t),\ds x\in[0,+\infty).
\end{equation}
Moreover, the numbers $a,b$, and the measure $\mu$ are uniquely determined by $f$, and
\begin{equation*}
b=\lim_{x\to+\infty}{f(x)\over x^2},\ds\ds\ds\ds
a=f(1)-f(0)-b.
\end{equation*}
\end{thm}
\begin{proof}
Obviously, if $f$ admits an integral representation as in \eqref{integral representation1} then $f$ is operator convex, and 
\begin{equation*}
f(1)=f(0)+a+b,\ds\ds\ds\ds b=\lim_{x\to+\infty}{f(x)\over x^2},
\end{equation*}
where the latter follows by the Lebesgue dominated convergence theorem, using \eqref{finite integral1} and that, for $x>1$,
\begin{equation*}
0\le\frac{1}{x^2}\bz\frac{x}{1+t}-\frac{x}{x+t}\jz=\frac{x-1}{x(x+t)(1+t)}\le
\frac{2x}{x(1+t)(1+t)}=\frac{2}{(1+t)^2}.
\end{equation*}
Hence what is left to prove is that any operator convex function admits a representation as in 
\eqref{integral representation1}, and that the measure $\mu$ is uniquely determined by $f$.

Assume now that $f$ is an operator convex function on $[0,+\infty)$. 
Then, by Kraus' theorem (see \cite{Kraus} or Corollary 2.7.8 in \cite{Hiai}), the function
\begin{equation*}
g(x):={f(x)-f(1)\over x-1},\qquad x\in[0,+\infty)\setminus\{1\},\ds\ds\ds
g(1):=f'(1),
\end{equation*}
is an operator monotone function on
$(0,+\infty)$. Therefore, it admits an integral representation
\begin{equation}\label{integral representation for op.mon.}
g(x)=a'+bx+\int_{(0,+\infty)}{x(1+t)\over x+t}\,dm(t),\ds\ds x\in[0,+\infty),
\end{equation}
where $m$ is a positive finite measure on $(0,+\infty)$, and 
\begin{equation*}
a'=g(0)=f(1)-f(0),\qquad 0\le b=\lim_{x\to+\infty}{g(x)\over x}=\lim_{x\to+\infty}{f(x)\over x^2}
\end{equation*}
(see Theorem 2.7.11 in \cite{Hiai} or pp.~144--145 in \cite{Bhatia}). Here, the measure $m$, as well as
$a',b$, are unique and
\begin{equation*}
m((0,+\infty))=g(1)-a'-b=f'(1)-f(1)+f(0)-b.
\end{equation*}
Thus, we have
\begin{align*}
f(x)&=f(1)+g(x)(x-1) \\
&=
f(1)+(f(1)-f(0))(x-1)+bx(x-1)
+\int_{(0,+\infty)}{x(x-1)(1+t)\over x+t}\,dm(t) \\
&=
f(0)+(f(1)-f(0)-b)x+bx^2
+\int_{(0,+\infty)}\biggl({x\over1+t}-{x\over x+t}\biggr)
(1+t)^2\,dm(t)\\
&=
f(0)+ax+bx^2
+\int_{(0,+\infty)}\biggl({x\over1+t}-{x\over x+t}\biggr)\,d\mu(t),
\end{align*}
where we have defined $ a:=f(1)-f(0)-b$ and $d\mu(t):=(1+t)^2\,dm(t)$. Finiteness of $m$ yields that $\mu$ satisfies \eqref{finite integral1}.

Finally, to see the uniqueness of the measure $\mu$, 
assume that $f$ admits an integral representation as in \eqref{integral representation1}.
Then, $f$ is operator convex, and hence the function $g$ on $[0,+\infty)$, defined as
\begin{align}\label{integral representation of g}
g(x)&:={f(x)-f(1)\over x-1}
=(a+b)+bx+\int_{(0,+\infty)}{x(1+t)\over x+t}
\,{d\mu(t)\over(1+t)^2},
\end{align}
is operator monotone. Therefore, it admits an integral representation as in 
\eqref{integral representation for op.mon.}, and the uniqeness of the parameters of that representation yields that $d\mu(t)=(1+t)^2 dm(t)$. Hence, the measure $\mu$ is uniquely determined by $f$.
\end{proof}

\begin{cor}\label{cor:integral representation2}
Assume that $f$ is a continuous operator convex function on $[0,+\infty)$ 
that is not a polynomial. Then it can be written in the form
\begin{equation}\label{integral representation2}
f(x)=f(0)+bx^2
+\int_{(0,+\infty)}\biggl(\psi(t)x-{x\over x+t}\biggr)\,d\mu(t),\ds x\in[0,+\infty),
\end{equation}
where $b=\lim_{x\to+\infty}f(x)/x^2\ge 0$, and $\mu$ is a non-negative measure on $(0,+\infty)$. Moreover, we can choose
\begin{equation}\label{psi}
\psi(t):=\frac{1}{1+t}+\frac{f(1)-f(0)-b}{f'(1)-f(1)+f(0)-b}\cdot\frac{1}{(1+t)^2},
\end{equation}
and if $b=0$ and $f'(1)\ge 0$ then $\psi(t)\ge 0,\,t\in (0,+\infty)$.
\end{cor}
\begin{proof}
Since $f$ is operator convex, it can be written in the form \eqref{integral representation1}
due to Theorem \ref{thm:integral representation}.
Since $f$ is not a polynomial, we have $m((0,+\infty))>0$, where
$dm(t):=d\mu(t)/(1+t^2)$.
Moreover, by \eqref{integral representation of g},
$f'(1)=g(1)=a+2b+m((0,+\infty))$,
from which $m((0,+\infty))=f'(1)-a-2b$. Using that $a=f(1)-f(0)-b$, we finally obtain
\begin{equation*}
a=\frac{a}{m((0,+\infty))}\int_{(0,+\infty)}\,dm(t)=
\frac{f(1)-f(0)-b}{f'(1)-f(1)+f(0)-b}\int_{(0,+\infty)}\frac{1}{(1+t)^2}\,d\mu(t).
\end{equation*}
Substituting it into \eqref{integral representation1}, we obtain 
\eqref{integral representation2} with $\psi$ as in \eqref{psi}.
Note that $(1+t)^2\psi(t)=1+t+\frac{a}{m((0,+\infty))}\ge 1+\frac{a}{m((0,+\infty))}$. Hence, if $b=0$ and $0\le f'(1)=a+2b+m((0,+\infty))=a+m((0,+\infty))$ then $\psi(t)\ge 0$, proving the last assertion.
\end{proof}

\begin{example}\label{ex:representations}
 \s\ 
\begin{enumerate}
\item\label{item:xlogx} $f(x):=x\log x$ admits the integral representation
\begin{equation*}
x\log x=\int_{(0,+\infty)}\biggl({x\over1+t}-{x\over x+t}\biggr)\,dt.
\end{equation*}
($f(0)=a=b=0$ and $\mu$ is the Lebesgue measure in \eqref{integral representation1}.)
\item\label{item:alpha<1} $f(x):=-x^\alpha$ ($0<\alpha<1$) admits the integral representation
(see \cite[Exercise V.1.10]{Bhatia})
\begin{equation*}
-x^\alpha={\sin\alpha\pi\over\pi}\int_{(0,+\infty)}
\biggl(-{x\over x+t}\biggr)t^{\alpha-1}\,dt.
\end{equation*}
($f(0)=b=0$, $d\mu(t)={\sin\alpha\pi\over\pi}t^{\alpha-1}dt$, and $\psi\equiv 0$ in \eqref{integral representation2}.)
Using that 
${\sin\alpha\pi\over\pi}\int_{(0,+\infty)}
{xt^{\alpha-1}\over1+t}\,dt=x$, we have 
\begin{align*}
-x^\alpha
&=-x+{\sin\alpha\pi\over\pi}\int_{(0,+\infty)}
\biggl({x\over1+t}-{x\over x+t}\biggr)t^{\alpha-1}\,dt.
\end{align*}
($f(0)=b=0,\,a=-1$, and $d\mu(t)={\sin\alpha\pi\over\pi}t^{\alpha-1}dt$ in \eqref{integral representation1}.)
\item\label{item:alpha>1} By the previous point, $f(x):=x^\alpha$ ($1<\alpha<2$) admits the representation
\begin{equation*}
x^\alpha={\sin(\alpha-1)\pi\over\pi}\int_{(0,+\infty)}
{x^2 t^{\alpha-2}\over x+t}\,dt 
={\sin(\alpha-1)\pi\over\pi}\int_{(0,+\infty)}
\biggl({x\over t}-{x\over x+t}\biggr)t^{\alpha-1}\,dt
\end{equation*}
($f(0)=b=0$, $\psi(t)=1/t$, and $d\mu(t)={\sin(\alpha-1)\pi\over\pi}t^{\alpha-1}\,dt$ in 
\eqref{integral representation2}.)
Using that
\begin{equation*}
{\sin(\alpha-1)\pi\over\pi}\int_{(0,+\infty)}
\biggl({x\over t}-{x\over1+t}\biggr)t^{\alpha-1}\,dt
={\sin(\alpha-1)\pi\over\pi}\int_{(0,+\infty)}
{xt^{\alpha-2}\over1+t}\,dt=x,
\end{equation*}
we also obtain
\begin{equation*}
x^{\alpha}=x+{\sin(\alpha-1)\pi\over\pi}\int_{(0,+\infty)}
\biggl({x\over1+t}-{x\over x+t}\biggr)t^{\alpha-1}\,dt.
\end{equation*}
($f(0)=0,\,a=1,\,b=0$ and $d\mu(t)={\sin(\alpha-1)\pi\over\pi}t^{\alpha-1}\,dt$ in \eqref{integral representation1}.)
\end{enumerate}
\end{example}
\medskip

Note that the function $\psi$ in \eqref{integral representation2} is not unique. 
For instance, if $\mu$ is finitely supported on a set $\{t_1,\ldots,t_r\}$ then only 
the sum $\sum_{i=1}^r\psi(t_r)$ is determined by $f$ while the individual values
$\psi(t_1),\ldots,\psi(t_r)$ are not.

Note also that in general, $\int_{(0,+\infty)}\frac{1}{1+t}\,d\mu(t)$ might not be finite and hence the term $\int_{(0,+\infty)}\frac{x}{1+t}\,d\mu(t)$ cannot be merged with $ax$ in
\eqref{integral representation1}. Similarly, the integral 
$\int_{(0,+\infty)}\psi(t)\,d\mu(t)$ might be infinite and hence it might not be possible to 
separate it as a linear term in the representation \eqref{integral representation2} of $f$. This is clear, for instance, from \ref{item:xlogx} of Example \ref{ex:representations}.
We have the following:

\begin{prop}\label{prop:finite of}
For a continuous real-valued function $f$ on $[0,+\infty)$ the following are equivalent:
\begin{itemize}
\item[\rm(i)] $f$ is operator convex on $[0,+\infty)$ with
$\lim_{x\to+\infty}f(x)/x<+\infty$;
\item[\rm(ii)] there exist an $\alpha\in\mathbb{R}$ and a positive measure $\mu$ on $(0,+\infty)$,
satisfying
\begin{equation}\label{finite int}
\int_{(0,+\infty)}{d\mu(t)\over1+t}<+\infty,
\end{equation}
such that
\begin{equation}\label{special representation}
f(x)=f(0)+\alpha x-\int_{(0,+\infty)}{x\over x+t}\,d\mu(t),\qquad x\in[0,+\infty).
\end{equation}
\end{itemize}
\end{prop}
\begin{proof}
First, note that if $f$ is convex on $[0,+\infty)$ as a numerical function, then
$\lim_{x\to+\infty}f(x)/x$ exists in $(-\infty,+\infty]$. In fact, by convexity,
$(f(x)-f(1))/(x-1)$ is non-decreasing for $x>1$, so that
$$
\lim_{x\to+\infty}{f(x)\over x}=\lim_{x\to+\infty}{f(x)-f(1)\over x-1}
$$
exists in $(-\infty,+\infty]$. Also, note that condition \eqref{finite int} is necessary for $f(1)$ to
be defined in \eqref{special representation}, and also sufficient to define $f(x)$ by \eqref{special representation} for all $x\in[0,+\infty)$.

(i) $\Rightarrow$ (ii).\enspace
By assumption, $f$ is an operator convex function on $[0,+\infty)$ such that
$\lim_{x\to+\infty}f(x)/x$ is finite, hence $\lim_{x\to+\infty}f(x)/x^2=0$. By Theorem \ref{thm:integral representation},
we have
$$
f(x)=f(0)+ax+\int_{(0,+\infty)}\biggl({x\over1+t}-{x\over x+t}\biggr)\,d\mu(t),
\qquad x\in[0,+\infty),
$$
where $a\in\mathbb{R}$ and $\mu$ is a positive measure on $(0,+\infty)$. We write
$$
{f(x)\over x}
={f(0)\over x}+a+\int_{(0,+\infty)}\biggl({1\over1+t}-{1\over x+t}\biggr)\,d\mu(t).
$$
Since
$$
0<{1\over1+t}-{1\over x+t}\nearrow{1\over1+t}\quad\mbox{as $1<x\nearrow+\infty$},
$$
the monotone convergence theorem yields that
$$
\lim_{x\to+\infty}{f(x)\over x}=a+\int_{(0,+\infty)}{d\mu(t)\over1+t},
$$
which implies \eqref{finite int} and
$$
f(x)=f(0)+\biggl(a+\int_{(0,+\infty)}{d\mu(t)\over1+t}\biggr)x
-\int_{(0,+\infty)}{x\over x+t}\,d\mu(t).
$$
Hence $f$ admits a representation of the form \eqref{special representation}.

(ii) $\Rightarrow$ (i).\enspace
It is obvious that $f$ given in \eqref{special representation} is operator convex on $[0,+\infty)$.  Since
${1\over x+t}\le{1\over1+t}$ for all $x>1$ and all $t\in[0,+\infty)$, the Lebesgue
convergence theorem yields that
\begin{equation*}
\lim_{x\to+\infty}\int_{(0,+\infty)}{d\mu(t)\over x+t}=0
\end{equation*}
and so
\begin{equation*}
{f(x)\over x}={f(0)\over x}+\alpha-\int_{(0,+\infty)}{d\mu(t)\over x+t}
\longrightarrow \alpha\quad\mbox{as $x\to+\infty$}.
\end{equation*}
Hence (i) follows.
\end{proof}

\begin{rem}
Note that the condition $\lim_{x\to+\infty}f(x)/x<+\infty$ puts a strong restriction on
an operator convex function $f$. Important examples for which it is not satisfied include
$f(x)=x\log x$ and $f(x)=x^\alpha$ for $\alpha\in(1,2]$.
\end{rem}

\section{Closing remarks}\label{sec:closing}

Quantum $f$-divergences are a quantum generalization of classical $f$-divergences, which class in the classical case 
contains most of the distinguishability measures that are relevant to classical statistics. 
Although our Corollary
\ref{cor:trace norm} shows that $f$-divergences are less universal in the quantum case, they still provide a very 
efficient tool to obtain monotonicity and convexity properties of several distinguishability measures that are 
relevant to quantum statistics, including the relative entropy, the R\'enyi relative entropies, and the Chernoff and 
Hoeffding distances. 

There are also differences between the classical and the quantum cases in the technical conditions needed to prove the 
monotonicity. For the approach followed here, it is important that the defining function is not only convex but 
operator convex, and the map is not only positive but it is also decomposable in the sense of Remark \ref{rem:decomposability}. It is 
unknown whether the monotonicity can be proved without these assumptions in general, although Corollary \ref{cor:0 
Renyi} and Lemma \ref{lemma:2 Renyi monotonicity} show for instance that positivity of $\map$ might be sufficient in 
some special cases. For measures that have an operational 
interpretation in state discrimination, like the relative entropy, the R\'enyi $\alpha$-relative entropies with 
$\alpha\in(0,1)$, and the Chernoff and Hoeffding distances, the monotonicity holds for any positive trace-preserving map $\map$ such that
$\map^{\otimes n}$ is positive for every $n\in\N$ \cite{Hayashibook,MH}. 
Note that both the set of maps satisfying this latter property and the set of maps that are decomposable in the sense of Remark \ref{rem:decomposability} contain all the 
completely positive trace-preserving maps, but we are not aware of any other explicit relation between these two sets. Moreover, the only example we know for a map $\map$ which 
is not completely positive 
but $\map^{\otimes n}$ is positive for every $n\in\bN$ is the transposition, which is trivial in the sense that it preserves any $f$-divergence (where $f$ does not even need to be 
convex; see Corollary \ref{cor:scaling} and Remark \ref{rem:transposition}).

Quantum $f$-divergences are essentially a special case of Petz' quasi-entropies with $K=I$ 
(see the Introduction) 
with the minor modification of allowing operators that are not strictly positive definite. 
While the monotonicity 
inequality in Theorem \ref{thm:monotonicity for F} can be proved for the 
quasi-entropies with general $K$ quite similarly to the case $K=I$, our 
analysis of the equality case in Theorem \ref{thm:equality} doesn't seem to extend to 
$K\ne I$. A special case has been treated recently in \cite{JR}, where a characterization
for the equality case in the joint convexity of the quasi-entropies 
$S^K_{f_\alpha}(.\|.)$  (see Example \ref{ex: f-div Renyi} for $K=I$) was given for 
arbitrary $K$ and $\alpha\in(0,2)$.
Note that joint convexity is a special case of 
the monotonicity under partial traces (see \cite[Theorem 6]{Petz} or Corollary \ref{cor:joint convexity} of this paper), 
while monotonicity under partial traces can also be proven from the joint convexity 
for $K$'s of 
special type \cite{LR}, which in turn implies the monotonicity under 
completely positive trace-preserving maps by using their Lindblad respresentation \cite{TCR}.
For a particularly elegant recent proof of the joint convexity for general $K$'s, see 
\cite{Effros}.

Various characterizations of the equality in the case $K=I$ have been given before for different types of maps and 
classes of functions, including the equality case for the strong subadditivity of entropy 
and the joint convexity of the R\'enyi relative entropies
\cite{HJPW,JP,JR,OP,Petz2,Petz3,Petzbook,PD138,Ruskai2,Sharma}.
Our Theorem \ref{thm:equality} extends all these results and it seems to be the most general characterization of the 
equality, at least in finite dimension. The relevant part from the point of view of application to quantum error 
correction is that the preservation of some suitable distinguishability measure yields the reversibility of the 
stochastic operation, and the reversal map can be constructed from the original one in a canonical way. There are 
various 
technical conditions imposed in Theorems \ref{thm:equality} and \ref{thm:error correction} that might be possible to 
remove. For instance, it is not clear whether the support condition in \eqref{support condition} is necessary 
or maybe the preservation of 
$S_{\vfi_t}(.\|.)$ for one single $t>0$ is sufficient for reversibility. 
It is also an open question whether the surjectivity condition 
in Theorem \ref{thm:BNPV} can be removed.

\section*{Acknowledgments}
Partial funding was provided by 
the JSPS-HAS Japan-Hungary Joint Project,
the Grant-in-Aid for Scientific Research (C)21540208 (FH), and
the Hungarian Research Grant OTKA T068258 (MM and DP).
The Centre for Quantum Technologies is funded by the
Singapore Ministry of Education and the National Research Foundation as part of the
Research Centres of Excellence program.
Part of this work was done when MM was a Research Fellow at the Erwin Schr\"odinger Institute for Mathematical Physics in 2009 and later when the first three authors participated in the 
Quantum Information Theory program of 
the Mittag-Leffler Institute in 2010.
Discussions with Tomohiro Ogawa and David Reeb (MM) and with Hui Khoon Ng (CB and MM) helped to improve the paper and are gratefully acknowledged here.
The authors are grateful to anonymous referees for their comments, especially for pointing out Reference \cite{Jencova},
and to Lajos Moln\'ar for pointing out an error in the original version of Proposition \ref{prop:continuity}.

\appendix

\section{Commuting operators and the operator H\"older inequality}\label{sec:classical}
\setcounter{equation}{0}

We will need the following two well-known lemmas in this section. The first one is a 
generalization of the so-called log-sum inequality, while the second one is a generalization of Jensen's inequality for the expectation values of self-adjoint operators.

\begin{lemma}\label{lemma:generalized log-sum}
Let $f:\,[0,+\infty)\to\R$ be a convex function. Let $a_i\ge 0,\,b_i>0,\,i=1,\ldots,r$, and define $a:=\sum_{i=1}^r a_i,\s b:=\sum_{i=1}^r b_i$.
Then,
\begin{equation}\label{log-sum}
bf(a/b)\le\sum_{i=1}^r b_if(a_i/b_i).
\end{equation}
Moreover, if $f$ is strictly convex, then equality holds if and only if
$a_i/b_i$ is independent of $i$.
\end{lemma}
\begin{proof}
Convexity of $f$ yields that
\begin{align*}
f(a/b)&=f\bz\sum_{i=1}^r\frac{b_i}{b}\frac{a_i}{b_i}\jz\le
\sum_{i=1}^r\frac{b_i}{b}f\bz\frac{a_i}{b_i}\jz,
\end{align*}
which yields \eqref{log-sum}, and the characterization of equality is immediate from the strict convexity of $f$.
\end{proof}

\begin{lemma}\label{lemma:Jensen}
Let $A$ be a self-adjoint operator and $\rho$ be a density operator on a finite-dimensional 
Hilbert space $\hil$. If $f$ is a convex function on the convex hull of $\spect(A)$ then
\begin{equation}\label{Jensen inequality}
f\bz\Tr A\rho\jz\le\Tr f(A)\rho.
\end{equation}
If $f$ is strictly convex then equality holds in \eqref{Jensen inequality} if and only if
$\rho^0$ is a subprojection of a spectral projection of $A$.
\end{lemma}
\begin{proof}
Let $A=\sum_a aP_a$ be the spectral decomposition of $A$.
Since $\{\Tr P_a\rho\,:\,a\in\spect(A)\}$ is a probability distribution on $\spect(A)$,
Jensen's inequality yields
$f\bz\Tr A\rho\jz=f\bz\sum_a a\Tr P_a\rho\jz\le\sum_a f(a)\Tr P_a\rho$, and it is obvious 
that equality holds whenever $\Tr P_a\rho=0$ for all but one $a\in\spect(A)$.
On the other hand, if there are more than one $a\in\spect(A)$ such that $\Tr P_a\rho>0$ then 
the above inequality is strict whenever $f$ is strictly convex.
\end{proof}

\begin{prop}\label{prop:classical}
Let $A,B\in\A_{1,+}$ be such that $A$ commutes with $B$ and 
let $\map:\,\A_1\to\A_2$ be a substochastic map such that $\map(A)$ commutes with $\map(B)$ and $\Tr\map(B)=\Tr B$.
For any convex function $f:\,[0,+\infty)\to\bR$,
\begin{equation}\label{classical monotonicity}
\rsr{\map(A)}{\map(B)}{f}\le\rsr{A}{B}{f}.
\end{equation}
If $\supp A\le\supp B$ and $f$ is strictly convex then equality 
holds in \eqref{classical monotonicity} if and only if $\map_B^*(\map(A))=A$.
\end{prop}
\begin{proof}
Let us consider first the inequality \eqref{classical monotonicity}.
Note that if $\of=+\infty$ and $\supp A\nleq\supp B$ then the RHS of \eqref{classical monotonicity} is $+\infty$, and hence the inequality holds trivially. On the other hand, if $\of<+\infty$ then it is enough to prove that 
$\rsr{\map(A)}{\map(B+\ep I)}{f}\le\rsr{A}{B+\ep I}{f}$ for every $\ep>0$, as taking the limit $\ep\searrow 0$ 
then yields \eqref{classical monotonicity} due to Proposition \ref{prop:continuity}. Hence, for the rest we can assume without loss of generality that $\supp A\le\supp B$.

Since $A$ and $B$ commute, there exists an orthonormal basis $\{e_x\}_{x\in\X}$ in $\supp B$ 
such that $A=\sum_{x\in\X}A(x)\pr{e_x}$ and $B=\sum_{x\in\X}B(x)\pr{e_x}$, where
$A(x):=\inner{e_x}{Ae_x},\,B(x):=\inner{e_x}{Be_x},\,x\in\X$. Similarly,
there exists a basis $\{f_y\}_{y\in\Y}$ in $\supp\map(B)$ 
such that $\map(A)=\sum_{y\in\Y}\map(A)(y)\pr{f_y}$ and $\map(B)=\sum_{y\in\Y}\map(B)(y)\pr{f_y}$, where
$\map(A)(y):=\inner{f_y}{\map(A)f_y}$, $\map(B)(y):=\inner{f_y}{\map(B)f_y}$.
We have
\begin{equation*}
\rsr{A}{B}{f}=\sum_{x}B(x)f\bz\frac{A(x)}{B(x)}\jz,\ds\ds\ds
\rsr{\map(A)}{\map(B)}{f}=\sum_{y}\map(B)(y)f\bz\frac{\map(A)(y)}{\map(B)(y)}\jz.
\end{equation*}
Let $T_{xy}:=\inner{f_y}{\map(\pr{e_x})f_y}$; then 
$\map(A)(y)=\sum_{x\in\X}T_{xy}A(x),\,\map(B)(y)=\sum_{x\in\X}T_{xy}B(x)$, and
Lemma \ref{lemma:generalized log-sum} yields
\begin{align}\label{pointwise inequality}
\map(B)(y)f\bz\frac{\map(A)(y)}{\map(B)(y)}\jz&
\le
\sum_x T_{xy}B(x)f\bz\frac{T_{xy}A(x)}{T_{xy}B(x)}\jz.
\end{align}
Since $\supp \pr{e_x}\le\supp B$, Lemma \ref{lemma:support inequality} yields that $\Tr\map(\pr{e_x})=\Tr\pr{e_x}=1,\,x\in\X$,
and hence $\sum_{y\in\Y}T_{xy}=1,\,x\in\X$. Summing over
$y$ in \eqref{pointwise inequality} yields \eqref{classical monotonicity}.

Obviously, equality holds in 
\eqref{classical monotonicity} if and only if \eqref{pointwise inequality} holds with 
equality for every $y\in\Y$. Assuming that $f$ is strictly convex, we obtain, due to Lemma 
\ref{lemma:generalized log-sum}, that for every $y\in\Y$ there exists a positive constant 
$c(y)$ such that $T_{xy}A(x)=c(y)T_{xy}B(x)$, i.e.,
\begin{equation}\label{pointwise inequality2}
A(x)=c(y)B(x)
\end{equation}
for every $x$ such that $T_{xy}>0$.
Assume that \eqref{pointwise inequality2} holds; then we have
$\map(A)(y)=\sum_x T_{xy}A(x)=\sum_x T_{xy}c(y)B(x)=c(y)\map(B)(y)$ and hence,
\begin{equation*}
\map^*_B(\map(A))(x)=
B(x)\sum_y T_{xy}\frac{\map(A)(y)}{\map(B)(y)}=
B(x)\sum_y T_{xy}\frac{A(x)}{B(x)}=A(x),\ds\ds\ds x\in\X.\qedhere
\end{equation*}
\end{proof}

The following Proposition gives an important special case where the monotonicity inequality 
\eqref{classical monotonicity} holds even though $A$ and $B$ don't commute and $f$ is only 
assumed to be convex.
\begin{prop}\label{prop:pinching}
Let $A,B\in\A_+$ be such that $B\ne 0$, let $B=\sum_{b\in\spect(B)} bQ_b$ be the spectral decomposition of $B$ and let 
$\E_B:\,X\mapsto \sum_{b\in\spect(B)} Q_bXQ_b$ be the pinching defined by $B$. For every convex function 
$f:\,[0,+\infty)\to\bR$,
\begin{equation}\label{pinching inequality}
\rsr{A}{B}{f}\ge\rsr{\E_B(A)}{\E_B(B)}{f}=
\rsr{\E_B(A)}{B}{f}\ge (\Tr B)f\bz\frac{\Tr A}{\Tr B}\jz.
\end{equation}

Moreover, if $\supp A\le\supp B$ and $f$ is strictly convex then the first inequality in 
\eqref{pinching inequality} holds with equality if and only if $A$ commutes with $B$, and the 
second inequality holds with equality if and only if $\E_B(A)$ is a constant multiple of $B$.
In particular, $\rsr{A}{B}{f}=(\Tr B)f\bz\frac{\Tr A}{\Tr B}\jz$ if and only if
$A$ is a constant multiple of $B$.
\end{prop}
\begin{proof}
All the assertions are obvious when $A=0$, so for the rest we assume $A\ne 0$.
Assume first that $\supp A\le\supp B$.
For every $b\in\spect(B)$ and $\lambda\in\bR$, let
$P_{\lambda}^{(b)}$ be the spectral projection of $Q_bAQ_b$
corresponding to the singleton $\{\lambda\}$, and let 
$\tilde P_{\lambda}^{(b)}:=Q_bP_{\lambda}^{(b)}Q_b$. Note that 
$\tilde P_{\lambda}^{(b)}=P_{\lambda}^{(b)}$ for every $\lambda\ne 0$,
and $Q_b=\sum_{\lambda}\tilde P_{\lambda}^{(b)}$.
The spectral projection of $\E_B(A)$
corresponding to the singleton $\{\lambda\}$ is $\sum_{b\in\spect(B)}\tilde P_{\lambda}^{(b)}$.
For every $b\in\spect(B)\setminus\{0\}$ and $\lambda\in\bR$, let $\rho_{b,\lambda}$ be a density operator such that $\rho_{b,\lambda}=\tilde P_{\lambda}^{(b)}/\Tr \tilde P_{\lambda}^{(b)}$ whenever $\tilde P_{\lambda}^{(b)}\ne 0$.
By \eqref{f-div explicit expression3}, we have
\begin{align}
&\rsr{\E_B(A)}{\E_B(B)}{f}=
\rsr{\E_B(A)}{B}{f}
=
\sum_{b\in\spect(B)\setminus\{0\}}\sum_{\lambda}bf(\lambda/b)
\Tr \sum_{b'\in\spect(B)}\tilde P_{\lambda}^{(b')}Q_b\nonumber\\
&=
\sum_{b\in\spect(B)\setminus\{0\}}\sum_{\lambda}bf(\lambda/b)\Tr \tilde P_{\lambda}^{(b)}
=
\sum_{b\in\spect(B)\setminus\{0\}}\sum_{\lambda}bf\bz\Tr ((A/b)\rho_{b,\lambda})\jz\Tr \tilde P_{\lambda}^{(b)}\nonumber\\
&\ds\le
\sum_{b\in\spect(B)\setminus\{0\}}\sum_{\lambda}b\Tr f(A/b)
\rho_{b,\lambda}\Tr \tilde P_{\lambda}^{(b)}
=
\sum_{b\in\spect(B)\setminus\{0\}}\sum_{\lambda}b\Tr f(A/b)\tilde P_{\lambda}^{(b)}
\label{convexity inequality}\\
&\ds=
\sum_{b\in\spect(B)\setminus\{0\}}b\Tr f(A/b)Q_b
\ds=
\sum_{b\in\spect(B)\setminus\{0\}}\sum_{a\in\spect(A)}bf(a/b)\Tr P_aQ_b=\rsr{A}{B}{f},\nonumber
\end{align} 
where $A=\sum_a aP_a$ is the spectral decomposition of $A$, and the inequality in 
\eqref{convexity inequality} follows due to Lemma \ref{lemma:Jensen}. 
This yields the first inequality in \eqref{pinching inequality}.
If $A$ commutes with $B$ then $\E_B(A)=A$ and hence the first inequality in 
\eqref{pinching inequality} holds with equality. Conversely, assume that 
the first inequality in \eqref{pinching inequality} holds with equality; then 
the inequality in \eqref{convexity inequality} has to hold with equality as well.
If $f$ is strictly convex then this implies that for every $b\in\spect(B)\setminus\{0\}$ and $\lambda\in\bR$, there exists an $a(b,\lambda)$ such that 
$\tilde P_{\lambda}^{(b)}\le P_{a(b,\lambda)}$, due to Lemma \ref{lemma:Jensen}. In particular, $\tilde P_{\lambda}^{(b)}$ commutes with $A$, and, since $Q_b=\sum_\lambda \tilde P_{\lambda}^{(b)}$, so does also $Q_b$, which finally implies that $B$ commutes with $A$.

Consider now the stochastic map $\map:\,\A\to\bC,\,\map(X):=\Tr X,\,X\in\A$.
Since $\E_B(A)$ and $B$, as well as $\map(\E_B(A))=\Tr A$ and $\map(B)=\Tr B$, commute, the 
second inequality in \eqref{pinching inequality} follows due to Proposition 
\ref{prop:classical}, which also yields that this inequality holds with equality if and only 
if $\E_B(A)=\map_B^*(\map(\E_B(A))=(\Tr A/\Tr B)B$.

Finally, consider the general case where $\supp A\le\supp B$ does not necessarily hold.
For every $\ep>0$, let 
$B_\ep:=B+\ep I$. Note that $\supp A\le\supp B_\ep$ and 
$\E_{B_\ep}=\E_B$ for every $\ep>0$, and hence by the above, 
$\rsr{A}{B_\ep}{f}\ge
\rsr{\E_B(A)}{B_\ep}{f}\ge (\Tr B_\ep)f\bz\frac{\Tr A}{\Tr B_\ep}\jz$ for every $\ep>0$. Taking the limit 
$\ep\searrow 0$ then yields \eqref{pinching inequality}. 
\end{proof}

The first inequality above was proved for the case $f=f_\alpha,\,\alpha>1$, in Section 3.7 
of \cite{Hayashibook}, and we followed essentially the same proof here. 
It was also proved in Section 3.7 of \cite{Hayashibook} that the monotonicity inequality 
\eqref{Renyi monotonicity2} extends for the values $\alpha\in(2,+\infty)$ if 
$\map(A)$ and $\map(B)$ commute.
We conjecture that this holds in more generality, namely that the monotonicity inequality 
$\rsr{\map(A)}{\map(B)}{f}\le\rsr{A}{B}{f}$ holds for every convex $f$ if $A$ and $B$ or $\map(A)$ and $\map(B)$ commute.
The inequality 
$\rsr{A}{B}{f}\ge (\Tr B)f\bz\frac{\Tr A}{\Tr B}\jz$
was given in Theorem 3 of \cite{Petz} for the case where $A$ and $B$ are invertible density 
operators and $f$ is a non-linear operator convex function. Note that the inequality between the first and the last term in \eqref{pinching inequality} is a non-commutative generalization of the generalized log-sum inequality \eqref{log-sum}.

\begin{cor}\label{cor:Holder}
For any positive semidefinite operators $A,B$ on a finite-dimensional Hilbert space $\hil$, we have
\begin{align}\label{Holder1}
\Tr A^{\alpha} B^{1-\alpha}\le (\Tr A)^{\alpha}(\Tr B)^{1-\alpha},\ds\ds\ds \alpha\in[0,1].
\end{align}
If, moreover, $\supp A\le\supp B$ then 
\begin{align}\label{Holder2}
\Tr A^\alpha B^{1-\alpha}\ge (\Tr A)^\alpha(\Tr B)^{1-\alpha},\ds\ds\ds \alpha\in[1,+\infty).
\end{align}
If $\supp A\le\supp B$ then $\Tr A^\alpha B^{1-\alpha}=(\Tr A)^\alpha(\Tr B)^{1-\alpha}$ 
for some $\alpha\in(0,+\infty)\setminus\{1\}$ if and only if $A$ 
is a constant multiple of $B$.
\end{cor}
\begin{proof}
The assertions are trivial when $A$ or $B$ is equal to zero, and hence we assume that both 
of them are non-zero. The inequality in \eqref{Holder1} 
is obvious when $\alpha=0$ or $\alpha=1$, and 
the inequality in \eqref{Holder2} is obvious when $\alpha=1$.
For $\alpha\in(0,+\infty)\setminus\{1\}$, the inequalities in \eqref{Holder1} and
\eqref{Holder2} follow immediately by applying Proposition \ref{prop:pinching} to the functions
$\tilde f_\alpha(x):=\sgn(\alpha-1)x^\alpha$. Since these functions are strictly convex for every $\alpha\in(0,+\infty)\setminus\{1\}$, if equality holds in \eqref{Holder1} or \eqref{Holder2}, and $\supp A\le\supp B$, then $A$ is a constant multiple of $B$, due to Proposition \ref{prop:pinching}. Conversely, the inequalities  \eqref{Holder1} and \eqref{Holder2} obviously hold with equality if $A$ 
is a constant multiple of $B$.
\end{proof}

Let $\hil$ be a finite-dimensional Hilbert space. For every $A\in\B(\hil)$ and $p\in\bR\setminus\{0\}$, let
\begin{equation*}
\norm{A}_p:=\begin{cases}
0,& A=0,\\
(\Tr |A|^p)^{1/p},& A\ne 0,
\end{cases}
\end{equation*}
where $|A|:=\sqrt{A^*A}$.
For $p\in[1,+\infty)$, this is the well-known $p$-norm.
Note that 
\begin{equation*}
\norm{A^*}_p=\norm{A}_p=\norm{|A|}_p
\end{equation*}
for every $A\in\B(\hil)$ and $p\in\bR\setminus\{0\}$.

Corollary \ref{cor:Holder} yields the following inverse H\"older inequality:

\begin{prop}\label{prop:inverse Holder}
Let $p\in(0,1)$ and $q<0$ be such that $1/p+1/q=1$. Let $A,B\in\B(\hil)$ 
for some finite-dimensional Hilbert space $\hil$, and
assume that $\supp|A|\le\supp|B^*|$. Then
\begin{equation}\label{inverse Holder}
\|AB\|_1\ge\|A\|_p\|B\|_q
\end{equation}
Moreover, the equality case occurs in the above inequality if and
only if $|A|^p$ and $|B^*|^q$ are proportional, i.e., $|A|^p=\alpha|B^*|^q$ for some
$\alpha\ge0$.
\end{prop}
\begin{proof}
The assertion is obvious if $A$ or $B$ is zero, and hence we assume that both of them are
non-zero.
Let $A=U|A|$ and $B^*=V|B^*|$ be the polar decompositions with $U,V$ unitaries. Then
$AB=U|A|\,|B^*|V^*$, and hence $\norm{AB}_1=\norm{|A||B^*|}_{1}$.
Let $\tilde A:=|A|^p,\s\tilde B:=|B^*|^q$ and $\alpha:=1/p$. 
Then $\alpha>1$ and $\supp\tilde A\le\supp\tilde B$ by assumption, and hence 
\begin{align*}
\Tr|A||B^*|&=
\Tr \tilde A^\alpha\tilde B^{1-\alpha}\ge(\Tr\tilde A)^\alpha(\Tr\tilde B)^{1-\alpha}
=
(\Tr|A|^p)^{1/p}(\Tr|B^*|^q)^{1/q}
=
\norm{A}_p\norm{B}_q,
\end{align*}
where the inequality follows due to Corollary \ref{cor:Holder}. 
It is well-known that $|\Tr X|\le\norm{X}_1$ for every 
$X\in\B(\hil)$; indeed, if $X=\sum_i s_i\diad{f_i}{e_i}$ is a singular-value decomposition
then $|\Tr X|=|\sum_i s_i\inner{e_i}{f_i}|\le\sum_i s_i=\Tr|X|=\norm{X}_1$. Hence,
$\Tr|A||B^*|\le\norm{|A||B^*|}_1=\norm{AB}_1$, which completes the proof of the inequality
\eqref{inverse Holder}.
The characterization of the equality case is immediate from Corollary \ref{cor:Holder}.
\end{proof}

\begin{rem}
Our interest in the inverse operator H\"older inequality was motivated by \cite{Hay}. 
The inequality was proved in \cite{Hayashibook2}
for positive semidefinite operators, using the usual H\"older inequality.
An alternative direct proof for the general case and the condition for the equality was
obtained in \cite{Hiaiproof}, based on majorization theory \cite{Bhatia,Hiai}.
\end{rem}

\end{document}